\documentclass[10pt]{article}
\usepackage{amsmath}
\usepackage{amsthm}
\usepackage{amsfonts}
\usepackage{amssymb}
\usepackage{mathtools}
\usepackage{algorithm}
\usepackage{algcompatible}
\usepackage{latexsym} 
\usepackage{subcaption}  
\usepackage{makecell} 
\usepackage{graphicx}
\usepackage[titletoc,title]{appendix}
\usepackage{caption}
\usepackage{float}
\usepackage[table]{xcolor}
\usepackage{diagbox}
\usepackage{changepage}
\usepackage{listings}
\usepackage{dsfont}
\usepackage{natbib}
\usepackage{setspace}
\usepackage{multirow}
\usepackage{scalerel}
\usepackage{mathpazo}
\usepackage{booktabs}
\usepackage[T1]{fontenc}
\usepackage{times}
\usepackage{tikz}
\usetikzlibrary{decorations.pathreplacing}
\usepackage{changes}
\usepackage{algpseudocode}
\usepackage{algorithm}
\usepackage{hyperref}
\hypersetup{
    colorlinks=true,
    linkcolor=blue,
    citecolor=blue,
    urlcolor=blue,
}  
\usepackage[capitalise,compress]{cleveref} 
\usepackage[matrix,tips,graph,curve]{xy}
\usepackage{graphicx}
\newcommand{\indep}{\rotatebox[origin=c]{90}{$\models$}}

\newcommand{\defeq}{\vcentcolon=}
\makeatletter

\setdeletedmarkup{\sout{\textcolor{red}{#1}}}
\setaddedmarkup{\textcolor{red}{#1}}

\makeatother

\makeatletter 
\@addtoreset{equation}{section}
\makeatother

\usepackage{hyperref}
\hypersetup{
            colorlinks={true}, 
            linkcolor={blue}, 
            citecolor=blue,
            pdfencoding=auto, 
            psdextra
            }

\usepackage{enumitem}

\usepackage[capitalise,compress]{cleveref} %
\newlist{assumpenum}{enumerate}{1} 
\setlist[assumpenum]{label=(\arabic*), ref=\theassumption~(\arabic*)}
\crefname{assumpenumi}{Assumption}{assumption}

\newlist{lemmaenum}{enumerate}{1}  
\setlist[lemmaenum]{label=(\arabic*), ref=\thelemma~(\arabic*)}
\crefname{lemmaenumi}{Lemma}{lemma}

\newlist{propenum}{enumerate}{1}  
\setlist[propenum]{label=(\arabic*), ref=\theproposition~(\arabic*)}

\newtheoremstyle{plain}     
  {3pt}{3pt}{}{}{\bfseries}{.}{ }{}    

\theoremstyle{plain}
\newtheorem{assumption}{Assumption}[section]
\newtheorem{theorem}{Theorem}[section]
\newtheorem{lemma}{Lemma}[section]     
\newtheorem{corollary}{Corollary}[section]
\newtheorem{proposition}{Proposition}[section]
\newtheorem{remark}{Remark}[section]

\theoremstyle{definition}

\newtheorem{example}{Example}

\newcommand{\argmax}{\operatorname{argmax}}
\newcommand{\Reg}{\mathrm{Reg}}

\begin{document}
\title{Policy Learning under Unobserved Confounding: A Robust and Efficient Approach\footnote{We thank Yuehao Bai, Toru Kitagawa, Bo Zhang, and participants at 35th POMS Conference 2025, INFORMS Workshop on Data Science 2025, the 2025 Shanghai Econometrics Workshop, and the 2026 AEA Annual Meeting for valuable feedback.} }
\author{
Zequn Jin\thanks{School of Economics, Shanghai University of Finance and Economics. \href{mailto:polaris@163.sufe.edu.cn}{\texttt{jinzequn@mail.shufe.edu.cn}}}
\and
Gaoqian Xu\thanks{Department of Economics, University of Iowa.
\href{mailto:gaoqian-xu@uiowa.edu}{\texttt{gaoqian-xu@uiowa.edu}} }
\and
Xi Zheng\thanks{Michael G. Foster School of Business, University of Washington. \href{mailto:xzheng01@uw.edu}{\texttt{xzheng01@uw.edu}}}
\and
Yahong Zhou\thanks{School of Economics, Shanghai University of Finance and Economics. \href{mailto:yahong.zhou@mail.shufe.edu.cn}{\texttt{yahong.zhou@mail.shufe.edu.cn}}}
}

\date{\today}

\maketitle

\begin{abstract}
This paper develops a robust and efficient method for policy learning from observational data in the presence of unobserved confounding, complementing existing instrumental variable based approaches.  We employ the marginal sensitivity model (MSM) to relax the commonly used yet restrictive unconfoundedness assumption by introducing a sensitivity parameter that captures the extent of selection bias induced by unobserved confounders.  Building on this framework, we consider two distributionally robust welfare criteria, defined as the worst-case welfare and  policy improvement functions,  evaluated over an uncertainty set of counterfactual distributions characterized by the MSM. Closed-form expressions for both welfare  criteria are derived. Leveraging these identification results, we construct doubly robust scores and estimate the robust policies by  maximizing the proposed criteria. Our approach accommodates flexible machine learning methods for estimating nuisance components, even when these converge at moderately slow rates. We establish asymptotic regret bounds for the resulting policies, providing a robust guarantee against the most adversarial confounding scenario. The proposed method is evaluated through extensive simulation studies and  empirical applications to the JTPA study and Head Start program.

   \vspace{12pt}

\textit{JEL codes}: C14, C31, C54

   \vspace{3pt}

    \textit{Keywords}: Partial Identification, Targeted Policy, Treatment Effects, Sensitivity Analysis

\end{abstract}

\newpage
\section{Introduction}



Policy targeting, allocating interventions based on individual characteristics, has become increasingly influential in applied economics and econometrics. For example, a local government may need to decide which workers receive job training, a school district may identify students most in need of educational support, and a regulatory agency has to determine which workplaces warrant safety inspections.

Policymakers often rely on observational data to guide targeting decisions. However, unobserved confounding, unobserved factors that simultaneously affect both treatment and potential outcomes, can severely bias causal inference from such data. Ignoring these confounders can result in misleading estimates and suboptimal policy decisions.  For example, \cite{johnson2023improving} note that safety inspections tend to target workplaces with a history of accidents or informal complaints, events that may reflect unobserved traits, such as poor safety culture or weak management, which also affect injury risk.  Given the potentially large benefits of well-targeted inspections,\footnote{Workplace injuries and illnesses place a significant economic burden on the United States, with annual social costs estimated at \$250 billion \citep{leigh2011economic}.}designing targeting without accounting for unobserved confounding may misidentify high-risk workplaces, leading to the inefficient use of limited regulatory resources. Similar challenges arise in education, where observational data guide targeting decisions for specific academic programs. However, actual enrollment often hinges on parental choice, which is based on factors inaccessible to policymakers, such as a child’s latent abilities, behavioral traits, or home environment. This self-selection process based on unobserved characteristics complicates the design of effective targeting policies.

Many existing policy learning (or targeting) methods using observational data assume no unobserved confounding, that is, treatment assignment is random conditional on the observed covariates; see \citep{kitagawa2018a,athey2021policy}. This assumption, commonly referred to as unconfoundedness, is often unreasonable, as historical policies or decisions may have been based on additional, unobserved information.  When this assumption breaks down, researchers often turn to instrumental variable (IV) methods, which enable causal inference and policy targeting; see \citep{cui2021semiparametric,qiu2021optimal,sasaki2024welfare}.  Yet in practice, identifying a valid and credible instrument is often challenging. Moreover, even with a valid binary instrument, heterogeneous treatment effects are typically only partially identified, thereby limiting the ability to effectively target interventions.

As a supplement to IV methods, sensitivity analysis is used to assess the impact of unobserved confounding on policy effects, as originally proposed by \cite{rosenbaum1983assessing}. Building on this idea, \cite{kallus2021minimax} incorporate sensitivity analysis into policy learning. Leveraging the marginal sensitivity model (MSM) of \cite{tan2006distributional}, they relax the unconfoundedness assumption and formulate a robust optimization framework to learn policies that remain effective under unobserved confounding. Specifically, building on the framework proposed by \cite{zhao2019sensitivity}, they construct the objective function as the largest  {inverse propensity weighting (IPW)} estimator of expected welfare, where the putative propensity scores are restricted by the MSM. While the method offers robustness to unobserved confounding, it is not without limitations. First, \cite{kallus2021minimax} derive a  {loose} upper bound on expected welfare improvement, which leads to overly conservative policies. Second, the method requires propensity scores to be estimated at the parametric rate, which precludes the use of many popular  {machine learning (ML)} methods, such as random forests and LASSO. Third, the policy optimization algorithm is tailored to differentiable parametric policy classes and relies on a computationally intensive iterative routine.

\subsection{Main Contributions}

This work proposes a method for policy learning that is robust to unobserved confounding, serving as a complement to existing IV-based approaches.  Building on the MSM framework, we construct a distributionally robust welfare criterion that captures the worst-case policy value over a plausible set of counterfactual outcome distributions.  Based on this criterion, we develop a two-stage policy learning procedure: the first stage estimates nuisance components using flexible  nonparametric or modern ML  methods; the second stage optimizes the estimated criterion over a constrained policy class. The resulting procedure is both computationally efficient and straightforward to implement.

In \cref{section: Identification}, building on \cite{dorn2023sharp} and \cite{dorn2024doubly}, we  derive closed-form expressions for the worst-case (i.e., sharp lower bounds of) average welfare and welfare improvement\footnote{These are commonly used objective functions in policy learning, which are formally defined in \cref{section: Policy Learning beyond Unconfoundedness}.} compatible with the MSM.  {Although our main analysis focuses on binary compulsory assignment, Appendix~\ref{appendix: self-selection} examines selection-driven policy targeting under the MSM, following \cite{ida2022choosing}, with self-selection as a policy option. This extension is relevant when the MSM-implied treatment-effect bounds are insufficiently informative for compulsory assignment, in which case preserving self-selection may help exploit individuals' private information for policy targeting.} Given the identification of these two robust criteria/objectives  using moment conditions, we can
construct the doubly robust estimation of them, even when the nuisance components are estimated at moderately slow rate, i.e.,  { $o(n^{-1/4})$ in the $L^2$-norm}.  

\Cref{section: regret bounds} presents algorithms for learning optimal confounding-robust policies, where the policy class may be constrained by practical considerations such as implementability, cost, and interpretability. Our method follows the mainstream two-stage paradigm of \citet{athey2021policy, zhou2023offline}. Unlike \citet{kallus2021minimax}, our second-stage policy optimization supports standard  {ML} methods and readily available software packages, such as logistic regression, policy trees, and neural networks, to learn the optimal policy by maximizing the estimated criterion. For example, given the doubly robust scores, our method can be implemented using the widely adopted \texttt{R}-package \texttt{policytree} \citep{sverdrup2020policytree}. Moreover, \Cref{section: regret bounds} provides asymptotic upper regret bounds as performance guarantees for our estimated confounding-robust policies.

We  validate our method via simulations in \cref{section: simulation}, comparing it to \cite{kallus2021minimax} and the empirical welfare maximization (EWM) of \cite{athey2021policy}. We then demonstrate its practical relevance in \cref{section: empirical} by applying it to two classic empirical settings where unobserved confounding is a primary concern, focusing on the contrast with the EWM.

Our applications suggest that ignoring such confounding may lead to misguided policy recommendations. In the National Job Training Partnership Act (JTPA) study, our method refines the naive EWM policy which treats nearly all participants. As concerns about unobserved confounding grow, our method shifts to identifying more selective policies that target subgroups with higher education but lower prior earnings.  {In our Head Start application, the EWM policy assigns no children to Head Start.} In contrast, our robust method prioritizes disadvantaged children with lower family income and maternal education, aligning with the program's mission. These findings demonstrate that accounting for unobserved confounding can fundamentally alter policy targeting, highlighting the importance of robust methods in practice.


\subsection{Related Literature}

Our work contributes to the growing literature on policy learning. Many policy learning works focus on mean-optimal policies under unconfoundedness \citep{kitagawa2018a,athey2021policy,zhou2023offline}. Alternative objectives under the same assumption have also been explored, including quantile-optimal policies \citep{wang2018quantile,leqi2021median} as well as those motivated by fairness \citep{kitagawa2021equality,fang2023fairness,viviano2023fair,fan2025policy}.

To address unobserved confounding, several studies have proposed IV-based approaches for policy targeting.  To avoid partial identification of the welfare function, \cite{sasaki2024welfare} assume the availability of a continuous instrument with sufficiently large support and identify the average welfare via the marginal treatment effect (MTE).  A similar approach has been applied to the design of optimal encouragement rules under endogeneity; see \citep{chen2022personalized,liu2022policy}. In the statistics literature, other works investigate optimal treatment rules or encouragement interventions using binary instruments; see \citep{cui2021semiparametric,qiu2021optimal}. However, these methods rely on quite stringent assumptions for a valid IV to point identify the conditional average treatment effect (CATE).

Binary instruments, though widely used, typically allow identification only of the local average treatment effect (LATE) for the subpopulation of compliers; see \citep{imbens1994identification,angrist1996identification,athey2017state}. This presents major challenges for policy learning: the resulting policy intervention is effective only for compliers, a subpopulation that cannot be identified in advance in a new sample, thus limiting external validity. Moreover, compliance status is instrument-dependent, leading to instability and lack of generalizability across different IV designs. To address this issue, \cite{pu2021estimating} and \cite{d2021orthogonal}, propose methods for binary IV settings that partially identify heterogeneous treatment effects, learning robust policies by optimizing policy criteria constructed from IV-identified bounds on the CATE.

 {
More broadly, our work contributes to the growing literature on policy learning and treatment choice under partial identification, where the relevant welfare function is only set-identified. To address this ambiguity, the literature commonly relies on the minimax principle \citep{manski2000identification,manski2004statistical,manski2025identification,stoye2012minimax,yata2025optimal,olea2023decision}, while \cite{Christensen2026Res} proposes a hybrid approach that combines minimax criteria with quasi-Bayesian methods.
Partial identification may also arise from distributional ambiguity, in the sense that the target population  may differ from the observed sample. In this setting, some works adopt distributionally robust optimization to learn policies that remain valid across new environments \citep{kido2022distributionally,si2023distributionally,qi2023robustness,lei2023policy,zhang2024minimax,adjaho2025externally}. These works mainly maximize worst-case expected outcomes over ambiguity sets defined by  Wasserstein distance or Kullback-Leibler divergence. 
Moreover, learning externally valid policies may require extrapolating causal parameters that cannot be point-identified from the available sample. To address this, recent studies \citep{khan2023off,zhang2024safe,ben2025safe,higbee2025policy} impose functional or shape restrictions to partially identify treatment effects and derive robust policies. Finally, we focus on learning robust policies when  causal parameters are  partially identified  due to unmeasured confounding, viewed through the lens of sensitivity analysis \citep{tan2006distributional, masten2018identification, zhao2019sensitivity,kallus2021minimax, dorn2023sharp,dorn2024doubly}.
}

\section{Problem Statement and Preliminaries}\label{section: sec_1}
We consider observational data consisting of random samples $\{Z_i\}_{i=1}^n=\{X_i, Y_i, A_i\}_{i=1}^n$, where $X_i\in \mathcal{X} \subseteq \mathbb{R}^d$ represents the observed covariates, $A_i \in \{0,1\}$ denotes a binary intervention/treatment, and $Y_i \in \mathbb{R}$ is the real-valued observed outcome. Within the Neyman-Rubin potential outcomes framework, let $Y_i(0)$ and $Y_i(1)$ denote the potential outcomes under control ($A_i = 0$) and treatment ($A_i = 1)$, respectively. The observed outcome satisfies $Y_i = Y_i(A_i)$. Throughout this work, we interpret $Y_i$ as a measure of welfare or utility, where higher values correspond to more desirable outcomes.

Our objective is to use the observational data $\{ Z_i \}_{i=1}^n$ to guide personalized policy interventions in settings where unobserved confounding may be present.

 {
\begin{example}[Job Training]\label{example: JTPA_Training}
Let $A_i$ denote whether individual $i$ participates in a job training program. Let $Y_i(1)$ and $Y_i(0)$ denote individual $i$'s post-program earnings under participation and non-participation, respectively. To improve labor market outcomes, the policymaker decides whether individual $i$ should receive the job training based on observed covariates $X_i$, such as years of education and pre-program earnings.
\end{example}
}

\begin{example}[Head Start Enrollment]\label{example: Head_Start}
Let \( A_i \) indicate whether child \( i \) is enrolled in the Head Start program. Let \( Y_i(1) \) and \( Y_i(0) \) denote academic outcomes under enrollment and non-enrollment (e.g., test scores or school readiness). To improve early childhood development, the policymaker allocates slots based on characteristics \( X_i \), including household income, parental education, and number of siblings.
\end{example}

\subsection{The Marginal Sensitivity Model}

Our analysis builds on the marginal sensitivity model (MSM) introduced by \cite{tan2006distributional} to control the selection bias caused by unobserved confounders. This model relaxes the unconfoundedness assumption by allowing for unobserved confounders $U \in \mathbb{R}^k$, where $k\in \mathbb{N}^+$ is unknown. Here, we assume that 
$U$ is a vector but have no prior knowledge of its dimension or other characteristics. We denote $P_{o}$ as the true distribution of $O \equiv \left(X,Y(1), Y(0), A,U \right) $.  The true propensity score is given by $e_o(x, u) =\mathbb{P}_{P_{o}}[A = 1 | X = x, U = u]$, which accounts for both observed covariates and unobserved confounders. In contrast, the nominal propensity score, defined as $e(x) = \mathbb{P}_{P_{o}}[A = 1 | X = x]$. In the MSM, unobserved confounders are assumed to have a bounded influence on the odds of treatment assignment,  restricting the extent of selection bias introduced by unobserved confounders.

We now present the formal specification of the MSM, a widely used framework for sensitivity analysis in the causal inference literature (see, e.g., \cite{zhao2019sensitivity, dorn2023sharp, oprescu2023b, dorn2024doubly}).

\begin{assumption}[Marginal Sensitivity Model] \label{assumption: Marginal sensitivity model}
Suppose there exists a vector of unobserved confounders $U \in \mathbb{R}^k$ such that  
\[
\left( Y(1), Y(0) \right) \indep  \ A \mid \left(X, U \right).
\]
The distribution of $O \equiv \left(X, Y(1),Y(0),A,U \right) $ satisfies the selection bias condition with $1 \leq \Lambda< \infty$ if the following inequality holds $P_o$-almost surely,
\begin{equation}\label{equation: marginal sensitivity model}
\frac{1}{\Lambda} \leq \frac{e_o(x, u) / \left( 1- e_o(x, u)  \right)}{e(x ) / \left( 1- e(x)  \right) }  \leq \Lambda.
\end{equation}
\end{assumption}

\begin{remark}
  {In \cref{example: JTPA_Training}, the latent variable $U$ may capture unobserved factors, such as motivation, expected gains from training, or participation constraints, that affect both training take-up and potential earnings.} Similarly, in \cref{example: Head_Start}, $U$ may reflect unobserved characteristics of the child and family, including aspects of socioeconomic background, parental preferences, and investments in early human capital.
\end{remark}

In practice, the sensitivity parameter $\Lambda$ in \cref{assumption: Marginal sensitivity model} is selected by policymakers based on their prior beliefs regarding the extent of unobserved confounding and the resulting selection bias. Implicitly, it is assumed that policymakers choose $\Lambda$ such that \cref{assumption: Marginal sensitivity model} is satisfied. The special case of $\Lambda=1$ corresponds to the unconfoundedness, also known as the selection-on-observables assumption. Higher values of $\Lambda$ impose fewer restrictions on the degree of unobserved confounding, allowing for greater potential selection bias. 

\begin{remark}
 {The sensitivity parameter, $ \Lambda $, measures the allowed deviation from the unconfoundedness benchmark. Following \cite{hsu2013calibrating}, we calibrate $ \Lambda $ assuming the effect of an unobserved confounder on the treatment-assignment odds is comparable in magnitude to that of observed covariates.  In practice, we quantify this magnitude by examining how the estimated treatment-assignment odds change when each observed covariate is omitted in turn.}

 {When credible bounds from an alternative identification strategy are available, they can also serve as a reference for calibrating $\Lambda$. Following \cite{masten2018identification}, we ask how much the baseline unconfoundedness assumption must be relaxed before the MSM-implied bounds become conservative relative to these reference bounds. For the JTPA study in \cref{section: JTPA}, we use the IV bounds induced by randomized eligibility assignment as this benchmark. Specifically, we choose the smallest value of $\Lambda$ such that the MSM-implied lower bound do not exceed the IV lower bound.}
\end{remark}


\subsection{Policy Learning under the MSM}\label{section: Policy Learning beyond Unconfoundedness}

In this subsection, we provide a brief review of policy learning  under the assumption of unconfoundedness ($\Lambda = 1$).  We then explore how to learn a robust policy when this assumption is violated, within the framework of the marginal sensitivity model. Throughout the rest of this work, we fix the sensitivity parameter $\Lambda \geq 1$.

A possibly randomized policy $\pi$ maps covariates $x \in \mathcal{X}$ to treatment assignment probabilities, with $\pi(x)$ denoting the probability of receiving treatment $A = 1$. The policymaker can pre-specify a policy class $\Pi$, defined as a collection of Borel measurable functions mapping from $\mathcal{X}$ to $[0,1]$.  This class typically incorporates application-specific constraints, including budgetary limitations, structural assumptions, and fairness requirements; see \cref{section: policy_class} for concrete examples.

We begin by recalling the framework of policy learning under the unconfoundedness assumption. Given a predefined policy class $\Pi$, the mean-optimal policy learning aims to identify a policy that maximizes the expected outcome:
\begin{equation}\label{equation: standard policy learning}
\pi^\star \in \underset{\pi \in \Pi}{ \argmax}  \ \mathbb{E} \left[ Y\left(\pi \left(X \right) \right) \right],
\end{equation}
where  $Y\left(\pi \left(X \right)\right) = \pi(X)Y(1) + (1-\pi(X))Y(0)$. For simplicity, let  {$W_1(\pi) = \mathbb{E} \left[ Y\left(\pi(X) \right) \right]$} denote the expected welfare function. Moreover, the regret of deploying a policy $\pi \in \Pi$ is defined as
\begin{equation}\label{equation: mean_regret} \mathrm{Reg}(\pi) = \sup_{\pi^\prime \in \Pi} \mathbb{E}[Y(\pi^\prime(X))] - \mathbb{E}[Y(\pi(X))]. 
\end{equation}
It is evident that a policy maximizing $W_1(\pi)$ simultaneously minimizes the regret given in \cref{equation: mean_regret}.

A central objective in policy learning is policy evaluation, that is, identifying and  estimating the criterion function.  Once an estimate of $W_1(\cdot)$ is available, the optimal policy can be estimated by maximizing the estimated welfare over the policy class. Unconfoundedness (i.e., $\Lambda = 1$) is a  widely adopted assumption for identifying the expected welfare function; see \cite{kitagawa2018a,athey2021policy, zhou2023offline}. Under this assumption, various methods such as inverse probability weighting (IPW) can be employed to identify and estimate $W_1(\pi)$. In particular, the expected welfare function can be expressed as
\begin{align}
W_1(\pi)  &= \mathbb{E}[Y(0)] + \mathbb{E}[ \tau(X) \pi(X) ]  \label{eq:welfare-decomposition} \\
 &= \mathbb{E}\left[ \frac{Y A \pi(X)}{e(X)} + \frac{Y (1-A) (1-\pi(X)) }{1-e(X)} \right] \label{eq:welfare-identification}, 
\end{align}
where $e(x)$ is the nominal propensity score, and $\tau(x) = \mathbb{E}[Y(1) - Y(0) \mid X = x]$ is the conditional average treatment effect (CATE). When $\Lambda = 1$, the outcome regression, nominal propensity score, and CATE are all identifiable.

When unconfoundedness fails (i.e., $\Lambda > 1$), the expected welfare $W_1(\pi)$ becomes unidentifiable, and policies based on \cref{eq:welfare-decomposition} or \cref{eq:welfare-identification} lack reliable guarantees. As a result, $W_1(\pi)$ is not a suitable criterion for policy learning in the presence of unobserved confounding.

We adopt a distributionally robust optimization (DRO) framework for policy learning, replacing $W_1(\pi)$ with its worst-case counterpart over distributions consistent with MSM specified in \cref{assumption: Marginal sensitivity model}. Within this framework, we construct two welfare criteria that are both robust to unobserved confounding and identifiable, and derive the policy by optimizing one of them.

To implement the DRO, we first characterize the distributional uncertainty set over the counterfactual distribution of $\left(Y(1), Y(0), X, U, A\right)$  subject to the constraints imposed by \cref{assumption: Marginal sensitivity model}.  Formally, let $\mathcal{P}(\Lambda)$ denote the set of all probability distributions $Q$ on $\mathbb{R}^2  \times \mathcal{X}\times \{0,1\} \times \mathbb{R}^k$ satisfying the following conditions: 
\begin{enumerate}
    \item[(1)]   If $\left(X, Y(0), Y(1), A, U \right)\sim Q$, then $\left(Y(0), Y(1) \right) \indep A \mid (X,U)$ under $Q$;
    \item[(2)] If $Y = A Y(1) + (1-A)Y(0)$, then the distribution of $(X,Y,A)$ under $Q$ is identical to the observed-data distribution $P$;
    \item[(3)] The odds ratio between the true propensity score and the nominal propensity score lies in $\left[1/\Lambda, \Lambda \right]$, i.e., 
    \[
  \frac{1}{\Lambda} \leq  \frac{ \mathbb{P}_{ {Q}}(A =1 |X, U)/  \mathbb{P}_{ {Q}}(A =0 |X, U) }{\mathbb{P}_{ {Q}}(A =1 |X)/  \mathbb{P}_{ {Q}}(A =0 |X)} \leq \Lambda.
    \]
\end{enumerate}

We now present two complementary policy learning methods that are robust to unobserved confounding under the MSM framework. Our first approach is the max-min expected welfare method. Specifically, the {\it worst-case welfare function} $W_{\Lambda}:\Pi\rightarrow\mathbb{R}$ is defined as
\[
W_{\Lambda}(\pi)=\inf_{ Q \in \mathcal{P}(\Lambda) } \mathbb{E}_Q \left[ 
 Y(\pi (X) )\right].
\]
The corresponding max-min welfare (MMW) policy is given by 
\begin{equation}\label{equation: max-min welfare PL}
\pi_{W, \Lambda} \equiv \pi_{W}(\cdot; \Lambda)  \in \underset{\pi \in \Pi}{\argmax} \ W_{\Lambda}(\pi).    
\end{equation}
Under the worst-case welfare
function $W_\Lambda$, we define the confounding-robust welfare regret (CRW-regret) of a policy  $\pi \in \Pi$, relative to the best possible policy in $\Pi$,  as 
\begin{equation}\label{robust_welfare_regret}
\mathrm{Reg}_W(\pi) = \sup_{\pi^\prime \in \Pi}  W_{\Lambda}(\pi^\prime) - W_\Lambda(\pi).
\end{equation}

The second approach builds on the conditional average treatment effect (CATE). When $\Lambda = 1$, unconfoundedness holds and, the policy learning problem reduces to maximizing the policy improvement function $\mathbb{E}[\tau(X)\pi(X)]$.
To handle the more general case with $\Lambda \geq 1$, we extend this objective to the \textit{worst-case policy improvement function} $\Delta_{\Lambda}: \Pi \rightarrow \mathbb{R}$, defined as
\begin{equation}\label{eq:MMI_function}
\Delta_{\Lambda}(\pi)=\inf_{ Q \in \mathcal{P}(\Lambda) } \mathbb{E}_Q \left[ 
\pi(X) ( Y(1) - Y(0)  ) \right].
\end{equation}
The corresponding max-min improvement (MMI)  policy is given by
 \begin{equation}\label{equation: CATE-based robust PL}
\pi_{\Delta, \Lambda} \equiv \pi_{\Delta}(\cdot; \Lambda) \in \underset{ {\pi \in \Pi} }{\argmax} \ \Delta_{\Lambda}(\pi).
\end{equation}
Under the criterion $\Delta_\Lambda$, we define the confounding-robust policy improvement regret (CRI-regret) of a policy  $\pi \in \Pi$, relative to the best possible policy in $\Pi$,  as 
\begin{equation}\label{robust_improvement_regret}
\mathrm{Reg}_{{\Delta}}(\pi) = \sup_{\pi^\prime \in \Pi}  \Delta_{\Lambda}(\pi^\prime) - \Delta_\Lambda(\pi).
\end{equation}

The quantity $\Delta_{\Lambda}(\pi) = \inf_{Q \in \mathcal{P}(\Lambda)} \mathbb{E}_Q \left[Y(\pi(X)) - Y(0)\right]$ represents the worst-case welfare gain of policy $\pi$ relative to the baseline policy $\pi_0(x) = 0$.  {In fact}, this approach can be generalized to the policy improvement of $\pi$ against any given baseline policy $\pi_0$, with the optimal policy obtained by:
\begin{equation}\label{eq: baseline MMI}
    \max_{\pi \in \Pi}\inf_{Q\in \mathcal{P}(\Lambda)} \mathbb{E}_Q\left[ Y(\pi(X)) - Y\left(\pi_0(X) \right) \right].
\end{equation}
 {We refer to any solution to \eqref{eq: baseline MMI} as a baseline-relative MMI policy.  Since $Y(\pi(X)) - Y\left(\pi_0(X)\right) = \left(\pi(X)-\pi_0(X)\right)(Y(1)-Y(0))$, both \eqref{eq: baseline MMI} and the resulting optimal policy can be  {derived} using the similar identification and estimation strategy as in \eqref{equation: CATE-based robust PL}. This objective is closely related to that in \cite{kallus2021minimax}. Their method proceeds by constructing a conservative lower bound for the inner minimization in \eqref{eq: baseline MMI} based on an uncertainty set for putative propensity weights. See \eqref{thm:MMI-baseline} for the sharp characterization of \eqref{eq: baseline MMI}, and \cref{remark:A1,remark:A2} provides a detailed comparison with \cite{kallus2021minimax}.
}


\begin{remark}
When $\Lambda = 1$, the MMW policy in  {\cref{equation: max-min welfare PL}} coincides with the MMI policy in  {\cref{equation: CATE-based robust PL}}. However, in the presence of unobserved confounding (i.e. $\Lambda > 1$), the two policies may differ. A detailed discussion of the differences between these policies is deferred to \cref{section: Identification}. 
\end{remark}



\section{Identification}\label{section: Identification}
Within our framework, learning a confounding-robust policy requires identifying and estimating either the worst-case welfare function
$W_{\Lambda}(\pi)$ or the worst-case policy
improvement function $\Delta_{\Lambda}(\pi)$.   In this section, we formally characterize these criterion functions by leveraging partial identification results for conditional means and CATE under the MSM. We then develop doubly robust/orthogonal moment functions for identifying $W_{\Lambda}(\pi)$ and  $\Delta_{\Lambda}(\pi)$,  which enable the construction of efficient estimators.

For notational simplicity, we define two quantile functions $q^{\pm}_{\Lambda}(x,a)$ as
\[
\begin{aligned}
q^{+}_{\Lambda}(x,a) &= \inf \left\{ q: F(q|x,a)\geq \frac{\Lambda}{1+\Lambda} \right\},\\
q^{-}_{\Lambda}(x,a)& = \inf \left\{ q: F(q|x,a)\geq \frac{1}{1+\Lambda} \right\},
\end{aligned}
\]
where $F(\cdot|x,a)$ denotes the CDF of $Y$ given $X=x$ and $A=a$. Moreover, let \( e_a(x) = \mathbb{P}(A = a | X = x) \) for $a \in \{0,1\}$, so that \( e_1(x) = e(x) \) and \( e_0(x) = 1 - e(x) \).



\subsection{Identification of Robust Criteria}\label{section: Identification of Robust Criteria}

In this subsection, we identify and characterize the robust criterion functions $W_\Lambda(\pi)$ and $\Delta_\Lambda(\pi)$. Additionally, we derive the first-best policies optimized under these robust criteria. Let $\Pi_o$ denote the set of all measurable policies:
\[
\Pi_o = \left\{ \pi: \mathcal{X} \rightarrow [0,1]  \text{ is Borel measurable}  \right\}.
\]
A first-best policy refers to a policy that maximizes the criterion function over the unrestricted policy class  $\Pi_o$.

We begin by identifying the worst-case welfare $W_\Lambda(\pi)$. When $\Lambda >1$, the true conditional mean function $\mu_{o}(x,a) = \mathds{E}_{P_o}[ Y(a) | X =x ]$ is no longer point identified due to unobserved confounding. Instead, we characterize the sharp bounds for  $\mu_{o}(x,a)$  using \cref{proposition: pi_mean} in \cref{section: identification_Conditional_Mean_CATE}, which also implies a sharp lower bound for the true welfare function $\mathbb{E}_{P_o}[Y(\pi)]$. Specifically, the proposition provides closed-form expressions for the upper and lower bounds, denoted $\mu^{\pm}_{\Lambda}(x,a)$, which are given by
\[
  \mu^{\pm}_{\Lambda}(x,a) = \mathbb{E}\left[ Y \mathds{1}\{ A = a \} \left[1+\frac{1- e_a(X)}{e_a(X)}\Lambda^{\pm\text{sgn}\left(Y-q^{\pm}_{\Lambda}(X,a)\right)}\right] \Big|X=x \right], \\
\]
where $\text{sgn}(t)=1$ if $t\geq0$ and $-1$ otherwise.

\begin{theorem} \label{theorem: max-min welfare assignment rule}
Under \cref{assumption: Marginal sensitivity model},  for any policy $\pi \in \Pi$, we have
\[
\begin{aligned}
W_\Lambda(\pi)=\inf_{ Q \in \mathcal{P}(\Lambda) } \mathbb{E}_Q \left[ 
 Y(\pi (X) )\right]  & = \mathbb{E}\left[ \mu_\Lambda^-(X, 1) \pi (X) +    \mu_\Lambda^-(X, 0) \left(1-\pi (X) \right)  \right]. \\
\end{aligned}
\]
Moreover, a first-best MMW policy that solves \cref{equation: max-min welfare PL} with $\Pi = \Pi_o$ is given by 
\[
\pi_{W,\Lambda}^\star(x) = \mathds{1} \left\{ \mu_\Lambda^-(x, 1) -     \mu_\Lambda^-(x, 0) > 0  \right\}.
\]
\end{theorem}

The identification of $\Delta_\Lambda(\pi)$ is established in \cref{theorem: CATE-based robust welfare}, which follows as a corollary of \cref{proposition: pi_CATE}. To formalize the result, define $\tau_\Lambda^-(x) = \mu^-_\Lambda(x,1) - \mu^+_\Lambda(x,0)$.

\begin{theorem}\label{theorem: CATE-based robust welfare}
Under \cref{assumption: Marginal sensitivity model},  for any policy $\pi \in \Pi$, we have
\[
\Delta_{\Lambda} (\pi) = \inf_{Q \in \mathcal{P}(\Lambda)} \mathbb{E}_{Q} \left[ (Y(1)-Y(0))  \pi (X) \right]  =   \mathbb{E}\left[  \tau^-_{\Lambda}(X) \pi (X)  \right] .
\]
Moreover, a first-best MMI policy that solves  \cref{equation: CATE-based robust PL} with $\Pi = \Pi_o$ is given by 
\[
\pi_{\Delta,\Lambda}^\star(x) = \mathds{1}\{ \tau^-_{\Lambda}(x) > 0 \} = \mathds{1} \left \{ \mu_\Lambda^-(x, 1) -     \mu_\Lambda^+(x, 0) > 0  \right\}.
\]
\end{theorem}

\Cref{theorem: max-min welfare assignment rule} characterizes the first-best MMW policy $\pi_{W,\Lambda}^\star(x)$, which assigns treatment by comparing the lower bounds $\mu_{\Lambda}^-(x,1)$ and $\mu_{\Lambda}^-(x,0)$. In contrast, \cref{theorem: CATE-based robust welfare} shows that the first-best MMI policy $\pi_{\Delta,\Lambda}^\star(x)$ assigns treatment by comparing $\mu_{\Lambda}^-(x,1)$ with the upper bound $\mu_{\Lambda}^+(x,0)$. As a result, the MMI policy is more conservative: it treats only when the worst-case treated outcome exceeds the best-case control outcome.  {Although the unrestricted policy class $\Pi_o$ permits randomized policies, the first-best MMW and MMI policies are deterministic and unique up to tie-breaking on indifference sets, because the corresponding robust criteria are affine functionals of $\pi$.\footnote{The first-best MMW policy is uniquely determined outside the indifference set $\{x:\mu_\Lambda^-(x,1)=\mu_\Lambda^-(x,0) \}$, whereas any  $\pi(x)\in[0,1]$ is optimal within this set. Consequently, if $\mathbb{P} \left[ \mu_\Lambda^-(X,1)=\mu_\Lambda^-(X,0) \right]=0$, the first-best MMW policy is unique up to null sets. The same argument applies to the MMI criterion.}}





We  {compare} our first-best MMW and  MMI policies with the Bayes decision rule (referred to as a policy in our framework) of \cite{pu2021estimating}, hereafter referred to as PZ-policy. 

\begin{remark}
The first-best MMW policy can also be interpreted through a CATE-based approach, as it assigns treatment by comparing the worst-case potential outcomes. Since this rule is confined to a comparison of identified lower bounds, it is uninformative about the magnitude of the true treatment effect. For this reason, we turn our focus to the MMI policy, and a comparison between the first-best MMI policy and  PZ-policy is more instructive.

\begin{figure}[htbp]
\centering

\begin{subfigure}[t]{0.45\textwidth}
    \centering
    \begin{tikzpicture}[>=stealth,scale=1.0]
        \coordinate (L) at (-2,0);
        \coordinate (mid) at (0.5,0);
        \coordinate (U) at (2,0);

        \draw[->,thick] (-2.5,0) -- (2.5,0);

        \draw (L) -- ++(0,-0.1) node[below] {$\tau_{\Lambda}^{-}$};
        \draw (mid) -- ++(0,-0.1) node[below] {$\frac{\tau_{\Lambda}^{-}+\tau_{\Lambda}^{+}}{2}$};
        \draw (U) -- ++(0,-0.1) node[below] {$\tau_{\Lambda}^{+}$};

        \coordinate (zeroL) at (-0.8,0);
        \draw (zeroL) -- ++(0,-0.1) node[below] {$0$};

        \draw[->,thick,dashed,blue] (L) ++(0,0.8) -- (L);
        \node[above=22pt] at (L) {\textcolor{blue}{MMI ($\pi =0$)}};

        \draw[->,thick,dashed,red] (mid) ++(0,0.8) -- (mid);
        \node[above=22pt] at (mid) { \textcolor{red}{PZ ($\pi =1$)}};
    \end{tikzpicture}
    \caption{When $\tau_\Lambda^- < 0 < \frac{1}{2}(\tau_\Lambda^+ + \tau_\Lambda^-)$, the PZ policy recommends treatment ($\pi = 1$), while the MMI policy recommends no treatment ($\pi = 0$).}
    \label{fig:left0}
\end{subfigure}
\hfill
\begin{subfigure}[t]{0.45\textwidth}
    \centering
    \begin{tikzpicture}[>=stealth,scale=1.0]
        \coordinate (L) at (-2,0);
        \coordinate (mid) at (0.3,0);
        \coordinate (U) at (2,0);

        \draw[->,thick] (-2.5,0) -- (2.5,0);

        \draw (L) -- ++(0,-0.1) node[below] {$\tau_{\Lambda}^{-}$};
        \draw (mid) -- ++(0,-0.1) node[below] {$\frac{\tau_{\Lambda}^{-}+\tau_{\Lambda}^{+}}{2}$};
        \draw (U) -- ++(0,-0.1) node[below] {$\tau_{\Lambda}^{+}$};

        \coordinate (zeroR) at (1.2,0);
        \draw (zeroR) -- ++(0,-0.1) node[below] {$0$};

             \draw[->,thick,dashed,blue] (L) ++(0,0.8) -- (L);
        \node[above=22pt] at (L) {\textcolor{blue}{MMI ($\pi =0$)}};

        \draw[->,thick,dashed,red] (mid) ++(0,0.8) -- (mid);
        \node[above=22pt] at (mid) {\textcolor{red}{PZ ($\pi =0$)}};
    \end{tikzpicture}
    \caption{When $\frac{1}{2}(\tau_\Lambda^+ +\tau_\Lambda^-) < 0$, both the PZ and MMI policies  {recommend} no treatment ($\pi=0$).}
    \label{fig:right0}
\end{subfigure}

\caption{Comparison of the MMI and PZ decision rules when applied to a given ambiguous CATE interval $[\tau_\Lambda^-, \tau_\Lambda^+]$.}
\label{fig:two_cases}
\end{figure}

While both approaches assign treatment based on CATE bounds, they differ in two aspects. First is the construction of the bounds: MMI policy uses the MSM framework, whereas the PZ-policy employs an IV approach. Second, and more importantly, the two policies differ in how they apply decision rules to these bounds. The first-best MMI policy is more conservative and follows a strict criterion: treatment is assigned only when the lower bound of the CATE exceeds zero. In contrast, the PZ-policy adopts a three-case decision rule: it treats if the lower bound is positive, withholds if the upper bound is negative,  {and} when the bounds straddle zero, treats if the upper bound has a larger absolute value than the lower bound. This contrast highlights why our MMI policy is more conservative: it avoids intervention in the presence of ambiguity, whereas the PZ-policy allows for treatment under uncertainty (see \cref{fig:two_cases}).
\end{remark}

\subsection{Doubly/Locally Robust Scores} \label{section: Identification and Estimation of the Worst-case Welfare Function}

To effectively learn the confounding-robust policy, it is essential to estimate the entire worst-case welfare function with minimal estimation error. To facilitate the use of modern ML methods, we derive doubly robust scores for the worst-case criteria,  $W_\Lambda (\pi)$ and $\Delta_\Lambda (\pi)$, such that the estimation of nuisance parameters has no first-order influence on the resulting policy evaluation. For notational convenience, we may write $\pi(a|x) = a \pi(x) + (1-a) \left(1-\pi(x)\right)$, which compactly represents the probability that treatment $a \in \{0,1\}$ is assigned under policy $\pi$ given covariate $x$.  {Moreover, we assume that the conditional distribution of $Y_i|X_i, A_i$ admits a bounded density, as in \cite{dorn2023sharp} and \cite{dorn2024doubly}.}

\begin{assumption}\label{assum: density}
  {   For each $(x,a) \in \mathcal{X} \times \{0,1\}$, the conditional distribution $F(y|x,a)$ is continuous with a uniformly bounded density $f(y|x,a)$ that is
positive on the interior of its support.}
\end{assumption}

\subsubsection*{Doubly Robust Score for $W_\Lambda(\pi)$}

We begin by presenting the doubly robust score for estimating $W_\Lambda (\pi)$. As shown in \cref{theorem: max-min welfare assignment rule}, this naturally leads to the following moment condition:
\[
\mathbb{E}\left[ \mu_\Lambda^-(X, 1) \pi (1|X) +    \mu_\Lambda^-(X, 0) \pi (0|X)  \right]  -  W_\Lambda(\pi) = 0.
\]
For $t \in \{0,1\}$, let
\begin{equation*}
 g_{t}\left(z;e,q_\Lambda^- \right) = y \mathds{1} \{a= t\} \left[1+\frac{1- e_t(x)}{e_t(x)}\Lambda^{-\text{sgn}\big(y-q_{\Lambda}^{-}(x,t)\big)}\right] .
\end{equation*}
Using explicit expressions for $\mu_\Lambda^{\pm}$ given in \cref{proposition: pi_mean}, the moment condition above can be rewritten as
\begin{equation}\label{equation: Moment_condition_W_Lambda}
\mathbb{E}\left[ \sum_{t\in\{0,1\}}g_{t}\left(Z;e,q_\Lambda^- \right)\pi(t|X) \right]  - W_\Lambda(\pi) = 0.
\end{equation}
Here, $W_\Lambda(\pi)$ is the parameter of interest, while $e(\cdot)$ and $q_{\Lambda}^-(\cdot, \cdot)$ are two unknown nuisance parameters that can be estimated using ML/nonparametric methods.

A plug-in estimator for $W_\Lambda(\pi)$ can be constructed by substituting estimates into the score function $g_t(Z; \cdot)$ and averaging over the sample, yielding $n^{-1} \sum_{i=1}^n\sum_{t\in\{0,1\}} g_t\left(Z_i; \widehat{e}, \widehat{q}^{-}_{\Lambda}\right)\pi(t|X_{i})$. However, this plug-in estimator is sensitive to estimators of the nominal propensity score and the quantile functions, and may suffer from severe bias. This motivates the doubly robust estimator proposed in this section.

Following \cite{chernozhukov2018double,chernozhukov2022locally}, we construct the doubly robust score for $W_\Lambda(\pi)$ as follows:
\begin{equation}\label{equation: phi_t^- definition}
    \begin{aligned}
         &\phi_{t}^{-} \left(z;e,q_\Lambda^-, \rho^-_{1,\Lambda}, \rho^-_{0,\Lambda}\right)\\
         & =   y \mathds{1} \{a= t\} \left[1+\frac{1- e_t(x)}{e_t(x)}\Lambda^{-\text{sgn}\big(y-q_{\Lambda}^{-}(x,t)\big)}\right]  \\
        &+ q_{\Lambda}^{-}(x,t) \mathds{1} \{a= t\}   \frac{1-e_t(x)}{e_t(x)}\left( \Lambda - \Lambda^{-1} \right) \left[ \frac{1}{1+\Lambda} - \mathbb{1}\left\{ y< q_{\Lambda}^{-}(x,t)\right\} \right]   \\
        & - \frac{1}{e_t(x)}\left[\Lambda\rho_{1,\Lambda}^{-}(x,t) + \Lambda^{-1} \rho_{0,\Lambda}^{-}(x,t)  \right] \left[ \mathds{1}\{ a =t \} -e_t(x) \right],
    \end{aligned}
\end{equation}
where
\begin{equation}\label{equation: rho}
    \begin{aligned}
        & \rho_{1,\Lambda}^{\pm}(x,t) = \mathbb{E}\left[ Y\mathbb{1}\left\{ Y < q_{\Lambda}^{\pm}(X,A) \right\} | X=x, A=t \right], \\
        & \rho_{0,\Lambda}^{\pm}(x,t) = \mathbb{E}\left[ Y\mathbb{1}\left\{ Y > q_{\Lambda}^{\pm}(X,A) \right\} | X=x, A=t \right]. \\
    \end{aligned}
\end{equation}
Let $\eta_{W, \Lambda} = \left(e,q_\Lambda^- ,\rho^-_{1,\Lambda}, \rho^-_{0,\Lambda} \right) $ denote the tuple of  {nuisance} functions, the doubly robust score for $W_\Lambda(\pi)$ is given by 
\[
\psi_{W}\left(z,\pi;\eta_{W, \Lambda} \right) = \sum_{t \in \{0,1\} } \phi_{t}^{-} \left(z; \eta_{W, \Lambda}  \right) \pi(t|x).
\]

The following \cref{proposition: doubly robust score for mean} establishes the Neyman orthogonality for the moment condition.
Moreover, the function $\psi_W(z, \pi; \eta_{W, \Lambda}) - W_\Lambda(\pi)$ serves as the efficient influence function for estimating $W_\Lambda(\pi)$.
\begin{proposition}\label{proposition: doubly robust score for mean}
Under  {Assumptions \ref{assumption: Marginal sensitivity model} and \ref{assum: density}}, for any policy $\pi \in \Pi$, the score function $\psi_{W}$ satisfies:
\begin{propenum}
    \item $\mathbb{E}\left[ \psi_{W}\left(Z,\pi;\eta_{W, \Lambda}  \right)  \right] = W_\Lambda(\pi)$. 
    \item  For any $\widetilde{\eta}_{W, \Lambda} = (\widetilde{e}, \widetilde{q}^{-}_{\Lambda}, \widetilde{\rho}_{1,\Lambda}^-, \widetilde{\rho}_{0,\Lambda}^- )$ with $\widetilde{e}: \mathcal{X} \rightarrow (0,1)$, $\widetilde{q}^-_{\Lambda}: \mathcal{X} \times \{0,1\} \rightarrow \mathbb{R}$, and $\widetilde{\rho}_{1,\Lambda}^-, \widetilde{\rho}_{0,\Lambda}^-: \mathcal{X} \times \{0,1\} \rightarrow \mathbb{R}$, the pathwise (or Gateaux) derivative satisfies 
    \begin{equation*}
        \frac{\mathrm{d}}{\mathrm{d} r}\mathbb{E}\left[ \psi_{W}\left(Z,\pi;\eta_{W, \Lambda}+r\left(\widetilde{\eta}_{W, \Lambda}-\eta_{W, \Lambda}\right)\right) \right]_{r=0} = 0.
    \end{equation*}
\end{propenum}
\end{proposition}

\begin{remark}
The framework of \citet{kallus2021minimax} relies on estimating the nominal propensity score at the parametric rate of $O_P(n^{-1/2})$ to achieve an $n^{-1/2}$ regret bound. However, this rate can only be attained under correct model specification, ruling out the use of  {ML} and nonparametric methods.  In contrast, our framework allows for the nuisance components to be estimated at a slower rate of $o_P(n^{-1/4})$, enabling the use of data-adaptive, model-agnostic estimation techniques.
\end{remark}

\subsubsection*{Doubly Robust Score for $\Delta_\Lambda(\pi)$}

We then construct the doubly robust score of $\Delta_{\Lambda}(\pi)$. \cref{theorem: CATE-based robust welfare} implies the moment condition as follows:
\[
\mathbb{E}\left[ \tau_{\Lambda}^{-} (X) \pi(X)  \right] - \Delta_{\Lambda}(\pi) = 0.
\]
Let $\eta_{\Delta, \Lambda} = \left(e,q_\Lambda^{\pm}, \rho^{\pm}_{1,\Lambda}, \rho^{\pm}_{0,\Lambda} \right)$ denote the tuple of nuisance parameters. The doubly robust score for $\Delta_\Lambda(\pi)$ can be constructed as
\begin{equation*}
    \psi_{\Delta} \left( z,\pi; \eta_{\Delta, \Lambda} \right) = \pi(x) \left[ \phi_{1}^{-} \left(z;e,q_\Lambda^-, \rho^-_{1,\Lambda}, \rho^-_{0,\Lambda}\right)  - \phi_{0}^{+} \left(  z;e,q^+_\Lambda
        ,\rho_{1,\Lambda}^+, \rho_{0,\Lambda}^+  \right) \right] ,
\end{equation*}
where
\begin{equation}\label{equation: phi_+}
    \begin{aligned}
        &\phi^{+}_t \left( z;e,q^+_\Lambda
        ,\rho_{1,\Lambda}^+, \rho_{0,\Lambda}^+ \right) \\
        & =   
           y \mathds{1}\{ a = t \} \left[1+\frac{1- e_t(x)}{e_t(x)}\Lambda^{\text{sgn}\big(y-q^{+}_{\Lambda}(x,t)\big)}\right] \\
        & -  q^{+}_{\Lambda}(x,t) \mathds{1}\{ a = t \}  \frac{1-e_t(x)}{e_t(x)}\left( \Lambda - \Lambda^{-1} \right) \left[ \frac{\Lambda}{1+\Lambda} - \mathbb{1}\left\{ y< q^{+}_{\Lambda}(x,t)\right\} \right]   \\
        & -   \frac{1}{e_t(x)}\left[\Lambda^{-1}\rho_{1,\Lambda}^{+}(x,t) + \Lambda\rho^{+}_{0,\Lambda}(x,t) \right]\left[\mathds{1}\{ a = t \} -e_t(x) \right].  \\
    \end{aligned}
\end{equation}
We summarize several crucial properties of $\psi_{\Delta}$ in the following proposition.
\begin{proposition}\label{proposition: doubly robust score for CATE}
    Under  {Assumptions \ref{assumption: Marginal sensitivity model} and \ref{assum: density}}, for any policy $\pi \in \Pi$, the score function $\psi_{\Delta}$ satisfies:
\begin{propenum}
    \item $\mathbb{E}\left[ \psi_{\Delta}\left(Z,\pi;\eta_{\Delta, \Lambda}  \right)  \right] = \Delta_\Lambda(\pi)$. 
    \item  For any $\widetilde{\eta}_{\Delta,\Lambda} = \left(\widetilde{e}, \widetilde{q}^{\pm}_{\Lambda}, \widetilde{\rho}_{1,\Lambda}^{\pm}, \widetilde{\rho}_{0,\Lambda}^{\pm} \right)$ with $\widetilde{e}: \mathcal{X} \rightarrow (0,1)$, $\widetilde{q}^{\pm}_{\Lambda}: \mathcal{X} \times \{0,1\} \rightarrow \mathbb{R}$, and $\widetilde{\rho}_{1,\Lambda}^{\pm}, \widetilde{\rho}_{0,\Lambda}^{\pm}: \mathcal{X} \times \{0,1\} \rightarrow \mathbb{R}$, the pathwise (or Gateaux) derivative satisfies 
    \begin{equation*}
        \frac{\mathrm{d}}{\mathrm{d} r}\mathbb{E}\left[ \psi_{\Delta}\left(Z,\pi;\eta+r\left(\widetilde{\eta}_{\Delta, \Lambda} -\eta_{\Delta, \Lambda} \right)\right) \right]_{r=0} = 0.
    \end{equation*}
\end{propenum}
\end{proposition}

\begin{remark}
The efficient influence functions for \( W_\Lambda(\pi) \) and \( \Delta_\Lambda(\pi) \) are given by 
\( \psi_W( {z}, \pi; \eta_{W, \Lambda}) - W_\Lambda(\pi) \) and 
\( \psi_\Delta( {z}, \pi; \eta_{\Delta, \Lambda}) - \Delta_\Lambda(\pi) \), respectively; see \cite{newey1994asymptotic}.
\end{remark}

\section{Algorithm and Performance Guarantees}\label{section: regret bounds}

In this section, we focus on the algorithms for estimating the MMW and MMI policies, as  defined in \cref{section: Policy Learning beyond Unconfoundedness}. Let  $\widehat{\pi}_{W,\Lambda}$ and $\widehat{\pi}_{\Delta,\Lambda}$ denote the estimated MMW and MMI policies, respectively.  We  also provide theoretical  guarantees for these estimated policies in the form of upper bounds on their associated regrets.  Specifically, we derive asymptotic upper bounds on the CRW-regret $\mathrm{Reg}_W\left(\widehat{\pi}_{W, \Lambda} \right)$ 
and the CRI-regret $\mathrm{Reg}_{\Delta}\left(\widehat{\pi}_{\Delta, \Lambda} \right)$, as defined in \cref{robust_welfare_regret} and \cref{robust_improvement_regret}.

\subsection{Algorithm}\label{section: algorithm}

We estimate the MMW and MMI policies using a  two-stage procedure: (1) estimation of the robust criterion, either  $W_\Lambda(\pi)$ or $\Delta_\Lambda(\pi)$, using $K$-fold cross-fitting; and (2) policy optimization based on the estimated objective.

 Given the doubly robust scores in \cref{section: Identification and Estimation of the Worst-case Welfare Function}, we first estimate $e(\cdot)$, $q^{\pm}_{\Lambda}(\cdot, \cdot)$ and $\rho_{t,\Lambda}^{\pm}(\cdot, \cdot)$ for $t\in\{0,1\}$.  To this end, we divide the sample into $K$ evenly-sized folds $\cup_{k=1}^{K}\mathcal{I}_{k}$. For each fold $k \in [K]$, we estimate $e, q^{\pm}_{\Lambda}$ and $\rho_{t,\Lambda}^{\pm}$ using  {the other $(K-1)$ folds}, and denote the resulting estimates by 
 \[
 \widehat{\eta}_{W,\Lambda}^{-k} =\left( \widehat{e}^{-k}, q^{-, -k}_{\Lambda}, \rho_{1,\Lambda}^{-, -k }, \rho_{0,\Lambda}^{-, -k } \right) \quad\text{and}\quad \widehat{\eta}_{\Delta,\Lambda}^{-k} =\left( \widehat{e}^{-k}, q^{\pm, -k}_{\Lambda}, \rho_{1,\Lambda}^{\pm, -k }, \rho_{0,\Lambda}^{\pm, -k } \right).
 \]
 Using these estimates, we construct the cross-fitted estimator for $W_\Lambda(\pi)$ and $\Delta_{\Lambda}(\pi)$ as
\begin{equation}\label{equation:W_Lambda}
 \begin{aligned}
     \widehat{W}_{\Lambda,n}(\pi) = \frac{1}{n} \sum_{k=1}^{K}\sum_{i\in\mathcal{I}_{k}} \psi_W\left(Z_i, \pi;\widehat{\eta}_{W,\Lambda}^{-k}  \right ) \ \text{and}\ \widehat{\Delta}_{\Lambda,n}(\pi) = \frac{1}{n} \sum_{k=1}^{K}\sum_{i\in\mathcal{I}_{k}} \psi_\Delta\left(Z_i, \pi;\widehat{\eta}_{\Delta,\Lambda}^{-k}  \right ).
 \end{aligned}
\end{equation}
The final policy is selected from a pre-specified policy class  $\Pi_n$ by maximizing the estimated objective, either  $\widehat{W}_{\Lambda,n}(\pi)$ or $\widehat{\Delta}_{\Lambda,n}(\pi)$. The complete procedures for estimating the MMW and MMI policies are summarized in \cref{algo:MMW_PL} and \cref{algo:MMI_PL}, respectively.

 \begin{algorithm}
\caption{Max-Min Welfare (MMW) Policy Learning }\label{algo:MMW_PL}
\begin{algorithmic}[1]
    \State {\bf Input:} Sample $\{Y_i,X_i, A_i\}_{i=1}^n$, sensitivity parameter $\Lambda \geq 1$, and a policy class $\Pi_n$

    \State Choose  $K \in \mathbb{N}^+$, and partition the sample into $K$ equally sized folds 
    $\cup_{k=1}^{K}\mathcal{I}_{k}$
    \State {\bf for} each $k = 1,\ldots,K$ {\bf do}

    \State Fit estimators for $\eta_{W, \Lambda} = \left(e,q_\Lambda^- ,\rho^-_{1,\Lambda}, \rho^-_{0,\Lambda} \right) $ using  {the other $(K-1)$ folds},  $\mathcal{I}_{k}^{c}\equiv[n]\backslash\mathcal{I}_{k}$, denoted by $\widehat{\eta}_{W,\Lambda}^{-k}$.

    \State {\bf end for}

    \State Estimate $W_\Lambda(\pi)$ as
    \[
    \widehat{W}_{\Lambda,n}(\pi) = \frac{1}{n} \sum_{k=1}^{K}\sum_{i\in\mathcal{I}_{k}} \psi_W\left(Z_i, \pi;\widehat{\eta}_{W,\Lambda}^{-k}  \right ).
    \]
    
   \State Return $\widehat{\pi}_{W,\Lambda} = \argmax_{\pi \in \Pi_n}  \widehat{W}_{\Lambda,n} (\pi)$. 
\end{algorithmic}
\end{algorithm}

 \begin{algorithm}[H]
\caption{Max-Min Improvement (MMI) Policy Learning }\label{algo:MMI_PL}
\begin{algorithmic}[1]
    \State {\bf Input:} Sample $\{Y_i,X_i, A_i\}_{i=1}^n$, sensitivity parameter $\Lambda \geq 1$, and a policy class $\Pi_n$

    \State Choose  $K \in \mathbb{N}^+$, and partition the sample into $K$ equally sized folds 
    $\cup_{k=1}^{K}\mathcal{I}_{k}$
    \State {\bf for} each $k = 1,\ldots,K$ {\bf do}

    \State Fit estimators for $\eta_{\Delta, \Lambda} = \left(e,q_\Lambda^{\pm} ,\rho^{\pm}_{1,\Lambda}, \rho^{\pm}_{0,\Lambda} \right) $ using  {the other $(K-1)$ folds},  $\mathcal{I}_{k}^{c}\equiv[n]\backslash\mathcal{I}_{k}$, denoted by $\widehat{\eta}_{\Delta, \Lambda}^{-k}$.

    \State {\bf end for}

    \State Estimate $\Delta_\Lambda(\pi)$ as
    \begin{equation}\label{equation: Delta_Lambda_estimation}
    \widehat{\Delta}_{\Lambda,n} (\pi) = \frac{1}{n} \sum_{k=1}^{K}\sum_{i\in\mathcal{I}_{k}} \psi_\Delta\left(Z_i, \pi;\widehat{\eta}_{\Delta,\Lambda}^{-k}  \right ).    
    \end{equation}
    
   \State Return $\widehat{\pi}_{\Delta,\Lambda} = \argmax_{\pi \in \Pi_n}  \widehat{\Delta}_{\Lambda,n}  (\pi)$. 
\end{algorithmic}
\end{algorithm}

\subsection{Assumptions about the Policy Class}\label{section: policy_class}

In policy design, it is crucial to account for multiple constraints, including budget, simplicity, interpretability, and functional form. Statistically speaking, to achieve regret bounds that decay at the rate of $n^{-1/2}$, one must control the complexity of the policy class. In the remainder of the paper, we allow the policy class $\Pi\equiv \Pi_n$  to vary with  $n$. Consequently, we restrict $\Pi_n$ to be VC-subgraph class;  see \cite{vaart2023empirical, gine2021mathematical} for further details. 

\begin{assumption}\label{assumption: policy size}
The policy class \( \Pi_n \) is VC- {subgraph} with VC dimension satisfying \( \mathrm{VC}(\Pi_n) \leq n^{\zeta_{\Pi}} \), where \( 0 < \zeta_{\Pi} < \tfrac{1}{2} \).
\end{assumption}

\begin{remark}
In the statistical learning literature,  $\mathrm{VC}(\Pi_n)$ is often referred to as the pseudo-dimension, which generalizes the classical Vapnik-Chervonenkis (VC) dimension from binary-valued function classes to real-valued ones; see \cite{anthony2009neural,bartlett2019nearly}. For binary-valued function classes, the VC dimension and pseudo-dimension are identical. In this work, we do not distinguish between the classical VC dimension and the pseudo-dimension, and refer to both simply as the VC dimension.
\end{remark}

Various  {ML} models can be used as policy classes. Below, we present several examples of such models along with their corresponding VC dimensions, covering both deterministic and randomized policy classes. 

\begin{example}[Linear Policies]\label{example: Linear policies}
 The class of (deterministic)  linear policies is defined as
\[
\Pi_{n} = \Big\{  \mathbb{1} \{ T(x)^\prime\beta>0 \}:\beta\in\mathbb{R}^{d_{n}}   \Big\},
\]
where $T(x) \in \mathbb{R}^{d_n}$ denotes a vector of transformed covariates constructed from the raw features $X$ using basis expansions such as polynomial terms, B-splines, or interaction terms. The VC dimension of linear policy class $\Pi_{n}$ is given by $d_{n}+1$.
\end{example}

\begin{example}[Decision Trees] A decision tree is a classifier $\pi:\mathcal{X}\to\{0,1\}$ that recursively partitions the feature space $\mathcal{X}$ into disjoint rectangular regions, assigning a label to each partition. The class of depth-$L$ decision trees in $\mathbb{R}^{d}$ has VC dimension bounded on the order of $O\left( 2^{L}\log{d} \right)$.
\end{example} 

Next, we provide several examples of randomized policy classes.

\begin{example}[Neural Networks] Deep neural networks have achieved remarkable success in solving complex classification and regression problems. Formally, a neural network defines a class of functions mapping from  $\mathcal{X}$ to $\mathbb{R}$. Such networks can be employed to represent both deterministic and randomized policy classes. Formally, neural networks model the relationship between inputs and outputs through layers of interconnected computational units (neurons), inspired by biological neural systems. The class of networks with $L$ layers, $p$ parameters and a piecewise linear activation function has a VC dimension bounded on the order of $O\left( Lp\log{p} \right)$; see \cite{bartlett2019nearly}.
\end{example}

\begin{example}[Logistic Policies]\label{example:Logistic Policies}
The class of logistic policies is defined as
\[
\Pi_n = \left \{ \sigma\left( T(x)^\prime \beta \right) : \beta \in \mathbb{R}^{d_n}  \right  \} ,
\]
where $\sigma$ is standard logistic function and \( T(x) \) denotes a transformation of covariates \( x \) as in \cref{example: Linear policies}. It is noted that  $\sigma$ is strictly increasing, and $\Pi_n = \sigma \circ \{ T(x)^\prime \beta : \beta \in \mathbb{R}^{d_n} \} $. As a result,   Theorem 2.6.18 in \cite{vaart2023empirical} implies $\mathrm{VC}(\Pi_n) = d_n + 1$.
\end{example}

\subsection{Nuisance Estimators and Uniform Coupling}

In this section, we present a key lemma demonstrating that the estimation error of the nuisance parameters becomes asymptotically negligible when the objective functions $W_{\Lambda}(\cdot)$ and $\Delta_{\Lambda}(\cdot)$ are estimated using the doubly robust score
with cross-fitting.

\begin{assumption}\label{assumption: bounded support}
Suppose that  $Y$ and $q^{\pm}_{\Lambda}(X, a)$, for $a\in \{0,1\}$ have finite second moments; that is, $\mathbb{E}|Y|^2< \infty$ and $\mathbb{E}\left|q^{\pm}_{\Lambda}(X, a)\right|^2 < \infty$.
\end{assumption}

 
 \begin{assumption}\label{assumption: strict overlap}
    There exist $\kappa\in\left(0, 1/2\right)$ such that the nominal propensity score satisfies $e(x)\in(\kappa,1-\kappa)$ for all $x\in \mathcal{X}$.
\end{assumption}

\Cref{assumption: bounded support} is straightforward to verify. \Cref{assumption: strict overlap} is the standard strict overlap condition in the literature, which is essential for establishing the  {asymptotic} regret upper bounds. Moreover, we adopt an agnostic stance on how the nuisance components, $\widehat{e}$, $ \widehat{q}_{\Lambda}$ and $\widehat{\rho}_{t, \Lambda}$, are estimated. Rather than specifying particular estimation procedures, we impose high-level assumptions on their convergence ratio, as stated below.

\begin{assumption}\label{assumption: nuisance parameter estimation error}
Suppose we have uniformly consistent estimators of the nuisance parameters such that
\begin{equation}\label{eq:sup_consistency}
\begin{aligned}
& \sup_{x} \left| \widehat{e} \left(x\right) - e(x)\right|, \quad  \sup_{x ,a } \left|   \widehat{q}^{\pm}_{\Lambda} (x,a) - {q}^{\pm}_{\Lambda} (x,a) \right|, \quad  \sup_{x, a, t}  \left| \widehat{\rho}_{t,\Lambda}^{\pm}(x,a) - \rho_{t,\Lambda}^{\pm}(x,a) \right|  \xrightarrow{P} 0. \\
\end{aligned}
\end{equation}
Furthermore, there exist constants $\zeta_e, \zeta_q, \zeta_\rho \geq 1/2$ and a sequence $b(n)=o(1)$ such that
\begin{equation}\label{eq:L2_consistency}
\begin{aligned}
& \mathbb{E}  \left[ \left| \widehat{e} \left(X_i\right) - e(X_i)\right|^2\right] \leq  
 \frac{b(n)}{n^{\zeta_e} }, \quad \quad  \mathbb{E} \left[ \left|   \widehat{q}^{\pm}_{\Lambda} (X_i, A_i) - {q}^{\pm}_{\Lambda} (X_i, A_i) \right|^2 \right] \leq  
 \frac{b(n)}{n^{\zeta_q} } , \\
& \mathbb{E} \left[  \left| \widehat{\rho}_{t,\Lambda}^{\pm}(X_i,A_i) - {\rho}_{t,\Lambda}^{\pm}(X_i,A_i) \right|^2 \right] \leq 
 \frac{b(n)}{n^{\zeta_\rho} }.
\end{aligned}
\end{equation}
\end{assumption}

\begin{remark}\label{Remark: examples for estimation of nuisance parameters}
The regression and quantile functions, $e(x)$ and  $q^{\pm}(x,a)$ can be estimated using kernel and  {sieve methods}; see \citep{li2007nonparametric,chen2007large,belloni2019conditional}. Recent advances have widely adopted  {ML} methods for estimating nuisance parameters, with desirable asymptotic properties well established in the literature; see, e.g., Lasso-based approaches \citep{belloni2011ℓ, belloni2014high, belloni2017program} and deep neural network estimators \citep{schmidt2020nonparametric,farrell2021deep, kohler2021rate,padilla2022quantile}. For estimating $\rho_{t,\Lambda}^{\pm}$, which corresponds to the conditional value at risk (CVaR) or expected shortfall, a variety of methods have been developed, including kernel, local linear, and deep neural network approaches; see \citep{fissler2023deep, olma2021nonparametric, yu2025estimation}.
\end{remark}

\begin{remark}
 {\cref{assumption: nuisance parameter estimation error} is similar to Assumption 2 in \citet{athey2021policy}. First, the condition \eqref{eq:L2_consistency} closely follows standard semiparametric assumptions, which typically require an $n^{-1/4}$-convergence rate in $L^2$-norm for nuisance parameters \citep{belloni2017program,chernozhukov2018double,chernozhukov2022locally}.\footnote{Note that when $\zeta_e=\zeta_q=\zeta_\rho=1/2$, the condition \eqref{eq:L2_consistency} corresponds to the $o(n^{-1/4})$ convergence rate in $L^2$-norm. Note that \eqref{eq:L2_consistency} is sufficient for pointwise welfare estimation, as this rate renders nuisance errors stochastically negligible for asymptotic normality.} Second, the condition \eqref{eq:sup_consistency}  is essential in  policy learning {\it with observational data} to withstand the policy optimization over large class $\Pi_n$. Specifically, this ensures that the welfare estimator with plug-in nuisance parameter estimates uniformly approximates its oracle counterpart, as established in \cref{lemma: nuisance parameter estimation error}. }
\end{remark}


The use of doubly robust scores combined with cross-fitting ensures that errors from nuisance parameter estimation become asymptotically negligible for policy learning, provided $\mathrm{VC}(\Pi_n)$ does not grow too quickly with $n$. To formalize this, suppose  the true nuisance functions $\eta_{W, \Lambda}$ and $\eta_{\Delta, \Lambda}$ defined in \cref{section: Identification and Estimation of the Worst-case Welfare Function} are known, the oracle estimators for $W_\Lambda(\pi)$ and $\Delta_\Lambda(\pi)$  {are given by}
\[
\begin{aligned}
W_{\Lambda,n}(\pi) & = \frac{1}{n} \sum_{i=1}^{n} \psi_{W}\left(Z_i,\pi; \eta_{W, \Lambda} \right), \\
\Delta_{\Lambda,n}(\pi) & = \frac{1}{n} \sum_{i=1}^{n} \psi_{\Delta}\left(Z_i,\pi; \eta_{\Delta, \Lambda} \right).
\end{aligned}
\]
We conclude this subsection by demonstrating that $\widehat{W}_{\Lambda,n}(\pi)$ is a valid approximation to $W_{\Lambda,n}(\pi)$, and similarly, $\widehat{\Delta}_{\Lambda,n}(\pi)$ approximates $\Delta_{\Lambda,n}(\pi)$, both with a convergence rate faster than $n^{-1/2}$.

\begin{lemma}\label{lemma: nuisance parameter estimation error}
Under   {Assumptions \ref{assum: density},  \ref{assumption: policy size}, \ref{assumption: bounded support}, \ref{assumption: strict overlap} and \ref{assumption: nuisance parameter estimation error}}, we have 
\[
\begin{aligned}
\mathbb{E} \left[ \sup_{\pi \in \Pi_n} \left|\widehat{W}_{\Lambda,n}(\pi) - W_{\Lambda,n}(\pi) \right| \right] &= o(n^{-1/2}), \\
\mathbb{E} \left[ \sup_{\pi \in \Pi_n} \left|\widehat{\Delta}_{\Lambda,n}(\pi) - \Delta_{\Lambda,n}(\pi) \right| \right] &= o(n^{-1/2}).
\end{aligned}
\]
\end{lemma}

\subsection{Asymptotic Regret Bounds}



In this subsection, we will demonstrate that, under appropriate bounds on potential hidden confounding, the CRW-regret and CRI-regret of the learned MMW and MMI policies decay at a rate that is upper bounded by  $\sqrt{\mathrm{VC}(\Pi_n)/n}$. Recall that the CRW-regret and CRI-regret of the learned optimal policies are defined as 
\[
\begin{aligned}
 \mathrm{Reg}_W \left( \widehat{\pi}_{W,n}  \right) & =   \sup_{\pi \in \Pi_n} W_\Lambda ( \pi) - W_\Lambda ( \widehat{\pi}_{W,n} ), \\
  \mathrm{Reg}_\Delta \left( \widehat{\pi}_{\Delta,n}  \right) & =   \sup_{\pi \in \Pi_n} \Delta_\Lambda ( \pi) - \Delta_\Lambda ( \widehat{\pi}_{\Delta,n} ),
\end{aligned}
\]
where $\widehat{\pi}_{W,n}$ and $ \widehat{\pi}_{\Delta,n}$ are learned from \cref{algo:MMW_PL} and \cref{algo:MMI_PL}.

\begin{theorem}\label{theorem: asymptotic welfare loss}
Under  {Assumptions \ref{assumption: Marginal sensitivity model}, \ref{assum: density}, \ref{assumption: policy size}, \ref{assumption: bounded support}, \ref{assumption: strict overlap} and \ref{assumption: nuisance parameter estimation error}}, then 
\begin{align}
\mathbb{E} \left[  \mathrm{Reg}_W \left( \widehat{\pi}_{W,n}  \right)  \right] &= O\left(  \sqrt{  \mathrm{VC} (\Pi_n) \big / n }    \right), \label{equation:regret_MMW}  \\
\mathbb{E} \left[  \mathrm{Reg}_{\Delta}  \left( \widehat{\pi}_{\Delta,n}  \right)  \right]   &= O\left(  \sqrt{  \mathrm{VC} (\Pi_n) \big / n }    \right).
\label{equation:regret_MMI} 
\end{align}
\end{theorem}

\cref{theorem: asymptotic welfare loss} improves upon the result in \cite{kallus2021minimax}, where Proposition 6 shows that estimating the nominal propensity score leads to a regret bound that exceeds the oracle regret bound (i.e., the regret bound assuming a known nominal propensity score)  by an additional term of the order
\[
\frac{1}{n} \sum_{i=1}^n \left|  1/\widehat{e}\left(X_i\right) - 1/e(X_i) \right|.
\]
With the application of the doubly robust score and cross-fitting, the regrets of the learned MMW and MMI policies decay at the rate  $\sqrt{\mathrm{VC}(\Pi_n)/n}$, and the estimation errors of the nuisance parameters have no asymptotic effect on the regret bound.

\begin{remark}
Let us analyze \cref{equation:regret_MMW} in \cref{theorem: asymptotic welfare loss} in detail; the analysis of \cref{equation:regret_MMI} proceeds in the same fashion.  To be more precise,  there is a universal constant $K > 0$ not depending on $\Lambda$  such that 
\begin{equation}\label{equation: Upper_bound_W}
\limsup_{n \rightarrow \infty}  \frac{\mathrm{Reg}_W \left( \widehat{\pi}_{W,n}  \right) }{  \sqrt{  \mathrm{VC} (\Pi_n) / n }  } \leq K \sigma_W(\Lambda)  
\end{equation}
where $\sigma_W(\Lambda)^2 =  \mathbb{E} \left| \left(\phi_0^- - \phi_1^-\right) (Z, \eta_{W, \Lambda})\right|^2$, and  $\phi_t^-$ for $t \in \{0,1\}$ are defined in \cref{equation: phi_t^- definition}.

The asymptotic upper bound given in \cref{equation: Upper_bound_W} depends on the universal constant $K$ and a variance complexity term $\sigma_W(\Lambda)$ that captures the second moments of $\phi_{t}^- (Z, \eta_{W, \Lambda})$ for $t \in \{0,1\}$. As shown in the proof of \cref{theorem: asymptotic welfare loss},  the constant $K$ is universal in the sense that it  does not depend on sample size $n$, the sensitivity parameter $\Lambda$, the nuisance function $\eta_{W, \Lambda}$, or even  the underlying distribution of $P_o\in \mathcal{P}(\Lambda)$. The variance complexity term $\sigma_W(\Lambda)$ reflects the level of uncertainty induced by the unobserved confounding. When $\Lambda = 1$, then $\sigma_W^2(1)$ coincides with the semiparametric efficient variance for the average treatment effect under the selection-on-observables assumption; see \cite{newey1994asymptotic, robins1994estimation}.
\end{remark}

\section{Simulation Studies}\label{section: simulation}

In this section, we conduct simulation studies to evaluate the performance of the MMW and MMI policies, which are learned using Algorithm~\ref{algo:MMW_PL} and \ref{algo:MMI_PL}, respectively. To facilitate comparison with the results in \citet{kallus2021minimax}, we use the logistic policies given in \cref{example:Logistic Policies}.

\subsection{Simulation Design}
\label{section:simulation-design}

The data-generating process (DGP) is specified as follows. Let $\log \Lambda^* = 1.5$, which implies that the true sensitivity parameter is $\Lambda^* = 4.482$. The observed covariates $X \in \mathbb{R}^2$, treatment assignment $A \in \{0,1\}$, and outcome $Y \in \mathbb{R}$ are generated according to:
\[
\begin{aligned}
X &\sim \mathcal{N}(\mu_X, I_2), \quad U \mid X \sim \mathrm{Bern}\left( \frac{\Lambda^*}{1 + \Lambda^*}e(X)  +  \frac{1}{1 +\Lambda^*}\big( 1-e(X) \big) \right), \\
A & \mid U,X \sim  \mathrm{Bern}\left( \frac{e(X)}{e(X)+{\Lambda^*}^{(1- 2U)}\big( 1-e(X) \big)} \right),
\end{aligned}
\]
where $\mu_X = [-1, 1]'$ and $I_2\in \mathbb{R}^{2\times 2}$ is the identity matrix. The nominal propensity score $e(X)$ is defined as:
\[
e(X) = \mathbb{P}(A=1 \mid X) = \sigma\left(  \zeta(X)^\prime\theta \right), \quad \text{with} \quad \sigma(z) = \frac{1}{1 + e^{-z}},
\]
where the parameter vector $\theta$ and the nonlinear feature map $\zeta(X)$ are given by
\[
\theta = \left[0.2, 0.4, 0.1, -0.1, 0.5, -0.5\right], \quad \zeta(X) = \left[\max(x_1, 0), \, \frac{x_1 x_2^2}{10}, \, \sin(x_2^2), \, x_1, \, x_2, \, 1\right]'.
\]

It can be readily verified that $U$ perturbs the nominal odds ratio $e(X)/(1 - e(X))$ by a multiplicative factor of ${\Lambda^*}^{\pm 1}$, ensuring the DGP satisfies the MSM with $\Lambda^* = 4.482$. Finally, the potential outcome is generated as:
\[
Y(A) = \beta_{\text{cons}} + \beta_A A + X' \beta_X + (A \cdot X)'\beta_{X,A}  + \beta_U U + \epsilon, \quad \epsilon \sim \mathcal{N}(0,1),
\]
where:
\[
\beta_{\text{cons}} = -0.2, \quad \beta_A = -0.1, \quad \beta_X = [1, -1]', \quad \beta_{X,A} = [0.2, 0.4]', \quad \beta_U = 1.5.
\]

We generate $\mathrm{i.i.d.}$ samples $\{ X_i, A_i, U_i, Y_i(1), Y_i(0)\}_{i=1}^{n}$ according to the DGP described above, with $n=2{,}000$. The observed data $\{ X_i, A_i, Y_i \}_{i=1}^{n}$ are used to learn the policies. The conditional quantile functions $q_{ {\Lambda}}^{\pm}$ are estimated using gradient boosted trees and the nominal propensity score $e$ and CVaR $\rho_{t,  {\Lambda}}^{\pm}$ are estimated using random forests.

We evaluate the performance of estimated policy $\widehat{\pi}$  using three metrics: (i) the expected welfare, defined in \cref{equation: standard policy learning}; (ii) the worst-case welfare $W_\Lambda(\widehat{\pi})$, as given in Theorem~\ref{theorem: max-min welfare assignment rule}; and (iii) the worst-case policy improvement $\Delta_\Lambda(\widehat{\pi})$ given in Theorem~\ref{theorem: CATE-based robust welfare}. 
All three metrics are estimated using sample averages computed on an independent out-of-sample dataset of 100,000 randomly drawn observations. To assess the stability of our method, we evaluate each metric over 100 repeated experiments. We report the mean from these repetitions and visualize the variability using shaded bands that represent the 95\% confidence bands, calculated from the standard deviation across the 100 experiments.

\subsection{Simulation Analysis}
\label{section:simulation-analysis}

We report the performance of the MMW and MMI policies, as well as the robust IPW-based policy of \citet{kallus2021minimax} (KZ for short). As a benchmark, we include the policy proposed by \citet{athey2021policy} (AW), which assumes no unobserved confounding $(\Lambda =1)$. To evaluate the MMW, MMI, and KZ policies, we conduct a sensitivity analysis by varying the sensitivity parameter over $\log \Lambda \in \{ 0.1,0.2,\ldots, 3.5 \}$. This assesses their robustness to miscalibration relative to the true level of unobserved confounding in the true DGP with $\log \Lambda^* = 1.5$. By design, the MMW and MMI policies reduce to the AW policy at the $\Lambda=1$ setting. In contrast, the KZ policy remains different from AW even when $\Lambda=1$, as its objective function does not include a bias-correction term for the estimated propensity scores.

\Cref{fig:fix_dgplambda_treated_probs} shows how each policy adjusts its treatment probability in response to the assumed level of unobserved confounding, and helps explain their relative performance in terms of expected welfare, presented in \Cref{fig:fix_dgplambda_average_welfare}. The AW policy, assuming no unobserved confounding($\Lambda = 1$), yields an aggressive strategy, assigning treatment to 95\% of individuals. In contrast, MMI, MMW and KZ are sensitivity-aware, adopting a conservative approach that is responsive to $\Lambda$. As $\log \Lambda$ increases, they systematically reduce treatment rates to hedge against selection bias. However, their conservative dynamics differ significantly. The KZ policy is overly conservative, exhibiting a precipitous drop in assignments due to its reliance on the loose lower bound on policy improvement from \cite{zhao2019sensitivity}. In contrast, the MMI policy is precisely conservative;  {by identifying the sharp lower bound of the CATE}, it establishes a sharper bound on policy improvement, leading to a more calibrated and gradual reduction in treatment. Finally, the MMW policy's aggressiveness under high uncertainty is a direct consequence of its decision rule, illustrated in \cref{section: Identification and Estimation of the Worst-case Welfare Function}, which favors treatment whenever  $\mu_{\Lambda}^{-}(x,1)$ exceeds $\mu^{-}_\Lambda(x,0)$.

\Cref{fig:fix_dgplambda_average_welfare} shows the resulting expected welfare. The MMI policy achieves its peak expected welfare at $\log \Lambda = 1$. It outperforms the MMW policy for assumed confounding levels below the true value of $\log \Lambda^* = 1.5$. In regimes of greater assumed uncertainty $\log \Lambda > 1.5$, however, the MMW policy's performance is superior, yielding significantly higher welfare. Their performance crossover occurs precisely at this true value, confirming their complementary nature. This divergence occurs because, in high-uncertainty regimes, the MMI policy becomes overly cautious, whereas the MMW policy adopts a more risk-tolerant strategy.

Finally, \Cref{fig:fix_dgplambda_worst_welfare,fig:fix_dgplambda_worst_improvement}  validate that the observed behaviors are consequences of each policy's specific design. As expected, \Cref{fig:fix_dgplambda_worst_welfare} shows the MMW policy achieving the highest worst-case welfare, and \Cref{fig:fix_dgplambda_worst_improvement} confirms the MMI policy's success in maximizing the worst-case policy improvement. This confirms that each policy effectively optimizes its intended objective, leading to the distinct and complementary performance profiles we identified.

\begin{figure}[H]
    \centering
    \begin{subfigure}[t]{0.48\textwidth}
        \centering
        \includegraphics[width=\linewidth]{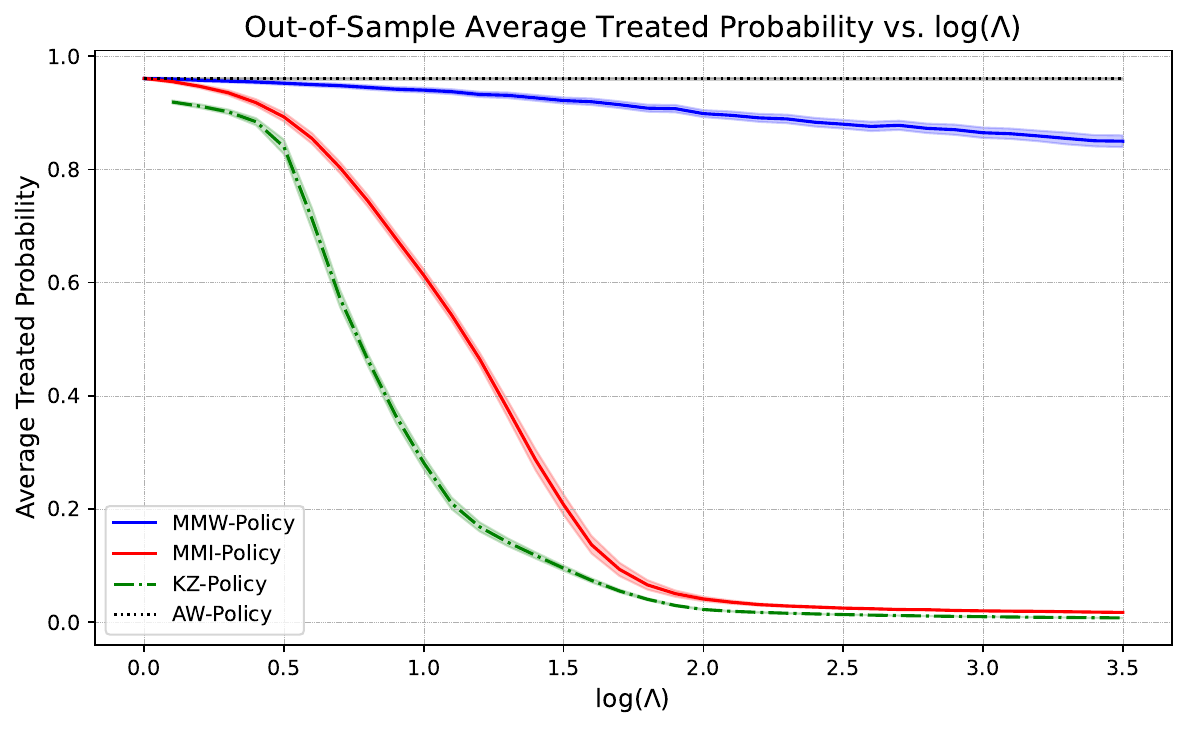}
        \caption{Average Treatment Probability}
        \label{fig:fix_dgplambda_treated_probs}
    \end{subfigure}
    \hfill
    \begin{subfigure}[t]{0.48\textwidth}
        \centering
        \includegraphics[width=\linewidth]{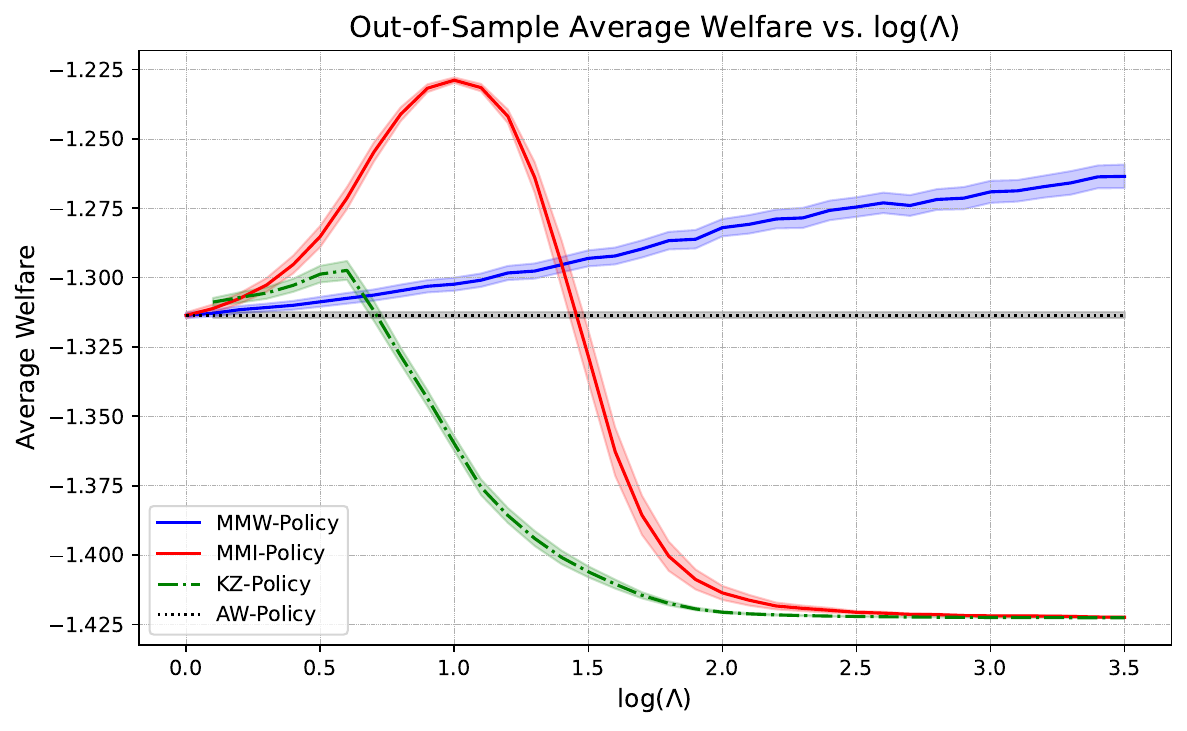}
        \caption{Expected Welfare}
        \label{fig:fix_dgplambda_average_welfare}
    \end{subfigure}

    \caption{Average treatment probability and expected welfare as a function of the sensitivity parameter $\log(\Lambda)$. The MMI and MMW policies are compared against the KZ and AW benchmarks. Shaded areas are 95\% confidence bands.}
\end{figure}

\begin{figure}[H]
    \centering
    \begin{subfigure}[b]{0.48\linewidth}
        \centering
        \includegraphics[width=\linewidth]{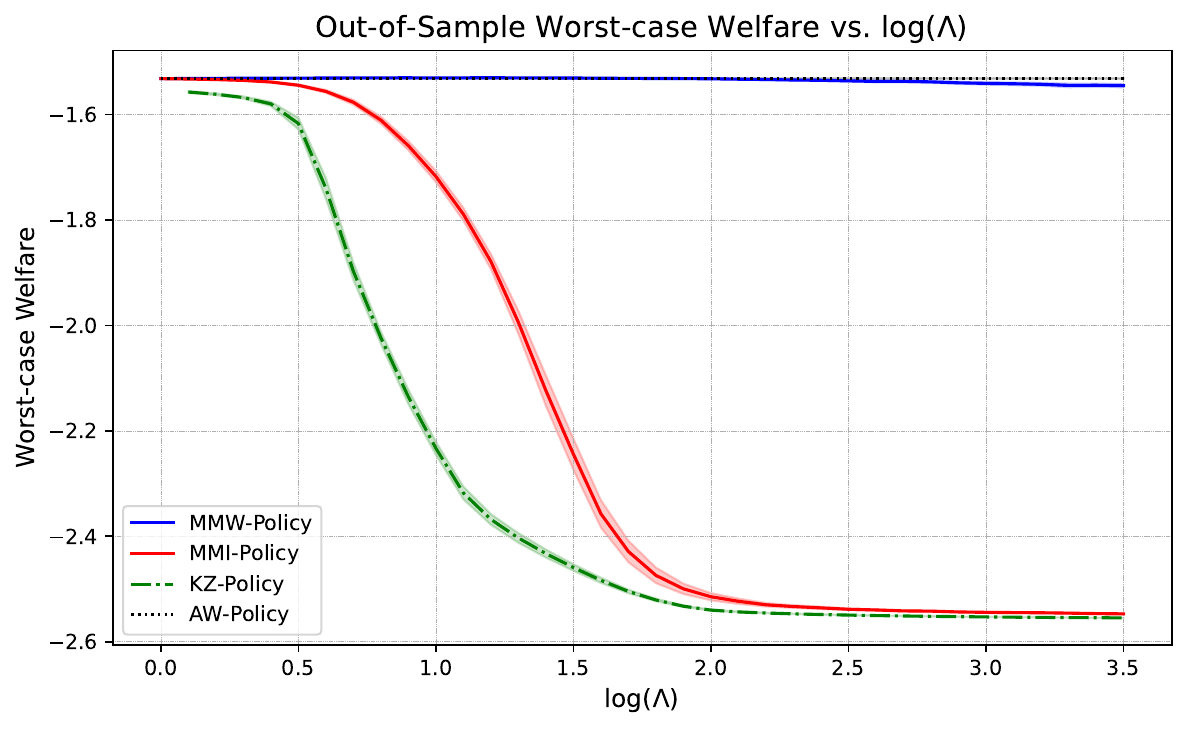}
        \caption{Worst-case Welfare}
        \label{fig:fix_dgplambda_worst_welfare}
    \end{subfigure}
    \hfill
    \begin{subfigure}[b]{0.48\linewidth}
        \centering
        \includegraphics[width=\linewidth]{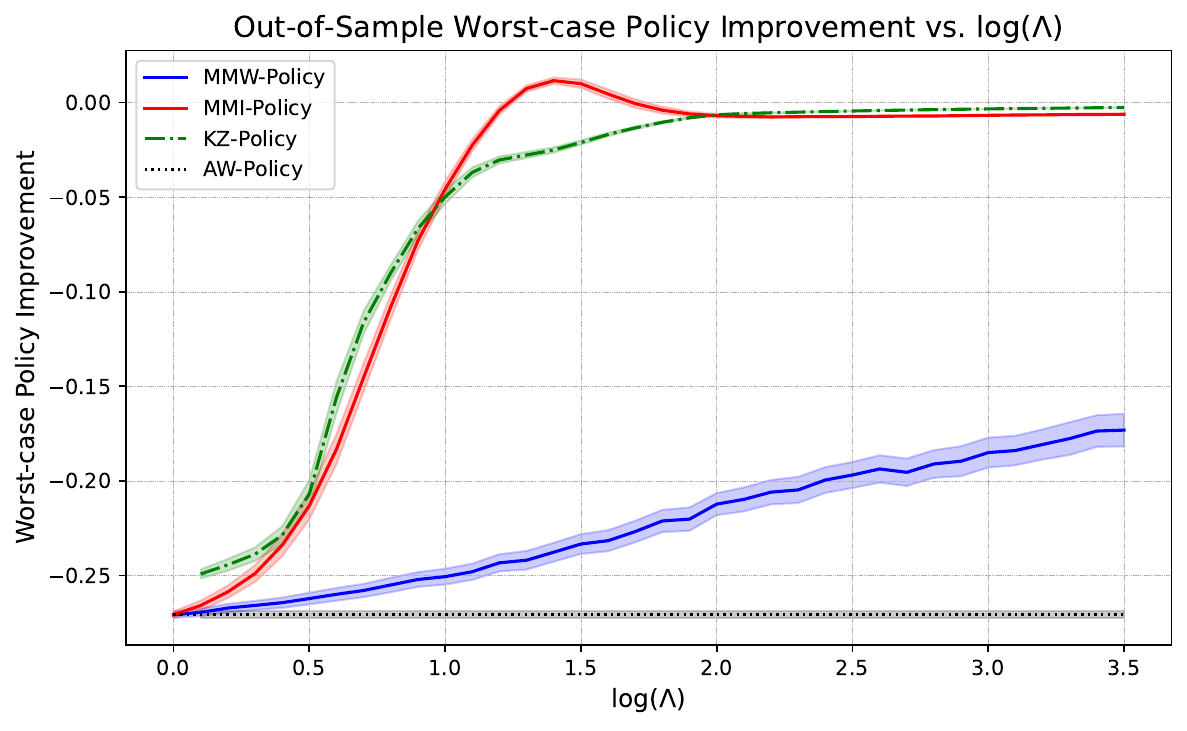}
        \caption{Worst-case Policy Improvement}
        \label{fig:fix_dgplambda_worst_improvement}
    \end{subfigure}
    \caption{ Performance on Intended Objectives. The figures plot (a) the MMW policy's objective of worst-case welfare and (b) the MMI policy's objective of worst-case policy improvement against $\log(\Lambda)$. Shaded areas denote 95\% confidence bands.}
\end{figure}

\section{Empirical Application}\label{section: empirical}

In this section, we apply our method to two empirical applications. \cref{section: JTPA} revisits the JTPA study and examines assignment to job training programs, while \cref{section: headstart} explores enrollment decisions in the Head Start program. In both settings, unobserved confounders likely give rise to endogenous selection.

\subsection{The JTPA Study}\label{section: JTPA}
We apply the MMW and MMI policy learning methods to the National Job Training Partnership Act (JTPA) Study, a large-scale randomized controlled trial evaluating training services for disadvantaged adults \citep{bloom1997benefits}. In the experiment, eligibility for services was randomly assigned, and participants' earnings were tracked over the 30-month period after the assignment. This dataset was  used by \cite{kitagawa2018a} to estimate an intent-to-treat optimal policy for 9,223 applicants using their Empirical Welfare Maximization (EWM) approach, based on covariates such as education and prior earnings.

While eligibility was randomly assigned, compliance in the JTPA study was imperfect: approximately 23\% of individuals did not adhere to their assigned eligibility status, as shown in \cref{tab:jtpa_compliance}. This imperfect compliance introduces self-selection into actual participation decisions, raising concerns about unobserved confounding. For instance, eligible individuals who opt into the program may exhibit stronger job search motivation or face fewer personal constraints, such as childcare responsibilities, health issues, or transportation barriers, than those who decline participation. These unobserved factors likely influence both program participation and post-program earnings, violating the unconfoundedness assumption required by conventional policy learning methods.

\begin{table}[ht]
\centering
\caption{Joint distribution of eligibility and participation, JTPA study}
\label{tab:jtpa_compliance}
\begin{tabular}{lcccc}
\toprule
 & \multicolumn{3}{c}{Eligibility (IV)} & \\
\cmidrule(lr){2-4}
Participation ($A_i$) & 0 & 1 & Total \\
\midrule
0 & 3047 & 2118 & 5165 \\
1 & 43 & 4015 & 4058 \\
\midrule
Total & 3090 & 6133 &  {9223} \\
\bottomrule
\end{tabular}
\vspace{0.5em}
\footnotesize
\\
Data source: \cite{kitagawa2018a} and \citet{abadie2002instrumental}.
\end{table}

We focus on actual program participation, as the policy-relevant objective is to target individuals who are most likely to benefit from job training and to allocate participation accordingly. Building on the perspective of \cite{d2021orthogonal}, who frame policy learning as deriving optimal treatment rules under potential unobserved confounding, we apply the MMW and MMI methods, both of which are explicitly designed to address such confounding. Unlike \cite{d2021orthogonal}, who use a binary instrumental variable (eligibility) to construct partial identification intervals for the CATE, we adopt a complementary strategy by employing with a pre-specified sensitivity parameter $\Lambda$.

\paragraph{Calibration of Sensitivity Parameter.}
 {We calibrate the  $\Lambda$ following the breakdown frontier perspective of \citet{masten2020inference}. The IV bounds discussed in \citet{d2021orthogonal} provide a  credible benchmark for this purpose. 
Let $\tau_{\mathrm{IV}}^{-}(x)$ be the  lower bound for CATE $\tau(x)$, and $\mu^-_{\mathrm{IV}}(x,a)$ be the  lower bound for the conditional response function $\mu(x,a)$, constructed using IV.\footnote{IV-based bounds can be constructed under different sets of identification  assumptions. Following \citet{d2021orthogonal}, we use the Balke-Pearl bounds as the IV-based benchmark in this application.}
We  select $\Lambda$ as the smallest relaxation level at which the MSM lower bounds are at least as conservative as the IV lower bounds.
In practice, we select $\Lambda$ based on the coverage share, that is, the fraction of sample observations for which the MSM-implied lower bound is no larger than the IV-based lower bound.
For the MMI policy, the coverage share is computed using $\tau^{-}_{\Lambda}(x)$ and $\tau^{-}_{\mathrm{IV}}(x)$. We select the smallest value of $\Lambda$ for which this coverage share reaches the pre-specified threshold. For the MMW policy, we apply the same rule to $\mu^{-}_{\Lambda}(x,a)$ and $\mu^{-}_{\mathrm{IV}}(x,a)$ for $a\in{0,1}$.  As shown in \Cref{fig:lambda_calibration_mmw_mmi_combined}, targeting an 80\% coverage share yields $\Lambda = 2.5$ for the MMI policy and $\Lambda = 4.5$ for the MMW policy.}

\begin{figure}[htb]
    \centering    \includegraphics[width=0.95\linewidth]{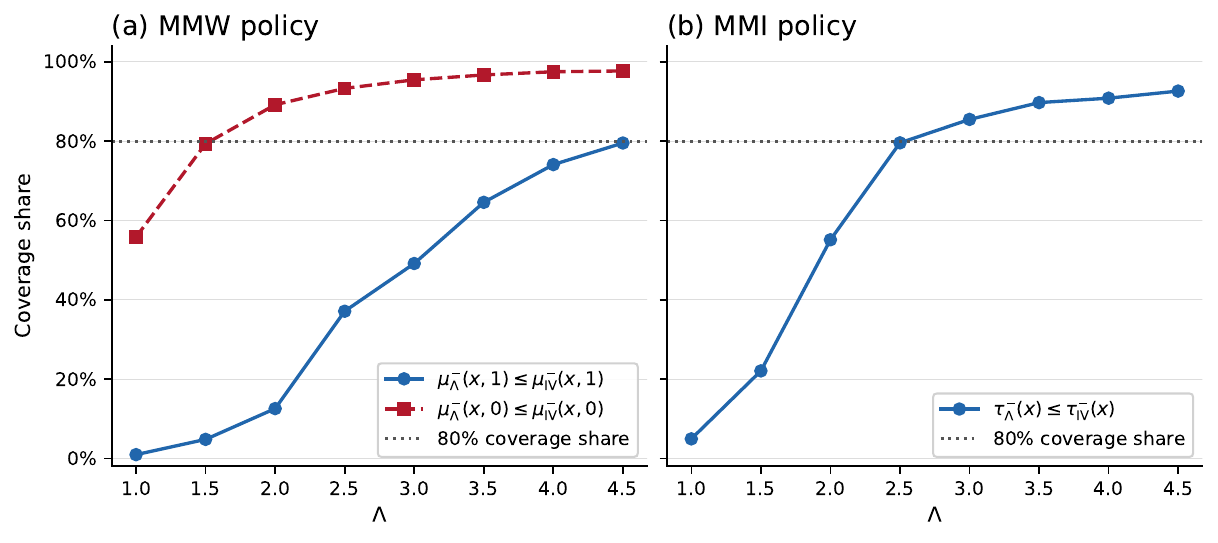}
    \caption{ {Calibration of the sensitivity parameter $\Lambda$ against the IV-based lower bounds. Panel (a) plots, as functions of $\Lambda$,  the coverage shares $n^{-1}\sum_{i=1}^{n}
\mathds{1}\left\{
\mu^{-}_{\Lambda}(x_i,a)
\le
\mu^{-}_{\mathrm{IV}}(x_i,a)
\right\}$ for $a \in \{0,1\}$. Panel (b) plots the coverage share $n^{-1}\sum_{i=1}^{n}
\mathds{1}\left \{
\tau^{-}_{\Lambda}(x_i)
\le
\tau^{-}_{\mathrm{IV}}(x_i)
\right\}$. The plotted coverage shares are computed using the estimated MSM and IV lower bounds. The horizontal dotted line marks the $80\%$ target threshold, attained at $\Lambda = 4.5$ for the MMW policy, determined by the treated-arm comparison, and $\Lambda = 2.5$ for the MMI policy.}
    }    \label{fig:lambda_calibration_mmw_mmi_combined}
\end{figure}

 We estimate the optimal MMW and MMI policies using Algorithms~\ref{algo:MMW_PL} and~\ref{algo:MMI_PL}, with $K=10$ folds for cross-fitting. The conditional quantile functions $q_{ {\Lambda}}^{\pm}$ are estimated using gradient boosted trees, and the nominal propensity score $e$ and CVaR $\rho_{t,  {\Lambda}}^{\pm}$ are estimated using random forests.  {The nominal propensity score estimates range from $0.27$ to $0.73$, so we retain the full sample without trimming.} To account for program costs, we subtract \$1,216 from the outcomes of treated individuals. This amount corresponds to the average cost of services per actual participant, as reported in Table 5 of~\citet{bloom1997benefits}. Following the approach of~\cite{kitagawa2018a} and~\citet{d2021orthogonal}, we consider the class of quadrant treatment policies due to its simplicity and interpretability:
\[
\Pi \equiv \left\{  \mathds{1} \left\{ s_1(\text{edu} - t_1) > 0, \; s_2(\text{earnings} - t_2) > 0 \right\} : s_1, s_2 \in \{-1, 1\}, \; t_1, t_2 \in \mathbb{R} \right\}.
\]
 {
Given the cross-fitted doubly robust scores, we maximize the estimated robust criterion by searching over all values of $(s_1,s_2)$ and all threshold pairs $(t_1,t_2)$ induced by the observed values of education and pre-program earnings, following the implementation in \citet{kitagawa2018a}.}

\begin{figure}[htb]
    \centering
    \includegraphics[width=0.95\linewidth]{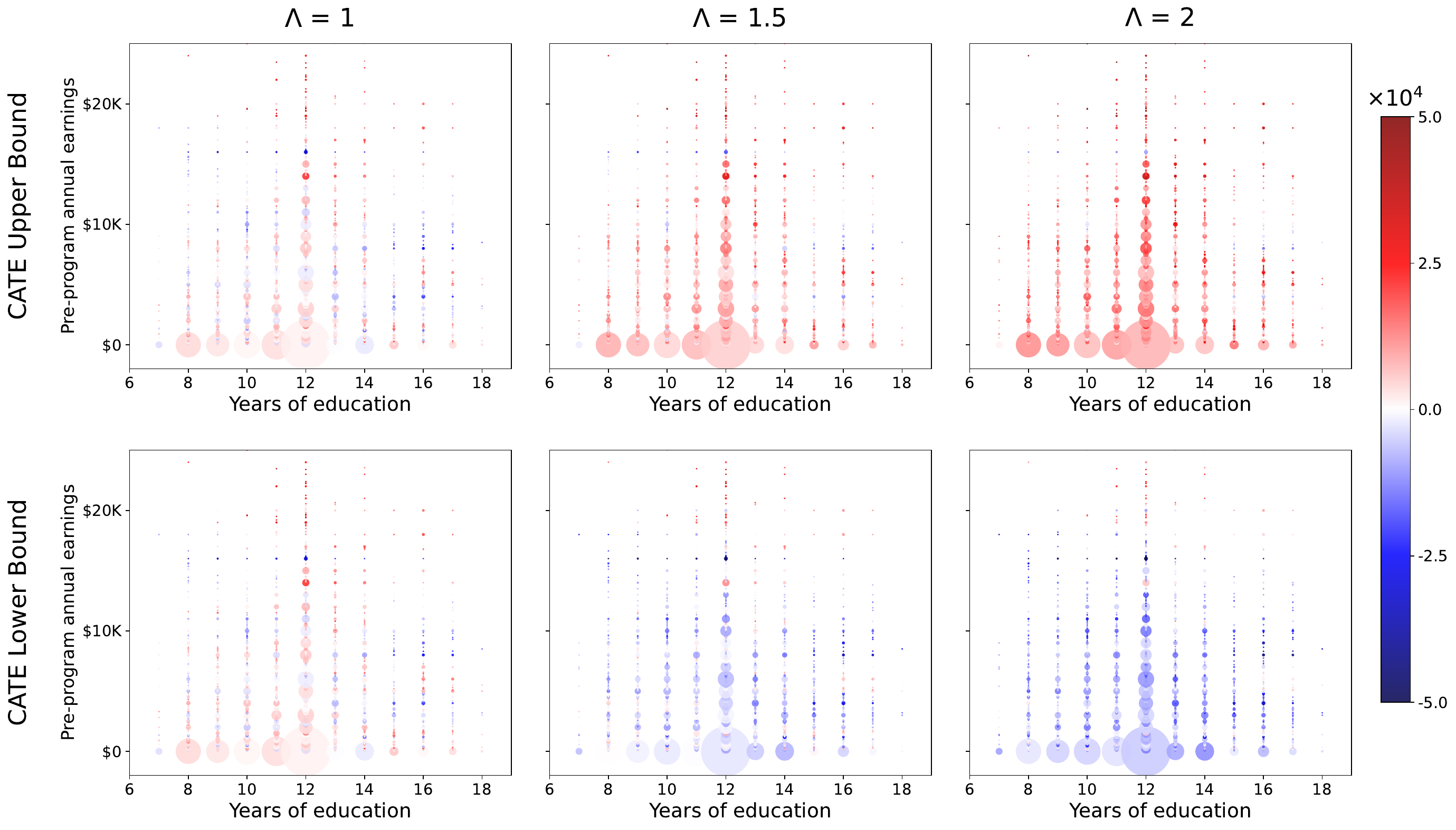}
    \caption{
        Upper and lower bounds of the CATE under uncertainty levels $\Lambda = 1, 1.5, 2$.
    }
    \label{fig:CATE_hi_lo_bounds}
\end{figure}

\Cref{fig:CATE_hi_lo_bounds} presents the estimated upper and lower bounds of the CATE under the MSM framework, evaluated across different values of the sensitivity parameter $\Lambda \in \{1, 1.5, 2\}$. When $\Lambda = 1$,  the estimated upper and lower bounds coincide. Choosing a larger value of $\Lambda$ can yield more robust policy targeting by ensuring that the identified set contains the true causal effect. However, this robustness comes at the cost of producing less informative estimates and more conservative policy intervention. Therefore, we recommend that practitioners examine results across a range of $\Lambda$ values to assess the sensitivity of their conclusions to unobserved confounding.

\Cref{fig:MMW_MMI_quadrant} illustrates the estimated optimal quadrant policies for MMW and MMI {, using the IV minimax-regret (MMR) policy of \cite{d2021orthogonal} as a benchmark}. When $\Lambda = 1$, both the MMW and MMI policies  {are} exactly the AW policy, but with a treated fraction approaching one.  As sensitivity parameter $\Lambda$ increases, the MMI policy becomes increasingly conservative, with the treatment rate dropping rapidly to near zero when $\Lambda \geq 1.5$. The extremely low treatment rates under the MMI policy arise because the estimated lower bounds of the CATE fall below zero for nearly all individuals when $\Lambda \geq 1.5$. Consequently, our subsequent analysis focuses on the MMW policy. When $\Lambda = 1$, the quadrant policy targets individuals with  {more than 7.5 years of education} and pre-program earnings below \$40{,}800. In contrast, at $\Lambda = 2$, the MMW quadrant policy assigns treatment to those with at least 10.5 years of education and earnings below \$35{,}144. As the sensitivity parameter $\Lambda$ increases, the estimated MMW policy becomes more selective, prioritizing individuals with higher educational level and lower earnings. {\footnote{The MMW policy becomes excessively conservative when $\Lambda \geq 2.5$. The treatment rate drops to roughly 4.6\%, which limits its usefulness for practical targeting.}}  {In sharp contrast, the IV MMR policy assigns treatment to roughly 96\% of the sample. Such an excessively high treatment rate fails to provide meaningful targeting guidance, offering little practical relevance for selective policy design.}

\begin{figure}[htb]
	\centering
	\includegraphics[width=0.95\linewidth]{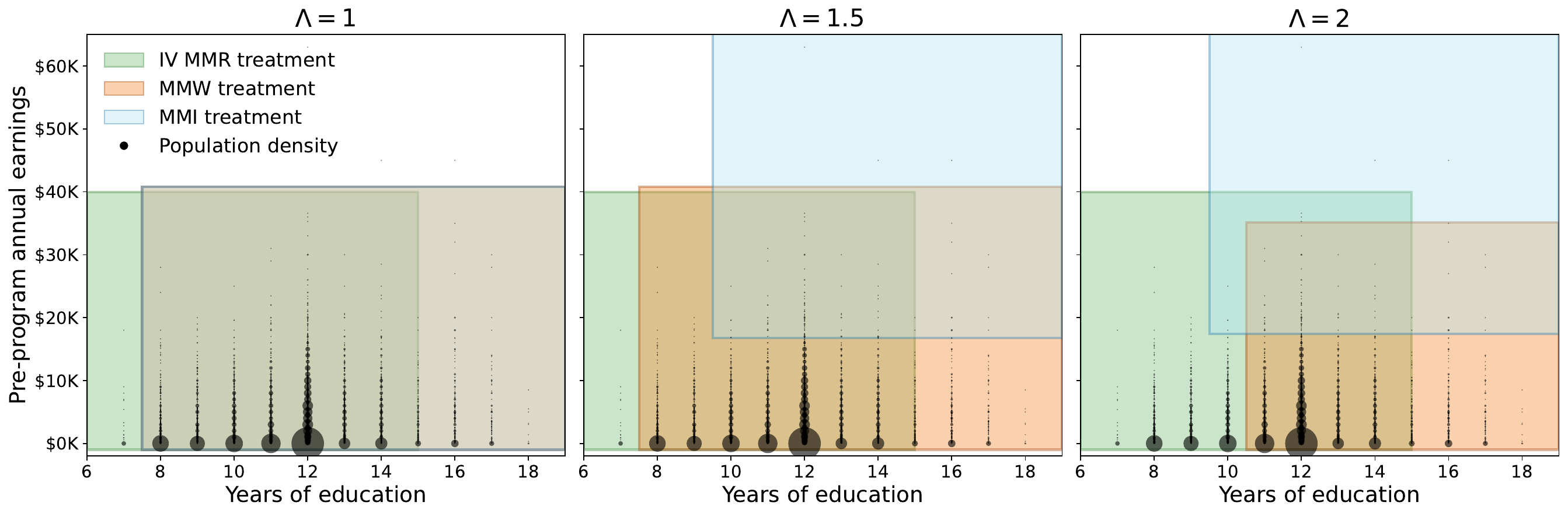}
	\caption{Optimal MMW and MMI quadrant policies under different uncertainty levels: $\Lambda = 1, 1.5$, $2$ {, compared with the $\Lambda$-invariant IV MMR policy}. }
	\label{fig:MMW_MMI_quadrant}
\end{figure}


\subsection{The Head Start Program}\label{section: headstart}
We further demonstrate our confounding-robust policy learning method by applying it to enrollment decisions for the Head Start program. The objective is to develop a policy that improves academic outcomes by targeting admission to children most likely to benefit. This application presents a multi-valued treatment setting with three alternatives, Head Start, other preschools, or no preschool, necessitating the use of our MMW approach extended for multiple treatments as detailed in \cref{appendix: discrete treatment}.

Head Start is a federal matching grant initiative designed to foster the cognitive, social, and health development of children from low-income families, with the aim of preparing them to enter school on a more equal footing with their more advantaged peers. A substantial body of research has evaluated the impact of Head Start participation on outcomes such as academic achievement, physical health, and long-term socioeconomic status; see, e.g., \citep{currie1993does, deming2009early, ludwig2007does, walters2015inputs}.

A key challenge in evaluating Head Start's impact is that enrollment is not randomly assigned. First, families self-select into Head Start or other preschool programs based on eligibility and perceived benefits. Those who value early education may also make unobserved investments in their children, such as fostering enriched home environments or providing supplemental learning. Second, limited program capacity in many areas gives administrators discretion over admissions, leading to selection based on both observed and unobserved characteristics, some of which also affect child outcomes. As a result, failing to account for such confounding can lead to targeting policies that misidentify the children most likely to benefit.

We draw on data from the National Longitudinal Survey of Youth (NLSY) and its child supplement, the National Longitudinal Survey of Youth 1979 Children and Young Adults (NLSCYA). The NLSY began in 1979 with a nationally representative sample of young women, and the NLSCYA has tracked their children since 1986. Our analysis sample is constructed by linking child-level outcomes from the NLSCYA with maternal background characteristics from the original NLSY cohort.

We use the child's standardized score on the Peabody Picture Vocabulary Test (PPVT), a widely used measure of early verbal ability, as a proxy for our policy objective of improving academic achievement. The pre-treatment covariates used in our analysis capture a range of child, maternal, and household characteristics. These include: (1) the percentile rank of average net family income between ages 0 and 4 (\textit{Income Pctl}); (2) child's gender  {(\textit{Gender})}; (3) birth weight (\textit{Birth Wt}); (4) weight at preschool entry (\textit{Preschool Entry Wt}); (5) firstborn status  {(\textit{Firstborn})}; (6) mother's highest completed grade  {(\textit{Mother's Grade})}; (7) mother's AFQT percentile score  {(\textit{Mother's AFQT})}; (8) number of the mother's biological siblings  {(\textit{Bio.\ Siblings})}; (9-10) number of household members with less than 12 years  {(\textit{HH Edu ${<}12$})} and at least 16 years of education  {(\textit{HH Edu ${\geq}16$})}; and (11) race  {(\textit{White}, \textit{Hispanic}, \textit{Black})}. After excluding observations with missing data, the final sample consists of 3,826 children: 755 attended Head Start, 2,020 enrolled in other preschool programs, and 1,051 did not attend any preschool. 

\paragraph{Calibration of Sensitivity Parameter.}  {Since valid instruments are unavailable in the Head Start study, we cannot calibrate the sensitivity parameter $\Lambda$ using the IV-based procedure in \cref{section: JTPA}. Following \citet{hsu2013calibrating} and \cite{kallus2021minimax}, we instead calibrate $\Lambda$ using observed covariates by computing the effect of omitting each observed
covariate on the odds ratio of the propensity score. Specifically, for each treatment arm $a\in\{0,1,2\}$,  we estimate two propensity score models using random forests: a full model $\widehat{e}_{a}(x)$ that includes all covariates, and a leave-one-out model $\widehat e_a^{-j}(x_{-j})$ that omits the $j$-th covariate. The individual-level calibrated sensitivity parameter associated with omitting the $j$-th covariate in arm $a$ is defined as
\[
    \Lambda_{j,a}(x)
    = \max\!\left(
        \frac{\widehat{e}_{a}(x)/(1-\widehat{e}_{a}(x))}
             {\widehat e_a^{-j}(x_{-j})/(1-\widehat e_a^{-j}(x_{-j}))},\;
        \frac{\widehat e_a^{-j}(x_{-j})/(1-\widehat e_a^{-j}(x_{-j}))}
             {\widehat{e}_{a}(x)/(1-\widehat{e}_{a}(x))}
    \right).
\]
We take the 95th percentile of $\Lambda_{j,a}(X_i)$ for $i \in [n]$ as the calibrated value associated with omitting covariate $j$ for arm $a$, denoted as $\Lambda_{j,a}$. We then compute the covariate-specific calibrated sensitivity parameter as $\Lambda_{j} = \max_a \Lambda_{j,a}$, that is reported in \Cref{tab: lambda_heatmap_P95}. Finally, taking the maximum across covariates gives the overall calibrated sensitivity parameter $\Lambda = 3.6$. The largest covariate-specific value is associated with the family income percentile.}

\begin{table}[ht]
\centering
\caption{ {Calibrated sensitivity parameters for each covariate (no propensity truncation).}}
\label{tab: lambda_heatmap_P95}
\resizebox{\textwidth}{!}{%
\begin{tabular}{l*{12}{>{\centering\arraybackslash}p{1.5cm}}}
\toprule
 & \makecell{\textit{Firstborn}}
 & \makecell{\textit{Birth}\\\textit{Wt}}
 & \makecell{\textit{Preschool}\\\textit{Entry Wt}}
 & \makecell{\textit{Gender}}
 & \makecell{\textit{Mother's}\\\textit{AFQT}}
 & \makecell{\textit{Mother's}\\\textit{Grade}}
 & \makecell{\textit{Bio.}\\\textit{Siblings}}
 & \makecell{\textit{HH Edu}\\${\geq}16$}
 & \makecell{\textit{HH Edu}\\${<}12$}
 & \makecell{\textit{Income}\\\textit{Pctl}}
 & \makecell{\textit{Hispanic}}
 & \makecell{\textit{Black}} \\
\midrule
$\Lambda_{j}$
 & 1.746 & 1.802 & 1.833 & 1.485 & 2.550
 & 1.655 & 1.718 & 1.683 & 1.576 & 3.596
 & 1.691 & 2.049 \\
\bottomrule
\end{tabular}%
}
\end{table}


 {We estimate the nuisance functions following the same estimation procedures used in  Sections \ref{section:simulation-analysis} and \ref{section: JTPA}. The initial propensity score estimates range from \(0.0005\) to \(0.95\), so we truncate them to \([\alpha,1-\alpha]\) with \(\alpha=0.05\) to enforce the strict overlap condition.}  While the full set of pre-treatment covariates is used to estimate doubly robust scores, we restrict the variables used for policy optimization to a small, interpretable subset: percentile rank of family income, mother's AFQT score, birth weight, and weight at preschool entry. These are selected based on fairness and ethical considerations.  
To ensure interpretability, we restrict the policy class to depth-2 decision trees,  {using  \texttt{R}-package \texttt{policytree}.}

The policy tree estimated under $\Lambda = 1$ serves as our benchmark and
reproduces the recommendation that would be obtained by methods such as
\citep{athey2021policy,zhou2023offline}, which assume no unobserved confounding. An analysis of this benchmark policy reveals several features that present challenges for direct implementation.  {First, it completely excludes Head Start, contradicting the program's mandate to serve disadvantaged households. Second, it misallocates these economically vulnerable households to other tuition-charging preschools rather than the subsidized public program.} This motivates our subsequent sensitivity analysis to assess how the policy recommendations change when accounting for unobserved confounding.  {Once we account for unobserved confounding, the policy recommendation changes markedly. Specifically, at the calibrated value $\Lambda = 3.6$, the robust policy targets Head Start toward a clearly defined disadvantaged group, while taking a more conservative non-intervention stance for the majority.\footnote{This targeted recommendation is stable as \(\Lambda\) varies over \(\{1.5,2,\ldots,4\}\) and is insensitive to the choice of the propensity score truncation level \(\alpha\).}}

\begin{figure}[htb]
    \centering
    \includegraphics[width=\linewidth]{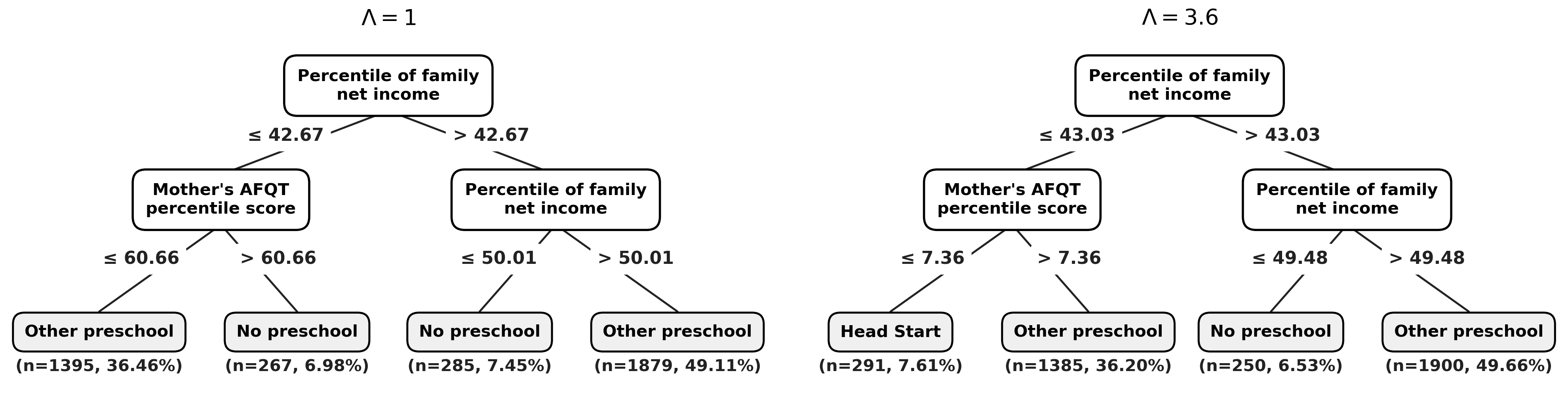}
    \caption{ {This figure presents optimal policy trees under unconfoundedness and calibrated value $\Lambda=3.6$. While the unconfoundedness benchmark completely excludes Head Start, incorporating unobserved confounding enables the optimal policy to effectively target this program toward the disadvantaged subgroup.}
    }
    \label{fig:policy_tree_lam1to2dot5}
\end{figure}



\section{Concluding Remarks}

Formulating effective policies from observational data is challenged by unobserved confounding, which can lead to biased and potentially harmful policy intervention. This paper proposes a robust and efficient framework for policy learning that explicitly accounts for such endogeneity. Our contributions are twofold: first, we derive sharp, closed-form identification results for robust welfare criteria under the Marginal Sensitivity Model; second, we develop doubly robust scores for their estimation. This innovation enables the use of flexible  {ML} methods for nuisance estimation within our framework. Building on these components, we establish regret bounds for the resulting confounding-robust policies.

Our work also offers insights into the practical challenges of applying policy learning methods. While powerful, ML-based approaches that optimize a single objective can yield policies that are difficult to trust and interpret, particularly when built on fragile assumptions. Our work underscores that sensitivity analysis is a critical step for examining a policy's foundational reliability. 

Finally, several directions remain open for future research. First, while our framework primarily focuses on binary treatments, extending it to multivalued or continuous treatments is a promising direction that would broaden its applicability and introduce new methodological challenges. Second, the choice of the sensitivity parameter 
$\Lambda$ remains an important yet unresolved issue.  {Although our empirical analysis illustrates two calibration strategies, namely observed-covariate calibration following \citet{hsu2013calibrating} and IV-based calibration, the choice of $\Lambda$ in a given application may still require substantive domain knowledge. Developing more systematic, data-driven methods for selecting $\Lambda$ remains a key avenue for future investigation.}

\begin{appendices}

\section{ {Policy Learning with Self-Selection as an Option}}\label{appendix: self-selection}

Most policy-learning work studies compulsory treatment rules, using observable characteristics to mandate program participation or exclusion \citep{kitagawa2018a,athey2021policy}.  However, individuals often possess private information, unobserved by the social planner, regarding their potential benefits from the program. So, self-selection serves as a valuable source of information for policy targeting \citep{manski2013public,alatas2016self,ito2023selection}. Building on these insights, \citet{ida2022choosing} integrate targeting by observables with targeting through self-selection. Specifically, they use (quasi)-experimental data to assign individuals to one of three arms: compulsory treatment, compulsory un-treatment, or self-selection, fulfilling the policymaker's objective.

\cref{section: Identification} studies robust policy learning under unobserved confounding, focusing on compulsory intervention. This section extends this framework to incorporate self-selection as a policy option. In observational studies, unmeasured confounding makes it harder to assess precisely the welfare gains from counterfactual compulsory assignments. Therefore, absent robust evidence that mandated assignments improve welfare, preserving the status quo self-selection is a natural policy option.

Motivated by this, we formalize the planner's goal as assigning each individual to one of three arms: compulsorily treated (indexed as $1$), compulsorily untreated (indexed as $0$), and self-selection (indexed as $\mathrm{S}$). An individual assigned to $1$ or $0$ is exposed to or excluded from the program with no opt-out or opt-in option, whereas an individual assigned to $S$ retains the autonomy to choose whether to take it up. For clarity, we focus on deterministic policies mapping $\mathcal{X}$ to $\{0,1,\mathrm{S}\}$ in this section. Let $\Pi^\dagger_o$ denote the deterministic policy class incorporating the self-selection option 
\[
\Pi^\dagger_o =  \left\{ \pi: \mathcal{X}\rightarrow \{0,1,\mathrm{S}\} \text{ is Borel measurable}  \right \}.
\]
Since the observational data reflects the status quo of self-selection, the counterfactual outcome under arm $\mathrm{S}$ coincides with the realized outcome $Y = A Y(1) + (1-A)Y(0)$.  For any policy $\pi^\dagger \in \Pi^\dagger_o $, the resulting post-treatment outcome $Y(\pi^\dagger)$ can be expressed as:
\begin{equation}\label{eq:post-treatment outcome}
\begin{aligned}
Y(\pi^\dagger) &= \mathds{1}\{ \pi^\dagger(X) = 1 \} Y(1) + \mathds{1}\{ \pi^\dagger(X) = 0 \} Y(0) + \mathds{1}\{ \pi^\dagger(X) = \mathrm{S} \} Y\\
& = \left[  \mathds{1}\{ \pi^\dagger(X) = 1 \}  + A \mathds{1}\{ \pi^\dagger(X) = \mathrm{S} \}   \right] Y(1) \\
& +  \left[  \mathds{1}\{ \pi^\dagger(X) = 0 \}  + (1- A) \mathds{1}\{ \pi^\dagger(X) = \mathrm{S} \}   \right] Y(0).
\end{aligned}
\end{equation}

We now generalize the MMW and MMI criteria from \cref{section: Policy Learning beyond Unconfoundedness} to an augmented policy class $\Pi^\dagger_o$ that incorporates self-selection. The  worst-case welfare $W_{\Lambda}$ and policy improvement $\Delta_{\Lambda}$ functions are defined as:
\begin{align*}
W_{\Lambda}(\pi^\dagger) = \inf_{Q \in \mathcal{P}(\Lambda)} \mathbb{E}_Q \big[ Y(\pi^\dagger) \big] \quad \text{and} \quad
\Delta_{\Lambda}(\pi^\dagger) = \inf_{Q \in \mathcal{P}(\Lambda)} \mathbb{E}_Q \big[ Y(\pi^\dagger) - Y(0) \big].
\end{align*}
Over the policy class $\Pi_o^\dagger$, the first best MMW and  MMI policies are given, respectively, by
\begin{subequations}\label{eq:AMMW-AMMI-policy}
\begin{align}
  \pi^\dagger_{W,\Lambda}
  &\in
  \underset{\pi^\dagger \in \Pi^\dagger_o}{\argmax}
  \ W_{\Lambda}(\pi^\dagger),
  \label{eq:AMMW-policy}\\
  \pi^\dagger_{\Delta,\Lambda}
  &\in
  \underset{\pi^\dagger \in \Pi^\dagger_o}{\argmax}
  \ \Delta_{\Lambda}(\pi^\dagger).
  \label{eq:AMMI-policy}
\end{align}
\end{subequations}

\cref{thm:AMMW} characterizes the worst-case welfare when self-selection is incorporated as a feasible policy option, subsequently deriving the first-best MMW policy. To this end, recall the notation $\mu_{\Lambda}^{\pm}(X, a)$ for $a \in \{0,1\}$ as defined in \cref{section: Identification of Robust Criteria}. Moreover, we define the worst-case conditional average treatment effect on the untreated (ATU) and the best-case conditional average treatment effect on the treated (ATT):
\[
\begin{aligned}
\mathrm{ATU}^-_{\Lambda}(x) & \equiv \inf_{Q \in \mathcal{P}(\Lambda)} \mathbb{E}_Q \left[ Y(1)- Y(0) |X=x, A =0 \right],\\
\mathrm{ATT}^+_{\Lambda}(x)  & \equiv \sup_{Q \in \mathcal{P}(\Lambda)} \mathbb{E}_Q \left[ Y(1)- Y(0) |X=x, A =1 \right].
\end{aligned}
\]
These quantities are point-identified  through the following closed-form relations:
\[
\begin{aligned}
\mathrm{ATU}^-_{\Lambda}(x)  = \frac{\mu_{\Lambda}^{-}(x, 1) - \mathbb{E}[Y|X=x]}{1 - e(x)} \ \ \  \text{and} \ \ \
\mathrm{ATT}^+_{\Lambda}(x) = \frac{\mathbb{E}[Y|X=x] - \mu_{\Lambda}^{-}(x, 0)}{e(x)}.
\end{aligned}
\]

\begin{theorem}\label{thm:AMMW}
 Under \cref{assumption: Marginal sensitivity model}, for any policy $\pi^\dagger \in \Pi^\dagger_o$, the worst-case welfare is given by
 \[
 \begin{aligned}
 W_{\Lambda}(\pi^\dagger) & = \mathbb{E} \left[ \mathds{1}\{ \pi^\dagger(X) = 1 \} \mu_{\Lambda}^{-} (X,1) + \mathds{1}\{ \pi^\dagger(X) = 0 \} \mu_{\Lambda}^{-} (X,0)\right]\\
 &+ \mathbb{E}\left[ \mathds{1}\{ \pi^\dagger(X) = \mathrm{S} \} \mathbb{E}[ Y| X ] \right].
 \end{aligned}
 \]
 Moreover, a first-best MMW policy that solves \cref{eq:AMMW-policy} is given by
\begin{equation}\label{eq:MMW_SS}
 \pi_{W,\Lambda}^{\dagger} (x) = \begin{cases}
1, & \text{if}\ \ \mu_{\Lambda}^{-} (x,1) > \mu_{\Lambda}^{-} (x,0)  \ \ \text{and}\ \ \mathrm{ATU}^-_{\Lambda}(x) >0, \\
0, & \text{if}\ \ \mu_{\Lambda}^{-} (x,0) > \mu_{\Lambda}^{-} (x,1)  \ \ \text{and}\ \ \mathrm{ATT}^+_{\Lambda}(x)  < 0, \\
\mathrm{S}, &  \text{if}\ \  \mathrm{ATU}^-_{\Lambda}(x)  \leq 0  \ \ \text{and}\ \ \mathrm{ATT}^+_{\Lambda}(x)  \geq 0   . \\
\end{cases}
\end{equation}
\end{theorem}

\begin{remark}\label{remark:selection_driven_targeting}
\cref{thm:AMMW} establishes a confounding-robust selection-driven targeting rule analogous to \cite{ida2022choosing}. The social planner assigns individuals to  $\{0,1,\mathrm{S}\}$, based on robust, subpopulation-specific treatment effects. Compared to the selection-absent MMW policy in \cref{theorem: max-min welfare assignment rule}, the selection-driven MMW policy in \cref{thm:AMMW} establishes a higher threshold for intervention. 

Specifically, compulsory treatment is assigned only when $\mathrm{ATU}_{\Lambda}^{-}(x)>0$, ensuring a guaranteed benefit even for non-takers under the most adversarial confounding.\footnote{Takers (non-takers) refer to individuals who currently choose to participate in (opt out of) the treatment under the status quo self-selection regime.} Conversely, compulsory un-treatment is assigned only when $\mathrm{ATT}_{\Lambda}^{+}(x) < 0$, implying that treatment may be harmful for takers even under the most favorable confounding scenario. When neither compulsory assignment can be robustly justified, the social planner defaults to self-selection, recognizing that there is no sufficient basis to override individual discretion.
\end{remark}

\cref{thm:AMMI} characterizes the selection-driven MMI policy. Define the worst-case conditional average treatment effect on the treated (ATT):
\[
\mathrm{ATT}^-_{\Lambda}(x) \equiv \inf_{Q \in \mathcal{P}(\Lambda)} \mathbb{E}_Q \left[ Y(1)- Y(0) | X=x, A =1 \right],
\]
which is point-identified as:
\[
\mathrm{ATT}^-_{\Lambda}(x) = \frac{\mathbb{E}[Y | X=x] - \mu_{\Lambda}^{+}(x, 0)}{e(x)}.
\]

\begin{theorem}\label{thm:AMMI}
 Under \cref{assumption: Marginal sensitivity model}, for any policy $\pi^\dagger \in \Pi^\dagger_o$, we have that
\[
\begin{aligned}
 \Delta_{\Lambda}(\pi^\dagger) &  = \mathbb{E} \left[ \mathds{1}\{ \pi^\dagger(X) = 1 \} \mu_{\Lambda}^{-} (X,1) + \mathds{1}\{ \pi^\dagger(X) = 0 \} \mu_{\Lambda}^{+} (X,0)\right]\\
  &  +    \mathbb{E}\left[ \mathds{1}\{ \pi^\dagger(X) = \mathrm{S} \} \mathbb{E}[Y|X] \right] -  \mathbb{E} \left[ \mu_{\Lambda}^{+} (X,0) \right],
  \end{aligned}
  \]
Moreover, a first-best MMI policy that solves \cref{eq:AMMI-policy}  is given by
 \begin{equation}\label{eq:AMMI}
     \pi_{\Delta,\Lambda}^{\dagger} (x) = \begin{cases}
1, & \text{if}\ \ \mu_{\Lambda}^{-} (x,1) > \mu_{\Lambda}^{+} (x,0)  \ \ \text{and}\ \ \mathrm{ATU}^-_{\Lambda}(x) >0 , \\
0, & \text{if}\ \ \mu_{\Lambda}^{+} (x,0) > \mu_{\Lambda}^{-} (x,1)  \ \ \text{and}\ \ \mathrm{ATT}^-_{\Lambda}(x)  < 0, \\
\mathrm{S}, &  \text{if}\ \  \mathrm{ATU}^-_{\Lambda}(x)  \leq 0  \ \ \text{and}\ \ \mathrm{ATT}^-_{\Lambda}(x)  \geq 0   . \\
\end{cases}
 \end{equation}
\end{theorem}

\begin{remark} Following \cref{remark:selection_driven_targeting}, the selection-driven MMI policy in \cref{eq:AMMI} imposes a stricter criterion for compulsory intervention than the MMW policy in \cref{eq:MMW_SS}. By evaluating the expected outcome under self-selection $\mathbb{E}[Y|X=x]$  against the best-case control outcome $\mu_{\Lambda}^{+}(x, 0)$, it further expands the ambiguity zone   and defaults a larger share of the population to self-selection. 
\end{remark}

\section{ {Baseline-Relative MMI and Minimax Regret}} \label{appendix: baseline policy and mmr}

This appendix collects two extensions of the robust policy learning framework: a baseline-relative version of the MMI criterion introduced at the end of \cref{section: Policy Learning beyond Unconfoundedness}, and a minimax regret criterion.

\subsection{Baseline-Relative Max-Min Improvement}

In this subsection, we investigate the baseline-relative policy improvement criterion introduced in \cref{eq: baseline MMI}. For an arbitrary baseline policy $\pi_{0} \in \Pi$, we define
\[
\Delta_{\Lambda}(\pi,\pi_{0}) = \inf_{Q\in \mathcal{P}(\Lambda)} \mathbb{E}_Q\left[ Y(\pi(X)) - Y\left(\pi_0(X) \right) \right].
\]
Given $\pi_{0}\in\Pi$, the optimal policy under this criterion is defined as
\begin{equation}\label{eq:MMI-baseline}
\pi_{\Delta,\Lambda}(\cdot;\Lambda,\pi_{0}) \in \mathop{\operatorname{argmax}}_{\pi \in \Pi} \Delta_{\Lambda}(\pi,\pi_{0}).
\end{equation}
The following corollary provides a sharp characterization of $\Delta_{\Lambda}(\pi,\pi_{0})$ and its associated first-best policy.

\begin{corollary}\label{thm:MMI-baseline}  Under \cref{assumption: Marginal sensitivity model}, for any policies $\pi,\pi_{0}\in \Pi$, we have
    \begin{align*}
        \Delta_{\Lambda}(\pi,\pi_{0})
         = \mathbb{E} \left[  \tau_{\Lambda}^{-}(X) \max \left\{ \pi(X) - \pi_{0}(X), 0 \right\} - \tau_{\Lambda}^{+}(X) \max \left\{ \pi_{0}(X) - \pi(X), 0 \right\} \right].
    \end{align*}
    Moreover, a first-best MMI policy relative to the baseline $\pi_0$, which solves \cref{eq:MMI-baseline} over the unconstrained policy space $\Pi=\Pi_{o}$, is given by
    \[
    \pi_{\Delta,\Lambda}^{\star}(x;\pi_{0}) = \begin{cases}
    1, & \text{if}\ \  \tau_{\Lambda}^-(x) \geq 0, \\
    \pi_{0}(x), & \text{if}\ \  \tau_{\Lambda}^-(x) < 0 < \tau_{\Lambda}^+(x), \\
    0, & \text{if}\ \  \tau_{\Lambda}^+(x) \leq 0. \\
\end{cases}
    \]    
\end{corollary}

\Cref{thm:MMI-baseline} shows that, under ambiguity, the first-best baseline-relative MMI policy has a baseline-preserving property: whenever the identified interval for the CATE contains zero, the optimal rule retains the baseline decision $\pi_0(x)$. Since the standard no-treatment baseline, $\pi_0(x)\equiv 0$, is nested as a special case, \cref{thm:MMI-baseline} directly generalizes \cref{theorem: CATE-based robust welfare}.

\subsection{Minimax Regret}\label{subsection:MMR}

Following \citet{manski2004statistical}, we consider an alternative criterion that minimizes the worst-case regret over the MSM uncertainty set. For any distribution $Q \in \mathcal{P}(\Lambda)$, the regret of a policy $\pi$ is defined as its welfare shortfall relative to the optimal policy within the class $\Pi$:
\[
\mathrm{Reg}_Q (\pi) = \sup_{\pi^\prime \in \Pi} \mathbb{E}_Q \left[ Y(\pi^\prime(X)) \right] - \mathbb{E}_Q \left[ Y(\pi(X)) \right].
\]
The MMR policy, denoted by $\pi^\star_{R,\Lambda}$, minimizes this welfare loss uniformly over the ambiguity set $\mathcal{P}(\Lambda)$. Formally, it is defined as the solution to the following minimax optimization problem:
\begin{equation}\label{equation: min-max regret PL}
\pi^\star_{R,\Lambda} \in \mathop{\operatorname{argmin}}_{\pi \in \Pi} \sup_{Q \in \mathcal{P}(\Lambda)} \mathrm{Reg}_Q (\pi).
\end{equation}

\begin{theorem}\label{theorem: minimax regret policy}
Under \cref{assumption: Marginal sensitivity model}, we have
\begin{equation}\label{equation: minimax regret expression - minimax formula}
\inf_{\pi \in \Pi} \sup_{Q \in \mathcal{P}(\Lambda)} \Reg_Q(\pi) =  \inf_{\pi \in \Pi} \sup_{\pi^\prime \in \Pi}   \mathbb{E} \left[ b_{\Lambda}(\pi, \pi^\prime) (X) \right],
\end{equation}
where
\[
b_{\Lambda}(\pi, \pi^\prime)(x) = \tau_{\Lambda}^{+}(x) \max \left\{ \pi'(x) - \pi(x), 0 \right\} - \tau_{\Lambda}^{-}(x) \max \left\{ \pi(x) - \pi'(x), 0 \right\}.
\]
In particular, if $\Pi =\Pi_o$, the first-best MMR policy
\[
\pi^{\star}_{R,\Lambda} (x) =  \begin{cases}
    1, & \text{if}\ \   \tau^-_\Lambda(x) \geq 0, \\
    \tau^+_\Lambda(x)/ \left(\tau^+_\Lambda(x) - \tau^-_\Lambda(x) \right), & \text{if}\ \  \tau^-_\Lambda(x) < 0 < \tau^+_\Lambda(x) \\
    0, & \text{if}\ \  \tau^+_\Lambda(x) \leq 0, \\
\end{cases}, 
\]
solves the minimax regret problem of \cref{equation: minimax regret expression - minimax formula}.
\end{theorem}

\begin{remark}
\Cref{theorem: minimax regret policy} highlights a distinction between the MMR criterion and the max-min criteria considered in the main text. If $\tau^-_\Lambda(x)>0$ or $\tau^+_\Lambda(x)<0$, the sign of the treatment effect is robustly determined, and the first-best MMR policy is deterministic. If $0 \in \left(\tau^-_\Lambda(x),\tau^+_\Lambda(x) \right)$, neither treatment nor control uniformly dominates, and randomization arises as a way to balance the worst-case regret from choosing the wrong action. Additionally, when $\Pi$ is restricted to be deterministic, the corresponding first-best MMR policy in this scenario becomes $ \mathds{1}\left\{ \tau^+_\Lambda(x) + \tau^-_\Lambda(x) > 0 \right\}$.
\end{remark}

\section{ {Additional Empirical Analyses}}\label{section: Additional Empirical Analyses}

This section provides supplementary empirical results and robustness checks to complement the results in \cref{section: empirical}. Specifically, we extend our confounding-robust policy learning framework to incorporate self-selection as an explicit policy option for both the JTPA and Head Start studies. We then conduct sensitivity analyses to verify that our empirical conclusions remain robust to alternative choices of the truncation threshold $\alpha$ and the sensitivity parameter $\Lambda$.

\subsection{JTPA Application}\label{appendix: JTPA}

We revisit the JTPA application from \cref{section: JTPA} and incorporate the self-selection as a policy option. Specifically, we apply the MMW policy with self-selection, as defined in \cref{eq:AMMW-policy} and characterized by \cref{thm:AMMW}. Relative to \cref{section: JTPA}, the only additional component is  $\mathbb{E}[Y|X=x]$, which is estimated using random forests.  We restrict the policy class to depth-2 decision trees over \(\{0,1,\mathrm{S}\}\), and estimate the MMW policy  using \texttt{policytree} \texttt{R} package. \cref{fig:jtpa-3arm-share-welfare}  reports assignment fractions across policy arms and the corresponding welfare performance over the $\Lambda$-grid \( \{1,1.5,\ldots,5\}\), which includes the sensitivity level calibrated using IV-bounds, \(\Lambda = 4.5\).

Panel~(a) of \Cref{fig:jtpa-3arm-share-welfare} reports the
assignment fraction of the estimated depth-2 policy tree across the
specified $\Lambda$-grid.  At \(\Lambda = 1\), the policy mirrors the unconfoundedness benchmark, assigning 93.5\% of applicants to compulsory treatment and 6.5\% to compulsory control, with no assignment to the self-selection arm. For \(\Lambda>1\), the policy assigns most applicants to the self-selection arm: \(98.5\%\) at \(\Lambda=1.5\), and more than \(99.6\%\) for all \(\Lambda \geq 3\). This pattern is consistent with the robust decision rule in \Cref{remark:selection_driven_targeting}. As \(\Lambda\) increases, the worst-case welfare bounds for compulsory assignments become less informative, and the resulting MMW policy tends to favor the self-selection arm, allowing applicants to make their own take-up decisions.

Panel~(b) reports welfare gains relative to the compulsory un-treatment-only baseline. We compare compulsory treatment only, self-selection only, selection-absent targeting, and selection-driven targeting, with the compulsory un-treatment-only policy normalized to zero. We evaluate these policies using the worst-case welfare criterion $W_{\Lambda}$ characterized in \Cref{thm:AMMW}. When $\Lambda = 1$, selection-driven and selection-absent targeting coincide  and deliver higher welfare gains than the self-selection-only policy. For $\Lambda \geq 1.5$, the selection-driven policy assigns almost all applicants to the self-selection arm, so its welfare gain becomes close to that of the self-selection-only policy.

\begin{figure}[ht]
  \centering
  \includegraphics[width=\textwidth]{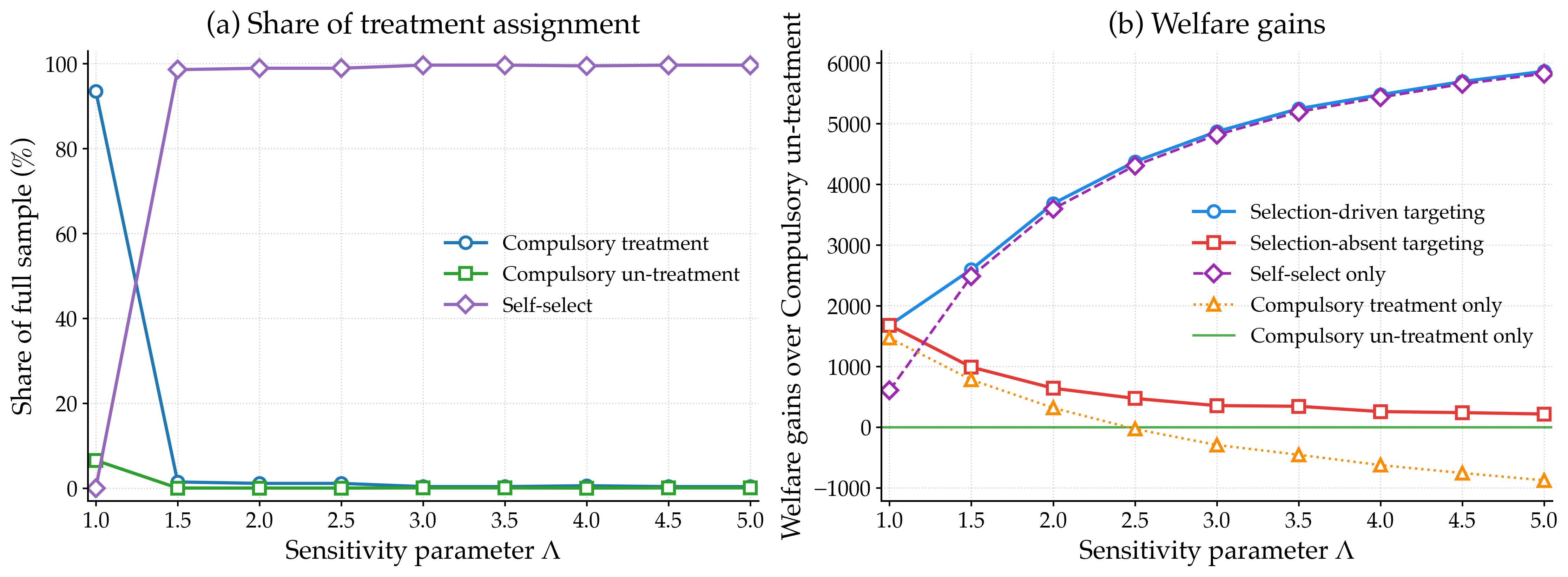}
\caption{Policy assignment fractions and welfare performance across values of \(\Lambda\). Panel~(a) plots assignment fractions across policy arms. Panel~(b) plots welfare gains of four candidate policies relative to the compulsory un-treatment-only baseline.}
  \label{fig:jtpa-3arm-share-welfare}
\end{figure}

\subsection{Head Start Application}\label{section: HS_4arm}

\subsubsection{Robustness Checks}\label{appendix: head start 3arm}

In this subsection, we examine the robustness of the Head Start results to two implementation choices: the propensity-score truncation level and the sensitivity parameter in the MSM. We first assess robustness to the propensity-score truncation level $\alpha$ by re-running the empirical analysis under $\alpha \in \{0,0.01,0.05,0.1\}$, holding all other implementation details fixed.

Table~\ref{tab: clip counts} reports the number and percentage of observations with estimated propensity scores outside  $[\alpha,1-\alpha]$, both in the full sample and separately by treatment arms. The table shows that only a small fraction of observations have estimated propensity scores outside the truncation interval. Specifically, this fraction is 0.03\% when $\alpha=0.01$, 0.39\% under our baseline choice $\alpha=0.05$, and 3.24\%  under the more conservative truncation level $\alpha=0.1$. These small fractions suggest that the estimated policy tree and the subsequent empirical findings are not sensitive to the particular choice of $\alpha$.

\begin{table}[H]
\centering
\small
\caption{Number and percentage of observations outside the propensity score truncation interval across choices of $\alpha$ with sample size $n=3{,}826$.}
\label{tab: clip counts}
\begin{tabular}{l rrrr}
\toprule
& $\alpha = 0$ & $\alpha = 0.01$ & $\alpha = 0.05$ & $\alpha = 0.1$ \\
\midrule
Head Start         & 0 & 1 & 14 & 51  \\
Other preschool    & 0 & 0 & 1  & 53  \\
No preschool       & 0 & 0 & 0  & 20  \\
\midrule
Full sample & 0 & 1 & 15 & 124 \\
Percentage           & 0\% & 0.03\% & 0.39\% & 3.24\% \\
\bottomrule
\end{tabular}
\end{table}

We next turn to the robustness check regarding the sensitivity parameter $\Lambda$. We re-estimate the depth-2 MMW policy tree over the grid $\Lambda\in\{1,1.5,2,2.5,3,3.5,3.6,4\}$; the resulting trees are displayed in Figure~\ref{fig:HS-policy-tree-Lambda-grid}. For \(\Lambda \geq 1.5\), the learned policy trees share the same qualitative structure: Head Start is targeted to children from lower-income families whose mothers have relatively low AFQT percentile scores. The main change across \(\Lambda\) is the AFQT cutoff, which rises from \(4.31\) to \(7.36\), increasing the Head Start assignment share from \(5.10\%\) to \(7.61\%\). For larger values of $\Lambda$, the MMW policy assigns a  larger fraction of children to Head Start. In particular, the Head Start assignment share increases from 5.10\% under \(\Lambda \in \{1.5,2,2.5,3\}\) to 7.61\% under \(\Lambda \in \{3.5,3.6,4\}\).

\begin{figure}[htb]
    \centering
    \includegraphics[width=\linewidth]{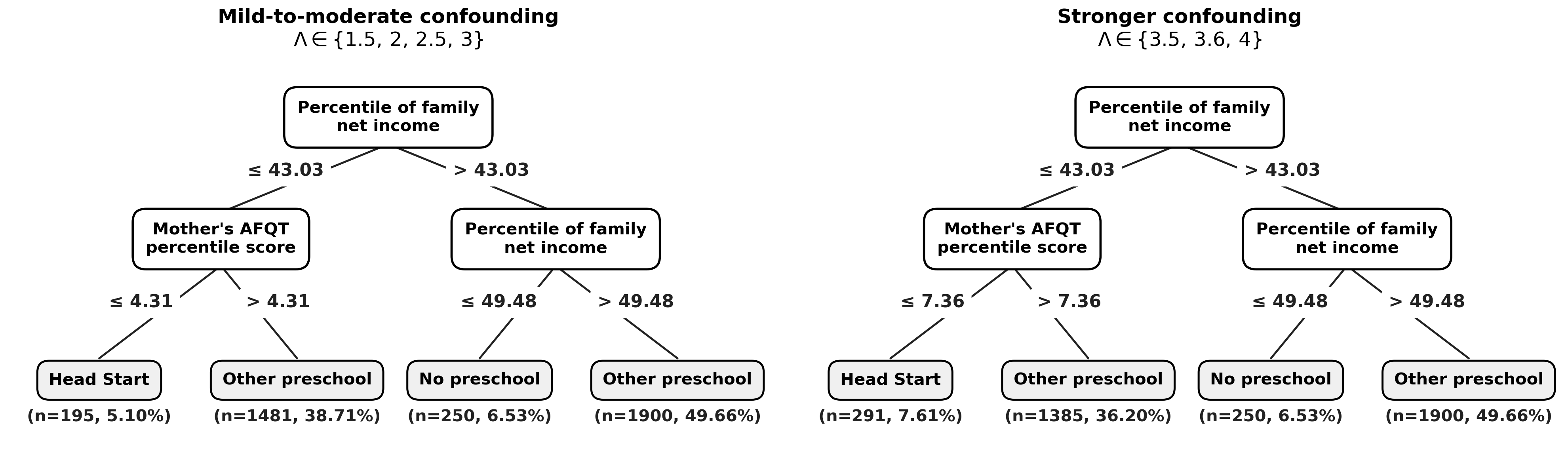}
    \caption{Sensitivity analysis of learned depth-2 policy trees across values of \(\Lambda\). The trees have the same qualitative structure across the reported values of \(\Lambda\), with changes only in the splitting thresholds. For larger values of \(\Lambda\), the higher AFQT cutoff expands the Head Start assignment probability.}

    \label{fig:HS-policy-tree-Lambda-grid}
\end{figure}

\subsubsection{Head Start program with self-selection as an option}

We revisit the Head Start application in \cref{section: headstart}, now allowing self-selection as a policy option. We use the same sample, covariates, and estimations as in \cref{section: headstart}; the additional component, \(\mathbb{E}[Y|X]\), is estimated by random forest regression. We report results over \(\Lambda \in \{1,1.5,\ldots,4\}\), a grid that includes the sensitivity level \(\Lambda=3.6\) calibrated using the method of \cite{hsu2013calibrating}.

Panel~(a) of \Cref{fig:headstart-4arm-share-welfare} reports the assignment fractions of the estimated depth-2 policy tree across values of \(\Lambda\). At \(\Lambda=1\), the policy assigns \(1.23\%\) of children to Head Start, \(86.83\%\) to other preschool, and \(11.94\%\) to no preschool, with no assignment to the self-selection arm. For values of \(\Lambda\) above the unconfoundedness benchmark, the policy assigns most children to self-selection: \(91.19\%\) at \(\Lambda=1.5\), and \(95.71\%\) for all \(\Lambda \geq 2\).

Panel~(b) reports welfare performance. We compare six candidate policies: no preschool only, Head Start only, other preschool only, self-selection only, selection-absent targeting, and selection-driven targeting. As in the JTPA application, welfare gains are measured relative to the no-preschool-only baseline using the worst-case welfare criterion. The pattern is similar to that in the JTPA application. For values of \(\Lambda\) above the unconfoundedness benchmark, the selection-driven policy assigns most children to the self-selection arm, so its welfare gain is closed to that of the self-selection-only policy and exceeds that of selection-absent targeting.

\begin{figure}[htb]
    \centering    \includegraphics[width=\linewidth]{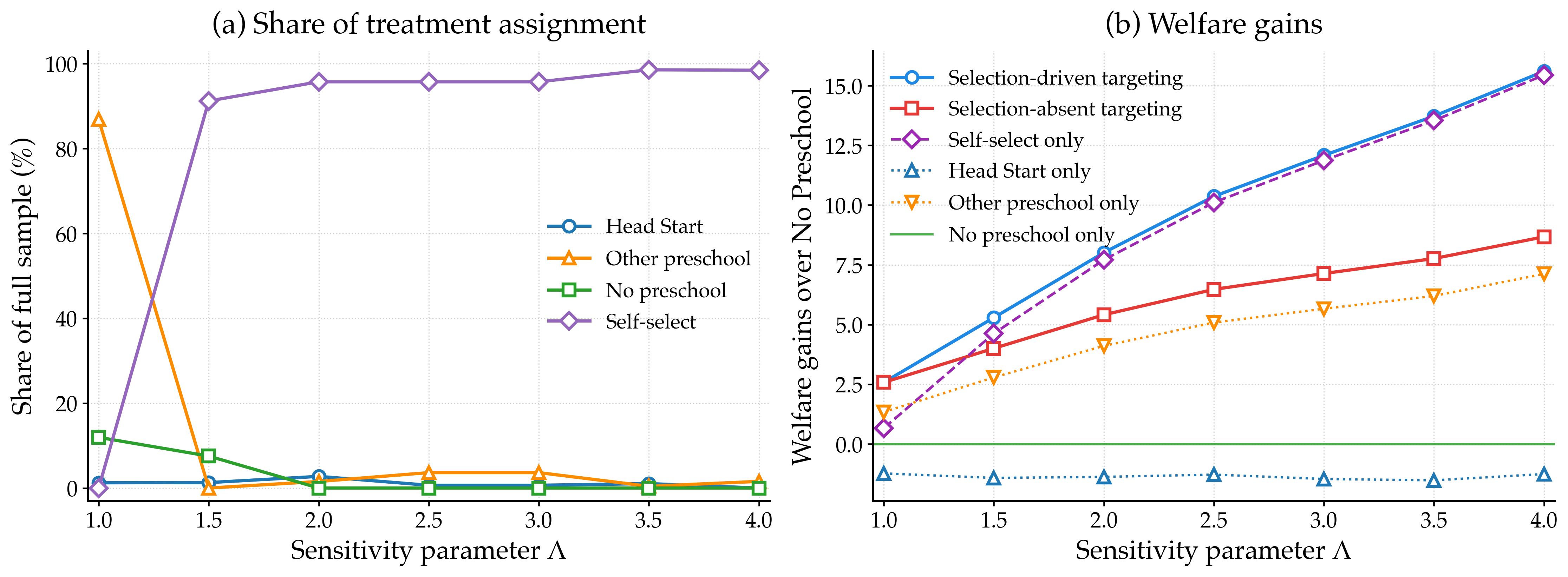}
    \caption{Policy assignment shares and welfare performance across $\Lambda$-grid. The figures plot (a) the share of treatment allocation for each policy arm, and (b) the welfare gains of candidate policies relative to the no-preschool-only.}
    \label{fig:headstart-4arm-share-welfare}
\end{figure}

\end{appendices}

\bibliographystyle{apalike} 
\bibliography{reference}

\newpage

\begin{appendices}
	\section{Identification Results under MSM}\label{section: identification_Conditional_Mean_CATE}


	In this section, we present partial identification results for the conditional mean and CATE under the MSM, drawing primarily from \citep{dorn2023sharp,oprescu2023b,dorn2024doubly}. Recall that $\mu_{o}(x,a)=\mathbb{E}_{P_{o}}\left[ Y(a) | X=x \right]$ denotes the conditional mean of the potential outcome, and the true CATE is defined as $\tau_o(x) = \mu_o(x,1)  - \mu_o(x,0)$.

	To formally characterize partial identification, we describe the identified sets of these functions under distributional uncertainty, where the true counterfactual distribution may deviate within an uncertainty set $\mathcal{P}(\Lambda)$. For any distribution $Q \in \mathcal{P}(\Lambda)$,  define $\mu_{Q}(x,a) = \mathbb{E}_{Q}[Y(a)|X=x]$ and $\tau_{Q}(x) = \mu_{Q}(x,1)-\mu_{Q}(x,0)$, which represent the conditional mean and CATE under $Q$. The identified sets for these functions are then  
	\[
	\begin{aligned}
		\Theta_{\mu, \Lambda}(x,a)  & \equiv \left \{ \mu_Q(x, a) : Q \in \mathcal{P}(\Lambda)  \right \},  \\
		\Theta_{\tau,\Lambda}(x)  & \equiv  \left\{ \tau_Q(x): Q \in \mathcal{P}(\Lambda)  \right \},
	\end{aligned}
	\]
	for any  $x \in \mathcal{X}$ and  $a \in \{0,1\}$. When $\Lambda=1$, the uncertainty set $\mathcal{P}(\Lambda)$ collapses to a singleton, implying point identification: $\mu_{o}(x,a)$ and $\tau_o(x)$ are uniquely determined by $\Theta_{\mu, \Lambda}(x,a)$ and $\Theta_{\tau,\Lambda}(x)$.
	
	For general $\Lambda$, let $ \mu^{\pm}_{\Lambda}(x,a)$ denote the endpoints of $\Theta_{\mu, \Lambda}(x,a)$, i.e.,
	\[
	\mu^{-}_{\Lambda}(x,a)  = \inf_{ Q \in \mathcal{P}(\Lambda) } \mu_{Q}(x,a)  \quad \text{and} \quad
	\mu^{+}_{\Lambda}(x,a)  = \sup_{ Q \in \mathcal{P}(\Lambda) } \mu_{Q}(x,a).
	\]
	The following \cref{proposition: pi_mean,proposition: pi_CATE} provide sharp upper and lower bounds for both $\Theta_{\mu, \Lambda}(x,a)$ and $\Theta_{\tau,\Lambda}(x)$.
	
	\begin{proposition}\label{proposition: pi_mean}
		Under \cref{assumption: Marginal sensitivity model}, there are distributions $P^{\pm}_{\mu,a} \in \mathcal{P}(\Lambda)$  such that   $\mu_{P_{\mu,a}^+}(x,a) = \mu^{+}_{\Lambda}(x,a)$ and 
		$\mu_{P_{\mu,a}^-}(x,a) = \mu^{-}_{\Lambda}(x,a)$, almost surely. In particular, these sharp bounds can be expressed as:  
		\begin{equation}~\label{equation: eqsharpbounds}
			\mu^{\pm}_{\Lambda}(x,a) = \mathbb{E}\left[ Y \mathds{1}\{ A = a \} \left[1+\frac{1- e_a(X)}{e_a(X)}\Lambda^{\pm\text{sgn}\left(Y-q^{\pm}_{\Lambda}(X,a)\right)}\right] \Big|X=x \right], \\
		\end{equation}
		where $\text{sgn}(t)=1$ if $t\geq0$ and $-1$ otherwise.
	\end{proposition}

	We next characterize the upper and lower bounds for $\Theta_{\tau,\Lambda}(x)$. Specifically, the endpoints of $\Theta_{\tau,\Lambda}(x)$ are defined as
	\[
	\tau^{-}_{\Lambda}(x) = \inf_{ Q \in \mathcal{P}(\Lambda) } \tau_{Q}(x) \quad \text{and} \quad \tau^{+}_{\Lambda}(x) = \sup_{ Q \in \mathcal{P}(\Lambda) } \tau_{Q}(x).
	\]
	The following result, \cref{proposition: pi_CATE}, builds on \cref{proposition: pi_mean}.
	
	\begin{proposition}\label{proposition: pi_CATE}
		Under \cref{assumption: Marginal sensitivity model}, there are distributions $P^{\pm}_\tau \in \mathcal{P}(\Lambda)$ such that   $ {\tau}_{P_\tau^+}(x) = \tau^+_{\Lambda}(x)$ and 
		$\tau_{P_\tau^-}(x) = \tau^-_{\Lambda}(x)$, almost surely. Moreover, the sharp bounds for the CATE satisfy:
		\begin{equation}\label{eqsharpcatebound}
			\tau^{\pm}_\Lambda(x) = \mu^{\pm}_{\Lambda}(x,1) - \mu^{\mp}_{\Lambda}(x,0).
		\end{equation}
	\end{proposition}

	We illustrate the intuition behind \cref{proposition: pi_mean}. Suppose we had oracle access to the true propensity score $e_{o}(X,U)$. Under \cref{assumption: Marginal sensitivity model}, the true conditional mean $\mu_o(x,a)$ would then be point-identified via IPW: 
	\[
	\mu_o(x,a) = \mathbb{E}_{P_o}\left[  \frac{Y \mathds{1}\{ A=a \} }{ a e_{o}(X,U) + (1-a) \left( 1-e_{o}(X,U) \right)  } \Big | X =x \right],
	\]
	However,  when $\Lambda >1$,  the true propensity score $e_o(X,U)$ is no longer identifiable. Instead, \cref{assumption: Marginal sensitivity model} restricts $ e_{o}(X,U)$ to an uncertainty set, which in turn induces a partially identified set for $\mu_o(x,a)$. Specifically, the uncertainty set consists of all random variables $E$ that satisfy both the sensitivity restriction and the usual balancing constraints:
	\begin{align}
		\Lambda^{-1} & \leq \frac{E/(1-E)}{e(X)/\left(1-e(X)\right)}  \leq \Lambda, \label{equation:MMS_RV}  \\
		\mathbb{E}\left[ A / E | X \right] & =  \mathbb{E}\left[ (1-A) / (1- E) | X \right] = 1.
		\label{equation: balancing constraints} 
	\end{align}
	Formally, we define
	\[
	\mathcal{E}(\Lambda) = \left \{ E: \text{\cref{equation:MMS_RV} and \cref{equation: balancing constraints} hold a.s.} \right\}.
	\]
	Given this uncertainty set, the identified set $\Theta_{\mu, \Lambda}(x,a)$ can be equivalently expressed as: 
	\[
	\Theta_{\mu, \Lambda}(x,a) = \left\{  \mathbb{E} \left[ \frac{Y \mathds{1}\{ A = a \} }{ a E  + (1-a) (1-E) } \Big | X = x \right ] : E \in \mathcal{E}(\Lambda) \right\}.
	\]
	For completeness, the formal proofs of \cref{proposition: pi_mean,proposition: pi_CATE} are provided in \cref{section: proofs of Identification} for interested readers.

		

	\begin{remark}\label{remark:A1}
		Our \cref{proposition: pi_mean,proposition: pi_CATE} extend the pointwise identification results of \cite{dorn2023sharp,oprescu2023b} to a uniform setting, ensuring that the identified bounds hold simultaneously for almost all $x$ and $a$.  Such uniform bounds are crucial for policy evaluation and learning under the MSM. First, computing the robust criteria $W_\Lambda(\pi)$ and $\Delta_\Lambda(\pi)$ involves integration over the covariate and treatment spaces, which requires uniformly valid bounds on $\Theta_{\mu}(x,a)$ and $\Theta_{\tau}(x)$. Second, the closed-form sharp bounds enable the construction of orthogonal moments for the efficient estimation of $W_\Lambda(\pi)$ and $\Delta_\Lambda(\pi)$. As shown by \cite{kallus2021minimax}, the policy improvement guarantee can be severely biased by the slow convergence of propensity score estimation; by leveraging our identification results, this bias can be mitigated via orthogonal moments, yielding $\sqrt{n}$-consistent estimates of $W_\Lambda(\cdot)$ and $\Delta_\Lambda(\cdot)$.
	\end{remark}
	
	\begin{remark}\label{remark:A2}
		\cite{kallus2021minimax}, building on \cite{zhao2019sensitivity}, constructs an uncertainty set for the putative propensity scores by imposing unconditional balancing constraints, specifically $\mathbb{E}[A/E] = \mathbb{E}[(1-A)/(1-E)] = 1$.  In contrast, our uncertainty set $\mathcal{E}(\Lambda)$ incorporates the conditional balancing constraints defined in \cref{equation: balancing constraints}, which are strictly stronger than those employed in \cite{kallus2021minimax}. According to \cite{dorn2023sharp},  such conditional balancing constraints are necessary for obtaining sharp bounds on the CATE, as established in \cref{proposition: pi_CATE}. Consequently, the policy learned in \cite{kallus2021minimax} is more conservative than our MMI policy.
	\end{remark}

	\section{Extending MMW Policy to Multi-Valued Treatments}\label{appendix: discrete treatment}
	
	This section extends our max-min welfare (MMW) policy learning approach to the multi-valued treatment setting, providing theoretical justification for its application in \cref{section: headstart}. Our identification here is constructive, directly informing a procedure for estimating the lower bound of the worst-case welfare function. This, in turn, allows for policy optimization using a method directly analogous to that in \cref{algo:MMW_PL}. As these procedures and the corresponding doubly robust estimation closely follow the binary treatment case, we omit a detailed discussion here.
	
	The identification results, including sharp upper and lower bounds for conditional mean outcomes and a lower bound for the worst-case welfare function, are established in \cref{appendix: MSM identification}. Detailed proofs are presented in \cref{section:Multi_Identification}.

	We first extend the MSM framework to the setting where the treatment variable $A$ takes more than two values. Let $\mathcal{A}$ denote the set of $d$ possible treatment levels: $\mathcal{A} = \{a_1,\ldots, a_d\}$. Define the true and nominal propensity scores, respectively, as $e_{o,a}(x,u) = \mathbb{P}_{P_{o}}[A=a|X=x,U=u]$ and $e_{a}(x) = \mathbb{P}_{P_{o}}[A=a|X=x]$ for $a\in\mathcal{A}$.
	
	\begin{assumption}\label{assumption: MSM with discrete treatment}
		Suppose there exists a vector of unobserved confounders $U \in \mathbb{R}^k$ such that  
		\[
		\left( Y(a_{1}), Y(a_{2}),\dots, Y(a_{d})  \right) \indep  A \mid \left(X, U \right).
		\]
		The distribution of $\left(X, Y(a_{1}), Y(a_{2}),\dots, Y(a_{d}), A, U \right) $ satisfies the selection bias condition with $1 \leq \Lambda< \infty$ if the following inequality holds $P_o$-almost surely,
		\begin{equation}
			\frac{1}{\Lambda} \leq \frac{e_{o,a}(x, u) / \left( 1- e_{o,a}(x, u)  \right)}{e_{a}(x ) / \left( 1- e_{a}(x)  \right) }  \leq \Lambda \quad \text{for} \quad a\in\mathcal{A}.
		\end{equation}
	\end{assumption}
	
	Assumption \ref{assumption: MSM with discrete treatment} naturally induces a distributional uncertainty set over the counterfactual distribution of $(X, Y(a_{1}), Y(a_{2}), \dots, Y(a_{d}), A, U)$. Formally, let $\mathcal{P}_{\mathrm{M}}$ denote the set of all probability distributions $Q$ on $\mathcal{X} \times \mathbb{R}^{d} \times \mathcal{A} \times \mathbb{R}^k$ that satisfy:
	\begin{enumerate}
		\item[(1)]   If $\left(X, Y(a_{1}), Y(a_{2}),\dots, Y(a_{d}), A, U \right)\sim Q$, then $\left(Y(a_{1}), Y(a_{2}),\dots, Y(a_{d}) \right) \indep A \mid (X,U)$ under $Q$;
		\item[(2)] If $Y = Y(A)$, then the distribution of $(X,Y,A)$ under $Q$ coincides with the observed-data distribution $P$;
		\item[(3)] For all $a\in\mathcal{A}$, the odds ratio between the true propensity score and the nominal propensity score lies in $\left[1/\Lambda, \Lambda \right]$, i.e., 
		\[
		\frac{1}{\Lambda} \leq  \frac{ \mathbb{P}(A = a |X, U)/  \big ( 1 - \mathbb{P}(A = a |X, U) \big ) }{\mathbb{P}(A =a |X)/ \big ( 1 - \mathbb{P}(A = a |X) \big ) } \leq \Lambda, \quad \text{$Q$-a.s.}
		\]
	\end{enumerate}
	
	\subsection{Identifying the Worst-Case Welfare Function}\label{appendix: MSM identification}
	
	For any  $(x,a)\in\mathcal{X}\times\mathcal{A}$, let  $\mu_Q(x, a)=\mathbb{E}_{Q}[Y(a)|X=x]$ denote the conditional mean under a counterfactual distribution $Q$. The identified set for $\mu_o(x,a) = \mathbb{E}_{P_{o}}[Y(a)|X=x]$ is formally given by
	\[
	\Theta_{\mathrm{M},\mu}(x,a)   \equiv \left \{ \mu_Q(x, a) : Q \in \mathcal{P}_{\mathrm{M}}  \right \}.
	\]
	Let $\mu^{\pm}(x,a)$ denote the endpoints of $\Theta_{\mathrm{M},\mu}(x,a)$, i.e.,
	\[
	\mu^{-}(x,a)  = \inf_{ Q \in \mathcal{P}_{\mathrm{M}}} \mu_{Q}(x,a) \quad \text{and} \quad  
	\mu^{+}(x,a)  = \sup_{ Q \in \mathcal{P}_{\mathrm{M}}} \mu_{Q}(x,a).
	\]
	The following \cref{proposition: sharp bounds for conditional mean with discrete treatment} establishes the sharp upper and lower bounds for the identified set $\Theta_{\mathrm{M},\mu}(x,a)$.
	
	\begin{proposition}\label{proposition: sharp bounds for conditional mean with discrete treatment}
		Under \cref{assumption: MSM with discrete treatment}, there exist distributions $P^{\pm}_{\mu,a} \in \mathcal{P}_{\mathrm{M}}$  such that   $\mu_{P_{\mu,a}^+}(x,a) = \mu^{+}(x,a)$ and 
		$\mu_{P_{\mu,a}^-}(x,a) =\mu^{-}(x,a)$ for almost all $(x,a)\in\mathcal{X}\times\mathcal{A}$, respectively. In particular, the sharp bounds admit the closed-form representation:
		\begin{equation}\label{eqsharpbounds}
			\mu^{\pm}(x,a) = \mathbb{E}\left[ Y \mathds{1}\{ A = a \} \left[1+\frac{1- e_a(X)}{e_a(X)}\Lambda^{\pm\text{sgn}\left(Y-q^{\pm}(X,a)\right)}\right] \Big|X=x \right]. \\
		\end{equation}
	\end{proposition}
	
	In the multi-valued treatment setting, a policy is a mapping from the input space $\mathcal{X}$ to a decision $a \in \mathcal{A}$. A randomized policy can be represented as a function from $\mathcal{X}$ to the probability simplex $\Delta(\mathcal{A})$ over the action space $\mathcal{A}$. In contrast, a deterministic policy is a function $\pi : \mathcal{X} \to \{0,1\}^d$, where the output indicates a deterministic choice among $d$ possible actions.
	
	The worst-case welfare function $W(\pi)$ introduced in Section \ref{section: Policy Learning beyond Unconfoundedness} can be naturally extended to the multi-valued treatment setting:
	\[
	W(\pi) =  \inf_{ Q \in \mathcal{P}_{\mathrm{M}}} \mathbb{E}_Q \left[ Y(\pi(X)) \right] = \inf_{ Q \in \mathcal{P}_{\mathrm{M}}} \mathbb{E}_Q \left[ \sum_{a \in \mathcal{A}} Y(a) \pi_{a}(X) \right],
	\]
	Let $\Pi_{\mathrm{M}}$ denote the set of multi-valued policies specified by the policy maker. The corresponding max-min welfare (MMW) policy is obtained by solving $\max_{\pi \in \Pi_{\mathrm{M}}} W(\pi)$. The following theorem characterizes a lower bound for $W(\pi)$.
	
	\begin{theorem}\label{theorem: max-min welfare under discrete treatment}
		Under \cref{assumption: MSM with discrete treatment}, for any policy $\pi\in\Pi_{\mathrm{M}}$, 
		\[
		W(\pi) \geq \mathbb{E}\left[ \sum_{a \in \mathcal{A}} {\mu}^{-}(X,a) \pi_{a}(X) \right].
		\]
	\end{theorem}
	
	\begin{remark}
		Whether the lower bound for $W(\pi)$ in Theorem~\ref{theorem: max-min welfare under discrete treatment} is sharp remains an open question. While sharp lower bounds for $\mu_o(x,a)$ are available for any fixed $(x,a)$, identifying a single distribution that simultaneously attains these bounds for all $(x,a)$ may be challenging and lies beyond the scope of this study.
	\end{remark}

	\section{Proofs for Results in the Main Text}\label{section: proofs of Identification}
	Since the sensitivity parameter $\Lambda$ is treated as fixed throughout, we omit the subscript $\Lambda$ for notational simplicity. For instance, we may write $\mathcal{P}\equiv \mathcal{P}(\Lambda)$,  $\mu^+(x,a) \equiv \mu_{\Lambda}^+(x,a)$ and $\rho^-(x,a) \equiv \rho^-_\Lambda(x,a)$. This convention will be used throughout the remainder of the appendix.
	
	\subsection{Proof of  Proposition \ref{proposition: pi_mean}}
	
	
	\begin{proof}[Proof of \cref{proposition: pi_mean}]
		
		\textbf{Step 1. Preliminary Results.} We begin by showing that for any $(x,a)\in\mathcal{X}\times\{0,1\}$, the partially identified set $\Theta_{\mu}(x,a)$ is an interval whose endpoints solve the following optimization problems:
		\begin{equation}\label{equation: optimization for identified set of mu}
			\begin{aligned}
				& \mu^{+}(x,a) = \sup_{E\in\mathcal{E}} \mathbb{E} \left[ \frac{Y \mathds{1}\{ A = a \} }{ a E  + (1-a) (1-E) } \Big | X = x \right ], \\
				& \mu^{-}(x,a) = \inf_{E\in\mathcal{E}} \mathbb{E} \left[ \frac{Y \mathds{1}\{ A = a \} }{ a E  + (1-a) (1-E) } \Big | X = x \right ]. \\
			\end{aligned}
		\end{equation}
		We focus on verifying the result for $\mu^{+}(x,a)$; the argument for $\mu^{-}(x,a)$ is analogous and omitted for brevity.
		
		By Proposition 1B of \cite{dorn2023sharp}, for any random variable $E$ defined on the same probability space as $(X,Y,A)$ and satisfying $\mathbb{E} \left[ \mathds{1}\{ A = a \} / \left( a E  + (1-a) (1-E) \right) | X \right] = 1$, we can construct a distribution $Q\in\mathcal{P}$ such that
		\begin{equation}\label{equation: existence of distribution}
			\mathbb{E}_{Q}\left[ Y(a) | X = x \right] = \mathbb{E}_{Q}\left[ \frac{Y\mathds{1}\{ A = a \}}{\mathbb{P}_{Q} (A=a | X,U ) } \Big | X = x  \right] = \mathbb{E} \left[ \frac{Y\mathds{1}\{ A = a \}}{a E  + (1-a) (1-E)}  \Big | X = x \right].
		\end{equation}
		Thus,
		\[
		\mu^{+}(x,a) = \sup_{Q\in\mathcal{P}}\mathbb{E}_{Q}\left[ Y(a) | X = x \right] \geq \mathbb{E} \left[ \frac{Y \mathds{1}\{ A = a \} }{ a E  + (1-a) (1-E) } \Big | X = x \right ].
		\]
		Since this inequality holds for any $E$ satisfying the balancing condition, it holds in particular for the supremum over all such $E$. 
		
		Conversely, for any $Q\in\mathcal{P}$,
		\[\begin{aligned}
			\mathbb{E}_{Q}\left[ Y(a) | X = x \right] = & \mathbb{E}_{Q}\left[ \frac{Y\mathds{1}\{ A = a \}}{\mathbb{P}_{Q} (A=a | X,U ) } \Big | X = x  \right] \\ 
			= & \mathbb{E}_{Q}\left[ Y\mathds{1}\{ A = a \} \mathbb{E} \left[ \frac{1}{\mathbb{P}_{Q} (A=a | X,U ) } \Big| X,Y,A=a   \right] \Big | X = x  \right]. \\
		\end{aligned}\]
		Define $e_{Q}(X,Y,A) = 1 / \mathbb{E} \left[ 1/\mathbb{P}_{Q} (A=a | X,U ) | X,Y,A \right]$, and introduce a random variable $E$ on the same probability space $P$ such that $AE + (1-A)(1-E) = e_{Q}(X,Y,A)$. It is straightforward to verify that $E\in\mathcal{E}$ and
		\[\begin{aligned}
			\mathbb{E}_{Q}\left[ Y(a) | X = x \right] = \mathbb{E}\left[ \frac{Y\mathds{1}\{ A = a \}}{e_{Q}(X,Y,a) } \Big | X = x  \right] = \mathbb{E} \left[ \frac{Y\mathds{1}\{ A = a \}}{a E  + (1-a) (1-E)}  \Big | X = x \right].
		\end{aligned}\]
		Thus,
		\[
		\mathbb{E}_{Q}\left[ Y(a) | X = x \right] \leq \sup_{E\in\mathcal{E}} \mathbb{E} \left[ \frac{Y \mathds{1}\{ A = a \} }{ a E  + (1-a) (1-E) } \Big | X = x \right ].
		\]
		Since $Q$ is arbitrary, the inequality continues to hold after taking the supremum over $Q\in\mathcal{P}$. 
		
		Finally, the proof that the partially identified set $\Theta_{\mu}(x,a)$ forms an interval follows the same argument as in \cite{dorn2023sharp}.

		\textbf{Step 2. The Closed-form Expression.} For $a\in\{0,1\}$, define $E^{\pm}_{a}$ by
		\begin{equation}\label{equation: closed-form solution}
			\frac{1}{aE^{\pm}_{a}+(1-a)(1-E^{\pm}_{a})} = 1+\frac{1- e_a(X)}{e_a(X)}\Lambda^{\pm\text{sgn}\left(Y-q^{\pm}(X,a)\right)}.
		\end{equation}
		It is straightforward to verify that $E^{\pm}_{a} \in \mathcal{E}$, and thus $E^{\pm}_{a}$ is feasible for \cref{equation: optimization for identified set of mu}. For any $(x,a) \in \mathcal{X} \times  {\{}0,1 {\}}$ and $E \in \mathcal{E}$, we have
		\[\begin{aligned}
			& \mathbb{E} \left[ \frac{Y \mathds{1}\{ A = a \} }{ a E  + (1-a) (1-E) } \Big | X = x \right ] \\
			= & q^{+}(x,a) \mathbb{E} \left[ \frac{  \mathds{1}\{ A = a \} }{ a E  + (1-a) (1-E) } \Big | X = x \right ] + \mathbb{E} \left[ \frac{ \left( Y - q^{+}(X,a) \right) \mathds{1}\{ A = a \} }{ a E  + (1-a) (1-E) } \Big | X = x \right ] \\
			\leq & q^{+}(x,a) \mathbb{E} \left[ \frac{  \mathds{1}\{ A = a \} }{ a E^{+}_{a}  + (1-a) (1-E^{+}_{a}) } \Big | X = x \right ] + \mathbb{E} \left[ \frac{ \left( Y - q^{+}(X,a) \right) \mathds{1}\{ A = a \} }{ a E^{+}_{a}  + (1-a) (1-E^{+}_{a}) } \Big | X = x \right ] \\
			= & \mathbb{E} \left[ \frac{Y \mathds{1}\{ A = a \} }{ a E^{+}_{a}  + (1-a) (1-E^{+}_{a}) } \Big | X = x \right ], \\
		\end{aligned}\]
		where the inequality follows because $1/\left( a E^{+}_{a}  + (1-a) (1-E^{+}_{a})  \right)$ attains the maximum (minimum) allowable value when $\left( Y - q^{+}(X,a) \right) \mathds{1}\{ A = a \}$ is positive (negative). By a similar argument, we can show that for any $(x,a)\in\mathcal{X}\times\{0,1\}$ and $E\in\mathcal{E}$,
		\[
		\mathbb{E} \left[ \frac{Y \mathds{1}\{ A = a \} }{ a E  + (1-a) (1-E) } \Big | X = x \right ] \geq \mathbb{E} \left[ \frac{Y \mathds{1}\{ A = a \} }{ a E^{-}_{a}  + (1-a) (1-E^{-}_{a}) } \Big | X = x \right ].
		\]
		Since $E$ is arbitrary, this completes the proof.

		\textit{Step 3. The Existence of Counterfactual Distribution.} \cref{equation: optimization for identified set of mu,equation: closed-form solution} together imply that random variables $E_{a}^{\pm}\in\mathcal{E}$ satisfy
		\[
		\begin{aligned}
			& \mu^{+}(x,a) =  \mathbb{E} \left[ \frac{Y \mathds{1}\{ A = a \} }{ a E_{a}^{+}  + (1-a) (1-E_{a}^{+}) } \Big | X = x \right ], \\
			& \mu^{-}(x,a) =  \mathbb{E} \left[ \frac{Y \mathds{1}\{ A = a \} }{ a E_{a}^{-}  + (1-a) (1-E_{a}^{-}) } \Big | X = x \right ], \\
		\end{aligned}
		\]
		in $(x,a)\in\mathcal{X}\times\{0,1\}$ almost surely. By Proposition 1B in \cite{dorn2023sharp} and \cref{equation: existence of distribution}, we conclude that there exist distributions $P_{\mu,a}^{\pm} \in \mathcal{P}$ such that $\mu_{P_{\mu,a}^+}(x,a) = \mu^{+}(x,a)$ and $\mu_{P_{\mu,a}^-}(x,a) = \mu^{-}(x,a)$ in $(x,a)\in\mathcal{X}\times\{0,1\}$ almost surely.
	\end{proof}
	

	\subsection{Proof of Proposition \ref{proposition: pi_CATE}}
	
	\begin{proof}[Proof of \cref{proposition: pi_CATE}]
		We complete the proof by constructing a data-compatible distribution in $\mathcal{P}$ that simultaneously attains $\mu^{+}(x,1)$ and $\mu^{-}(x,0)$. Similar arguments can be applied to derive the lower bound, which we omit for brevity.
		
		Define $E_{\tau}^{+} = AE_{1}^{+}+(1-A)E_{0}^{-}$. It is straightforward to verify that $E_{\tau}^{+}\in\mathcal{E}$. Applying \cref{proposition: pi_mean} and similar arguments as C.4.2 in \cite{dorn2023sharp}, it is straightforward to construct the distribution $P^{+}_{\tau}\in\mathcal{P}$ such that 
		\[
		\mathbb{E}_{P^{+}_{\tau}}\left[ Y(1) - Y(0) | X=x \right] = \mathbb{E} \left[ \frac{YA}{E_{\tau}^{+}} \Big| X=x  \right] - \mathbb{E} \left[ \frac{Y(1-A)}{ {1-}E_{\tau}^{+}} \Big| X=x  \right] = \mu^{+}(x,1) - \mu^{-}(x,0)
		\]
		in $x\in\mathcal{X}$ almost surely. Since
		\[\begin{aligned}
			\tau^{+}(x) = & \sup_{Q\in\mathcal{P}} \mathbb{E}_{Q}\left[ Y(1) - Y(0) | X=x \right] \\
			\leq & \sup_{Q\in\mathcal{P}} \mathbb{E}_{Q}\left[ Y(1)  | X=x \right] - \inf_{Q\in\mathcal{P}} \mathbb{E}_{Q}\left[  Y(0) | X=x \right] \\
			= & \mu^{+}(x,1) - \mu^{-}(x,0), \\
		\end{aligned}\]
		in $x\in\mathcal{X}$ almost surely. We conclude that 
		$\tau^{+}(x) = \mu^{+}(x,1) - \mu^{-}(x,0)$ in $x\in\mathcal{X}$ almost surely. 
	\end{proof}

	\subsection{Proof of Theorem \ref{theorem: max-min welfare assignment rule} }
	
	\begin{proof}[Proof of \cref{theorem: max-min welfare assignment rule}]
		
		Notice that
		\[\begin{aligned}
			\inf_{ Q \in \mathcal{P} } \mathbb{E}_{Q} [Y(\pi(X))] = & \inf_{ Q \in \mathcal{P} } \left\{ \mathbb{E}\Big[ \pi(X) \mu_{Q}(X,1) \Big] + \mathbb{E}\Big[ (1-\pi(X)) \mu_{Q}(X,0) \Big] \right\} \\
			\geq & \inf_{ Q \in \mathcal{P} } \mathbb{E}\Big[ \pi(X) \mu_{Q}(X,1) \Big] + \inf_{ Q \in \mathcal{P} } \mathbb{E}\Big[ (1-\pi(X)) \mu_{Q}(X,0) \Big]. \\
		\end{aligned}\]
		
		Then we will show that inf and expectation operators are exchangeable. On one hand, for any distribution $Q \in \mathcal{P}$,
		\[
		\mathbb{E}\Big[ \pi(X) \mu_{Q}(X,1) \Big] \geq \mathbb{E}\left[ \pi(X) \inf_{ Q \in \mathcal{P} } \mu_{Q}(X,1) \right].
		\]
		Since $Q$ is arbitrary, the inequality continues to hold after taking the infimum over $\mathcal{P}$ on both sides. For the other side, \cref{proposition: pi_mean} has shown that there exists a distribution $P_{\mu,1}^{-}\in \mathcal{P}$ satisfying $\mu^{-}(x,1)=\mu_{P_{\mu,1}^{-}}(x,1)$ for any $x\in\mathcal{X}$. Thus,
		\[
		\inf_{ Q \in \mathcal{P} } \mathbb{E}\Big[ \pi(X) \mu_{Q}(X,1) \Big] \leq \mathbb{E} \Big[ 
		\pi(X)\mu_{P_{\mu,1}^{-}}(X,1) \Big] = \mathbb{E} \Big[ 
		\pi(X)\mu^{-}(X,1) \Big] = \mathbb{E}\left[ \pi(X) \inf_{ Q \in \mathcal{P} } \mu_{Q}(X,1) \right].
		\]
		Combining the results above, we can conclude that 
		\[
		\inf_{ Q \in \mathcal{P} }\mathbb{E}\Big[ \pi(X) \mu_{Q}(X,1) \Big] = \mathbb{E}\left[ \pi(X) \inf_{ Q \in \mathcal{P} } \mu_{Q}(X,1) \right].
		\]
		Similar arguments can be applied to derive that
		\[
		\inf_{ Q \in \mathcal{P} } \mathbb{E}\Big[ (1-\pi(X)) \mu_{Q}(X,0) \Big] =  \mathbb{E}\left[ (1-\pi(X)) \inf_{ Q \in \mathcal{P} } \mu_{Q}(X,0) \right].
		\]
		
		Define $E^{-} = AE_{1}^{-} + (1-A)E_{0}^{-}$ satisfying $E^{-}\in\mathcal{E}$. Applying \cref{proposition: pi_mean} and similar arguments as C.4.3 in \cite{dorn2023sharp}, we can construct the distribution $P^{-}$ fulfilling
		\[
		\mathbb{E}_{P^{-}}[Y(1)|X=x] = \mathbb{E}\left[ \frac{YA}{E^{-}} \Big| X=x  \right] = \mu^{-}(x,1)
		\]
		and
		\[
		\mathbb{E}_{P^{-}}[Y(0)|X=x] = \mathbb{E}\left[ \frac{Y(1-A)}{ {1-}E^{-}} \Big| X=x  \right] = \mu^{-}(x,0).
		\]
		which imply that
		\[
		\mathbb{E}_{P^{-}}[Y(\pi(X))] = \mathbb{E}\left[  \pi(X) \mu^{-}(X,1) + (1-\pi(X)) \mu^{-}(X,0)  \right].
		\]
		Thus, we conclude that
		\begin{align*}
			& \inf_{ Q \in \mathcal{P} } \mathbb{E}_{Q} [Y(\pi(X))] \leq \mathbb{E}_{P^{-}}[Y(\pi(X))] \\
			= & \mathbb{E}\left[  \pi(X) \mu^{-}(X,1) + (1-\pi(X)) \mu^{-}(X,0)  \right] \\
			= & \mathbb{E}\left[ \pi(X) \inf_{ Q \in \mathcal{P} } \mu_{Q}(X,1) \right] + \mathbb{E}\left[ (1-\pi(X)) \inf_{ Q \in \mathcal{P} } \mu_{Q}(X,0) \right]. 
		\end{align*}
		Results above together indicate that
		\[\begin{aligned}
			\inf_{ Q \in \mathcal{P} } \mathbb{E}_{Q} [Y(\pi(X))] 
			= & \mathbb{E}\left[ \pi(X) \inf_{ Q \in \mathcal{P} } \mu_{Q}(X,1) \right] + \mathbb{E}\left[ (1-\pi(X)) \inf_{ Q \in \mathcal{P} } \mu_{Q}(X,0) \right] \\
			= & \mathbb{E} \Big[ \pi(X) \left( \mu^{-}(X,1) - \mu^{-}(X,0) \right)  \Big] + \mathbb{E}[\mu^{-}(X,0)],
		\end{aligned}\]
		which completes the proof.
	\end{proof}

	\subsection{Proof of Theorem \ref{theorem: CATE-based robust welfare} }
	
	\begin{proof}[Proof of \cref{theorem: CATE-based robust welfare}]
		Applying similar arguments as the proof of \cref{theorem: max-min welfare assignment rule}, we have that
		\begin{equation}\label{appendix equation: CATE PL}
			\inf_{Q \in \mathcal{P}} \mathbb{E} \left[ \tau_{Q}(X) \pi (X) \right]  =   \mathbb{E} \left[ \inf_{Q \in \mathcal{P}} \tau_{Q}(X) \pi (X) \right]  = \mathbb{E} \left[  \tau^-(X) \pi (X) \right].
		\end{equation}
		According to \cref{appendix equation: CATE PL}, the first-best policy is given by $\pi_{\Delta}^\star(x) = \mathds{1}\{ \tau^-(x) > 0 \}$.
	\end{proof}

	\subsection{Proof of  Proposition \ref{proposition: doubly robust score for mean}}
	
	Recall that
	\[
	\psi_{W}\left(z,\pi;\eta_{W} \right) = \sum_{t \in \{0,1\} } \phi_{t}^{-} \left(z; \eta_{W}  \right) \pi(t|x),
	\]
	where $\eta_{W} = \left(e,q^- ,\rho^-_{1}, \rho^-_{0} \right) $,  $\phi_{t}^{-} \left(z;\eta_{W}\right)$ are defined in \cref{equation: phi_t^- definition}, 
	and the functions $\rho_{1}^{\pm}(x,t)$ and $\rho_{0}^{\pm}(x,t)$ are defined in \cref{equation: rho}.

	\begin{proof}[Proof of \cref{proposition: doubly robust score for mean}]
		
		\noindent {\bf Step 1.} The proof of part (1) follows directly from \cref{theorem: max-min welfare assignment rule}, together with the following derivation:
		\begin{equation}\label{appendix equation: conditional mean zero}
			\mathbb{E}\left[ \mathds{1}\{A=t\}\left( \frac{1}{1+\Lambda} - \mathbb{1}\left\{ Y< q^{-}(X,t)\right\} \right) \Big| X \right] =  \mathbb{E}[\mathds{1}\{A=t\}-e_{t}(X)|X] = 0.
		\end{equation}

		\noindent {\bf Step 2.} To show part (2), given the linearity of pathwise derivative, it suffices to show 
		the pathwise derivative of $r \mapsto \mathbb{E}\left[ \psi_{W}(Z,\pi;\eta_{W, \Lambda} + r \bar{\eta}  )  \right] = 0$  for all perturbation directions $\bar{\eta}$ varying only in the components corresponding to  $e,q^-, \rho_1^-$ and $\rho_0^-$, respectively. We only show that 
		\begin{equation}\label{appendix equation: orthogonality}
			\frac{ \mathrm{d} }{\mathrm{d}r}\mathbb{E}\left[ \psi_{W}\left(Z,\pi; {e}, q^{-}+r(\widetilde{q}^{-}-q^{-}) {,\rho^{-}_{1},\rho^{-}_{0}} \right) \right]_{r=0} = 0,
		\end{equation}
		for any $\widetilde{q}^{-}$ belonging to some small neighborhood of $q^-$. The derivations of the pathwise derivatives with respect to the perturbation directions in $ e$, $\rho^{-}_{1}$ 
		and $\rho^{-}_{0}$ are analogous and thus omitted for brevity. Notice that
		\[\begin{aligned}
			& \frac{ \mathrm{d} }{\mathrm{d}r}\mathbb{E}\left[ \psi_{W}\left(Z,\pi; {e,}q^{-}+r  {(\widetilde{q}^{-}-q^{-}),\rho^{-}_{1},\rho^{-}_{0}} \right) \right] \\
			= &  \sum_{t \in  \{0,1\} } \frac{ \mathrm{d} }{\mathrm{d}r}\mathbb{E}\left[ \phi_{t}^{-}\left(Z; {e,}q^{-}+r(\widetilde{q}^{-}-q^{-}) {,\rho^{-}_{1},\rho^{-}_{0}} \right) \pi(t| X) \right] .
		\end{aligned}\]
		To derive \cref{appendix equation: orthogonality}, it suffices to show that
		\[\begin{aligned}
			& \frac{ \mathrm{d} }{\mathrm{d}r}\mathbb{E}\left[ \phi_{ {t}}^{-}\left(Z; {e,}q^{-}+r(\widetilde{q}^{-}-q^{-}) {,\rho^{-}_{1},\rho^{-}_{0}} \right) \pi(t|X) \right]_{r=0}  = 0, \\
		\end{aligned}
		\]
		for all $t \in \{0,1\}$.
		By  {expanding the term for $t=1$ and using} \cref{appendix equation: conditional mean zero}, we have
		\[\begin{aligned}
			& \frac{ \mathrm{d} }{\mathrm{d}r}\mathbb{E}\left[ \phi_{1}^{-}\left(Z;e,q^{-}+r\left(\widetilde{q}^{-}-q^{-}\right),\rho^{-}_{1},\rho^{-}_{0}\right) \pi(X) \right]_{r=0} \\
			= & \frac{ \mathrm{d} }{\mathrm{d}r}\mathbb{E}\left[ YA \left(1+\frac{1- e(X)}{e(X)}\Lambda^{-\text{sgn}\big(Y-q^{-}(X,1)-r\left(\widetilde{q}^{-}(X,1)-q^{-}(X,1)\right)\big)}\right) {\pi(X)}\right]_{r=0} \\
			+ & \frac{ \mathrm{d} }{\mathrm{d}r}\mathbb{E}\left[ q^{-}(X,1) A \frac{1-e(X)}{e(X)}\left( \Lambda - \Lambda^{-1} \right) \left( \frac{1}{1+\Lambda} - \mathbb{1}\left\{ Y< q^{-}(X,1)+r\left(\widetilde{q}^{-}(X,1)-q^{-}(X,1)\right)\right\} \right) {\pi(X)}\right]_{r=0},  \\
		\end{aligned}\]
		 {where $\pi(1|X)=\pi(X)$ by notation.} Let $\mathrm{II}_{A.1}$ and $\mathrm{II}_{A.2}$ denote the two summands in the expression above.  It is not difficult to verify that
		\[\begin{aligned}
			\mathrm{II}_{A.1} & =  \frac{ \mathrm{d} }{\mathrm{d}r}\mathbb{E}\left[ YA \frac{1-e(X)}{e(X)}\left( \Lambda - \Lambda^{-1} \right) \mathbb{1}\left\{ Y< q^{-}(X,1)+r\left(\widetilde{q}^{-}(X,1)-q^{-}(X,1)\right)\right\}  {\pi(X)} \right]_{r=0} \\
			& =  { \frac{ \mathrm{d} }{\mathrm{d}r} \mathbb{E} \left[ (1-e(X)) \left( \Lambda - \Lambda^{-1} \right) \pi(X) \mathbb{E} \left[ Y \mathbb{1}\left\{ Y< q^{-}(X,1)+r\left(\widetilde{q}^{-}(X,1)-q^{-}(X,1)\right)\right\} | X, A=1 \right]  \right]_{r=0}  } \\
			& =  {  \mathbb{E} \left[ (1-e(X)) \left( \Lambda - \Lambda^{-1} \right) \pi(X) \frac{ \mathrm{d} }{\mathrm{d}r} \left[ \int_{-\infty}^{q^{-}(X,1)+r\left(\widetilde{q}^{-}(X,1)-q^{-}(X,1)\right)} y f_{Y}(y|X,1) dy  \right]_{r=0}    \right]}   \\
			& =  \mathbb{E}\left[ q^{-}(X,1)   {(1-e(X))} \left( \Lambda - \Lambda^{-1} \right) f_{Y}\left( q^{-} (X,1) | X,1 \right) \left(\widetilde{q}^{-} (X,1)-q^{-} (X,1)\right)  {\pi(X)} \right] \\
			& =  {-} \mathrm{II}_{A.2},
		\end{aligned}
		\]
		 {where the fourth equality follows from \cref{assum: density}. This} implies that the pathwise derivative with respect to $q^-$ is zero.  {The results for $t=0$ follow from an analogous procedure, thereby completing} the proof of Part (2).
	\end{proof}
	
	\subsection{Proof of  Proposition \ref{proposition: doubly robust score for CATE}}
	\begin{proof}[Proof of \cref{proposition: doubly robust score for CATE}]
		The proof is identical to that of \cref{proposition: doubly robust score for mean} and is omitted here.
	\end{proof}

	\subsection{Proof of Lemma \ref{lemma: nuisance parameter estimation error}}

	\begin{proof}[Proof of \cref{lemma: nuisance parameter estimation error}] We focus on deriving the result for $\widehat{W}_{n}(\pi)-W_{n}(\pi)$; the argument for   $\widehat{\Delta}_{n}(\pi)-\Delta_{n}(\pi)$ is similar and omitted for brevity. Notice that
		\[\begin{aligned}
			\widehat{W}_{n}(\pi)-W_{n}(\pi) = & \frac{1}{n}\sum_{k=1}^{K}\sum_{i\in{\mathcal{I}_{k}}} \pi(X_{i}) \left( \phi_{1}^{-}\left(Z_{i};\widehat{\eta}_{W}^{-k}\right) - \phi_{1}^{-}\left(Z_{i};\eta_{W}\right) \right) \\
			&   {+}  \frac{1}{n}\sum_{k=1}^{K}\sum_{i\in{\mathcal{I}_{k}}} \left( 1-\pi(X_{i}) \right) \left( \phi_{0}^{-}\left(Z_{i};\widehat{\eta}_{W}^{-k}\right) - \phi_{0}^{-}\left(Z_{i};\eta_{W}\right) \right). \\
		\end{aligned}\]
		We bound the first term at the right-hand-side. The analysis of the second term is similar and omitted. To simplify notation, we suppress the superscript $-k$ of the nuisance estimators, e.g., $\widehat{\eta}_{W} = \widehat{\eta}_{W}^{-k}$. By definition,
		\[
		\begin{aligned}
			& \frac{1}{n}\sum_{k=1}^{K}\sum_{i\in{\mathcal{I}_{k}}} \pi(X_{i}) \left( \phi_{1}^{-}\left(Z_{i};\widehat{\eta}_{W}\right) - \phi_{1}^{-}\left(Z_{i};\eta_{W}\right)  \right) \\ 
			= & \frac{1}{n}\sum_{k=1}^{K}\sum_{i\in{\mathcal{I}_{k}}} \pi(X_{i}) Y_{i}A_{i} \left( \frac{1- \widehat{e} (X_{i})}{\widehat{e} (X_{i})}\Lambda^{-\text{sgn}\big(Y_{i}-\widehat{q}^{-} (X_{i},1)\big)} -   \frac{1- e(X_{i})}{e(X_{i})}\Lambda^{-\text{sgn}\big(Y_{i}-q^{-}(X_{i},1)\big)}\right) \\
			+ & \frac{\Lambda - \Lambda^{-1}}{n}\sum_{k=1}^{K}\sum_{i\in{\mathcal{I}_{k}}} \pi(X_{i}) A_{i}\left( \frac{1}{1+\Lambda} - \mathbb{1}\left\{ Y_{i}< \widehat{q}^{-} (X_{i},1)\right\} \right) \left( \widehat{q}^{-} (X_{i},1) \frac{1- \widehat{e} (X_{i})}{\widehat{e} (X_{i})} - q^{-}(X_{i},1) \frac{1- e(X_{i})}{e(X_{i})} \right) \\
			- & \frac{\Lambda - \Lambda^{-1}}{n}\sum_{k=1}^{K}\sum_{i\in{\mathcal{I}_{k}}} \pi(X_{i}) A_{i}q^{-}(X_{i},1) \frac{1- e(X_{i})}{e(X_{i})} \left( \mathbb{1}\left\{ Y_{i}< \widehat{q}^{-} (X_{i},1)\right\} - \mathbb{1}\left\{ Y_{i}< q^{-}(X_{i},1)\right\} \right) \\
			- & \frac{1}{n}\sum_{k=1}^{K}\sum_{i\in{\mathcal{I}_{k}}} \pi(X_{i}) A_{i} \left(\Lambda\widehat{\rho}_{1}^{-}(X_{i},1) + \Lambda^{-1}\widehat{\rho}_{0}^{-}(X_{i},1) \right) \left( \frac{1}{\widehat{e} (X_{i})} - \frac{1}{e(X_{i})}  \right) \\
			- & \frac{1}{n}\sum_{k=1}^{K}\sum_{i\in{\mathcal{I}_{k}}} \pi(X_{i}) \frac{A_{i}-e(X_{i})}{e(X_{i})} 
			\left[ \Lambda \left( \widehat{\rho}_{1}^{-}(X_{i},1)  - \rho_{1}^{-}(X_{i},1) \right)  + \Lambda^{-1} \left( \widehat{\rho}_{0}^{-}(X_{i},1)  - \rho_{0}^{-}(X_{i},1) \right) \right]. \\
		\end{aligned}
		\]
		Denote these five summands by $D_j(\pi)$ for $1 \leq j \leq 5$.
		
		{\it First term.} To bound the first term, it is useful to separate the contributions of each of the $K$ folds:
		\begin{equation*}\label{equation: I_k_pi_1}
			\begin{aligned}
				D_{1,k} (\pi) = & \frac{1}{n} \sum_{i \in \mathcal{I}_{k} } \pi(X_{i}) Y_{i}A_{i} \left( \frac{1- \widehat{e} (X_{i})}{\widehat{e} (X_{i})}\Lambda^{-\text{sgn}\big(Y_{i}-\widehat{q}^{-} (X_{i},1)\big)} -   \frac{1- e(X_{i})}{e(X_{i})}\Lambda^{-\text{sgn}\big(Y_{i}-q^{-}(X_{i},1)\big)}\right) \\
				= & \frac{1}{n} \sum_{i \in \mathcal{I}_{k} } \pi(X_{i}) Y_{i}A_{i} \Lambda^{-\text{sgn}\big(Y_{i}-q^{-}(X_{i},1)\big)} \left( \frac{1}{\widehat{e} (X_{i})} -   \frac{1}{e(X_{i})}\right) \\
				+ & \frac{\Lambda - \Lambda^{-1}}{n} \sum_{i \in \mathcal{I}_{k} } \pi(X_{i}) Y_{i}A_{i}\frac{1- e(X_{i})}{e(X_{i})} \left( \mathbb{1}\left\{ Y_{i}< \widehat{q}^{-} (X_{i},1)\right\} - \mathbb{1}\left\{ Y_{i}< q^{-}(X_{i},1)\right\} \right) \\
				+ & \frac{\Lambda - \Lambda^{-1}}{n} \sum_{i \in \mathcal{I}_{k} } \pi(X_{i}) Y_{i}A_{i} \left( \frac{1}{\widehat{e} (X_{i})} -   \frac{1}{e(X_{i})}\right) \left( \mathbb{1}\left\{ Y_{i}< \widehat{q}^{-} (X_{i},1)\right\} - \mathbb{1}\left\{ Y_{i}< q^{-}(X_{i},1)\right\} \right) \\
				= & D_{1,k}^{(1)} (\pi) + D_{1,k}^{(2)} (\pi) + D_{1,k}^{(3)} (\pi). \\
			\end{aligned}
		\end{equation*}
		We can use the Cauchy-Schwartz inequality to verify that
		\[\begin{aligned}
			\sup_{\pi \in \Pi_n}\left| D_{1,k}^{(3)} (\pi) \right| \leq & \left( \Lambda - \Lambda^{-1} \right)  \sqrt{ \frac{1}{n} \sum_{i \in \mathcal{I}_{k} } Y_{i}^{2} \left( \frac{1}{\widehat{e} (X_{i})} -   \frac{1}{e(X_{i})}\right)^{2} } \\
			&\times\sqrt{ \frac{1}{n} \sum_{i \in \mathcal{I}_{k} }  \left( \mathbb{1}\left\{ Y_{i}< \widehat{q}^{-} (X_{i},1)\right\} - \mathbb{1}\left\{ Y_{i}< q^{-}(X_{i},1)\right\} \right)^{2} } \\
			= & O_{P}\left( \frac{b(n)}{n^{(\zeta_{e}+\zeta_{q})/2}}  \right) =  {o}_{P} \left( n^{-1/2} \right), \\
		\end{aligned}\]
		where the first equality holds by Markov's inequality, \cref{assumption: bounded support} and \cref{assumption: nuisance parameter estimation error}. Then, uniformly over $\pi \in \Pi_n$,
		\[
		D_{1,k} (\pi) = D_{1,k}^{(1)} (\pi) + D_{1,k}^{(2)} (\pi) +  {o}_{P} \left( n^{-1/2} \right).
		\]
		
		{\it Second term.} For bounding the second term,  we still separate out the contributions of the $K$ different folds. After applying similar arguments as the preceding one, we obtain that uniformly over $\pi \in \Pi_n$,
		\begin{equation}\label{decompose of D2k}
			D_{2,k}(\pi) = \widetilde{D}_{2,k}(\pi) +  {o}_{P} \left( n^{-1/2} \right),
		\end{equation}
		where
		\[
		\begin{aligned}
			\widetilde{D}_{2,k}(\pi) 
			=  \frac{\Lambda - \Lambda^{-1}}{n}\sum_{i\in{\mathcal{I}_{k}}}  &\pi(X_{i}) A_{i}\left( \frac{1}{1+\Lambda} - \mathbb{1}\left\{ Y_{i}< q^{-}(X_{i},1)\right\} \right) \\
			& \times \left( \widehat{q}^{-} (X_{i},1) \frac{1- \widehat{e} (X_{i})}{\widehat{e} (X_{i})} - q^{-}(X_{i},1) \frac{1- e(X_{i})}{e(X_{i})} \right) . \\ 
		\end{aligned}
		\]
		Since $\widehat{q}^{-} $ and $\widehat{e} $ were computed using data $\{Z_i\}_{i \in \mathcal{I}_{k}^c}$, one has 
		\[
		\mathbb{E} \left[ \widetilde{D}_{2,k}(\pi) \Big| \widehat{q}^{-} , \widehat{e}   \right] = 0.
		\]
		By \cref{assumption: nuisance parameter estimation error}, we have that uniformly over $i$,
		\[
		\left| \widehat{q}^{-} (X_{i},1) \frac{1- \widehat{e} (X_{i})}{\widehat{e} (X_{i})} - q^{-}(X_{i},1) \frac{1- e(X_{i})}{e(X_{i})}  \right| \leq 1
		\]
		with probability approaching to 1. So the individual summands in $\widetilde{D}_{2,k}(\pi)$ are all $\nu$-sub Gaussian with probability approaching to 1. Define
		\[
		\widehat{V}_{2,k} = \frac{\Lambda}{(1+\Lambda)^{2}} \mathbb{E} \left[ e(X_{i}) \left( \widehat{q}^{-} (X_{i},1) \frac{1- \widehat{e} (X_{i})}{\widehat{e} (X_{i})} - q^{-}(X_{i},1) \frac{1- e(X_{i})}{e(X_{i})} \right)^{2} \Big| \widehat{q}^{-} , \widehat{e}  \right].
		\]
		We can apply Corollary 3 in \cite{athey2021policy} to establish that
		\begin{equation}\label{eq.D2k}
			\begin{aligned}
				\mathbb{E} \left[\sup_{\pi \in \Pi_n} \left| \widetilde{D}_{2,k}(\pi) \right|  \Big | \widehat{q}^{-} , \widehat{e}  \right] 
				= & \mathbb{E} \left[\sup_{\pi \in \Pi_n} \left| \widetilde{D}_{2,k}(\pi) - \mathbb{E} \left[ \widetilde{D}_{2,k}(\pi) \Big| \widehat{q}^{-} , \widehat{e}   \right]  \right|  \Bigg| \widehat{q}^{-} , \widehat{e}  \right]   \\
				= & O \left( \sqrt{ \mathrm{VC}(\Pi_n) \frac{\widehat{V}_{2,k}}{n}  } \right) .
			\end{aligned}
		\end{equation}
		\cref{assumption: nuisance parameter estimation error}, an application of Jensen's inequality and \cref{decompose of D2k,eq.D2k} result that
		\[
		\mathbb{E} \left[\sup_{\pi \in \Pi_n} \left| D_{2}(\pi) \right|   \right] = O \left( \sqrt{ \mathrm{VC}(\Pi_n) \frac{ b(n) }{ n^{1+\min\{ \zeta_{e}, \zeta_{q} \} }  } } \right).
		\]

		{\it Third term.}  Denote
		\[
		D_{3,k}(\pi) = - \frac{\Lambda - \Lambda^{-1}}{n}\sum_{i\in{\mathcal{I}_{k}}} \pi(X_{i}) A_{i}q^{-}(X_{i},1) \frac{1- e(X_{i})}{e(X_{i})} \left( \mathbb{1}\left\{ Y_{i}< \widehat{q}^{-} (X_{i},1)\right\} - \mathbb{1}\left\{ Y_{i}< q^{-}(X_{i},1)\right\} \right).
		\]
		as the $k$-th fold in $D_{3}(\pi)$. Then
		\[\begin{aligned}
			& D_{1,k}^{(2)}(\pi) + D_{3,k}(\pi) \\
			= & \frac{\Lambda - \Lambda^{-1}}{n} \sum_{i \in \mathcal{I}_{k} } \pi(X_{i}) A_{i}\frac{1- e(X_{i})}{e(X_{i})} \Big( Y_{i} - q^{-}(X_{i},1) \Big) \left( \mathbb{1}\left\{ Y_{i}< \widehat{q}^{-} (X_{i},1)\right\} - \mathbb{1}\left\{ Y_{i}< q^{-}(X_{i},1)\right\} \right). \\
		\end{aligned}\]
		Define
		\[
		\widehat{V}_{3,k} = \mathbb{E} \left[ A_{i}  \Big( Y_{i} - q^{-}(X_{i},1) \Big)^{2} \left( \mathbb{1}\left\{ Y_{i}< \widehat{q}^{-} (X_{i},1)\right\} - \mathbb{1}\left\{ Y_{i}< q^{-}(X_{i},1)\right\} \right)^{2} \Big|  \widehat{q}^{-}  \right].
		\]
		By the law of iterated expectation, $\mathbb{E}\left[ \widehat{V}_{3,k} \right]$ can be expressed as
		\[\begin{aligned}
			\mathbb{E}\left[ \widehat{V}_{3,k} \right] = &\mathbb{E} \left[ e(X_{i})  \mathbb{E} \left[  \Big( Y_{i} - q^{-}(X_{i},1) \Big)^{2} \left( \mathbb{1}\left\{ Y_{i}< \widehat{q}^{-} (X_{i},1)\right\} - \mathbb{1}\left\{ Y_{i}< q^{-}(X_{i},1)\right\} \right)^{2} \Big|  X_{i}, A_{i}=1, \widehat{q}^{-}  \right] \right]. \\
		\end{aligned}\]
		Notice that
		\[\begin{aligned}
			& \mathbb{E}\left[ \Big( Y_{i} - q^{-}(X_{i},1) \Big)^{2} \left( \mathbb{1}\left\{ Y_{i}< \widehat{q}^{-} (X_{i},1)\right\} - \mathbb{1}\left\{ Y_{i}< q^{-}(X_{i},1)\right\} \right)^{2} \Big|  X_{i}, 
			A_{i}=1, \widehat{q}^{-}   \right] \\
			= & \int_{\widehat{q}^{-} (X_{i},1)}^{q^{-}(X_{i},1)} \Big( y - q^{-}(X_{i},1) \Big)^{2} dF_{Y_{i}}\left( y \Big|  X_{i}, 
			A_{i}=1, \widehat{q}^{-}  \right) 
			+ \int_{q^{-}(X_{i},1)}^{\widehat{q}^{-} (X_{i},1)} \Big( y - q^{-}(X_{i},1) \Big)^{2} dF_{Y_{i}}\left( y \Big|  X_{i}, 
			A_{i}=1, \widehat{q}^{-}  \right) \\
			= & O_{P} \left( \left( \widehat{q}^{-} (X_{i},1) - q^{-}(X_{i},1) \right)^{3} \right), \\
		\end{aligned}\]
		where the last equality follows by the mean value theorem. Hence, by \cref{assumption: nuisance parameter estimation error},
		\[
		\mathbb{E}\left[ \widehat{V}_{3,k} \right] = O \left( \frac{b(n)}{n^{3\zeta_{q}/2}}\right).
		\]
		After applying a similar argument as the proof of $\widetilde{D}_{2,k}(\pi) $, we have that for any $k$,
		\begin{equation}\label{eq.concentration bound}
			\mathbb{E} \left[\sup_{\pi \in \Pi_n} \left| D_{1,k}^{(2)}(\pi) + D_{3,k}(\pi) - \mathbb{E}\left[ D_{1,k}^{(2)}(\pi) + D_{3,k}(\pi) \Big|  \widehat{q}^{-}   \right] \right|  \Bigg|  \widehat{q}^{-}  \right] = O \left( \sqrt{ \mathrm{VC}(\Pi_n) \frac{ b(n) }{ n^{1+3\zeta_{q}/2 }  } } \right).
		\end{equation}
		By \cref{assumption: strict overlap}, $\pi(X_{i})\big( 1-e(X_{i}) \big)/e(X_{i})$ is bounded uniformly over $i$ and $\pi$.
		Thus,
		\begin{equation}\label{eq.D3k bound}
			\begin{aligned}
				& \mathbb{E}\left[ D_{1,k}^{(2)}(\pi) + D_{3,k}(\pi) \Big|  \widehat{q}^{-}   \right] \\
				\leq & \frac{\Lambda - \Lambda^{-1}}{K} \frac{1-\kappa}{\kappa} \mathbb{E}\left[ A_{i} \Big| Y_{i} - q^{-}(X_{i},1) \Big| \left| \mathbb{1}\left\{ Y_{i}< \widehat{q}^{-} (X_{i},1)\right\} - \mathbb{1}\left\{ Y_{i}< q^{-}(X_{i},1)\right\} \right| \Big|  \widehat{q}^{-}   \right] \\
				= & O \left( \frac{b(n)}{n^{\zeta_{q}}}\right)
			\end{aligned}
		\end{equation}
		uniformly over $\pi\in\Pi_{n}$. Applying Jensen's inequality, triangle inequality, and \cref{eq.concentration bound,eq.D3k bound} yield
		\[
		\mathbb{E} \left[\sup_{\pi \in \Pi_n} \left| D_{1}^{(2)}(\pi) + D_{3}(\pi) \right|   \right] = O \left( \sqrt{ \mathrm{VC}(\Pi_n) \frac{ b(n) }{ n^{1+3\zeta_{q}/2 }  } } + \frac{b(n)}{n^{\zeta_{q}}} \right).
		\]

		{\it Fourth term.} Define the $k$-th fold of $D_{4}(\pi)$ as
		\[
		D_{4,k}(\pi) = -  \frac{1}{n}\sum_{i\in{\mathcal{I}_{k}}} \pi(X_{i}) A_{i} \left(\Lambda\widehat{\rho}_{1}^{-}(X_{i},1) + \Lambda^{-1}\widehat{\rho}_{0}^{-}(X_{i},1) \right) \left( \frac{1}{\widehat{e} (X_{i})} - \frac{1}{e(X_{i})}  \right).
		\]
		Applying similar arguments as the preceding one, we can derive that
		\[
		D_{4,k}(\pi) = \widetilde{D}_{4,k}(\pi) + O_{P} \left( n^{-1/2} \right)
		\]
		uniformly over $\pi\in\Pi_{n}$, where
		\[
		\widetilde{D}_{4,k}(\pi) = -  \frac{1}{n}\sum_{i\in{\mathcal{I}_{k}}} \pi(X_{i}) A_{i} \left(\Lambda\rho_{1}^{-}(X_{i},1) + \Lambda^{-1}\rho_{0}^{-}(X_{i},1) \right) \left( \frac{1}{\widehat{e} (X_{i})} - \frac{1}{e(X_{i})}  \right).
		\]
		Then we consider
		\[\begin{aligned}
			& D_{1,k}^{(1)}(\pi) + \widetilde{D}_{4,k}(\pi) \\
			= & \frac{1}{n}\sum_{i\in{\mathcal{I}_{k}}} \pi(X_{i}) A_{i} \left( Y_{i} \Lambda^{-\text{sgn}\big(Y_{i}-q^{-}(X_{i},1)\big)} - \Lambda\rho_{1}^{-}(X_{i},1) - \Lambda^{-1}\rho_{0}^{-}(X_{i},1) \right) \left( \frac{1}{\widehat{e} (X_{i})} - \frac{1}{e(X_{i})}  \right). \\
		\end{aligned}\]
		It is not difficult to verify that
		\[
		\mathbb{E}\left[ D_{1,k}^{(1)}(\pi) + \widetilde{D}_{4,k}(\pi) \Big|  \widehat{e}   \right] =0,
		\]
		because
		\[
		\mathbb{E}\left[ Y_{i} \Lambda^{-\text{sgn}\big(Y_{i}-q^{-}(X_{i},1)\big)} - \Lambda\rho_{1}^{-}(X_{i},1) - \Lambda^{-1}\rho_{0}^{-}(X_{i},1) \Big|  X_{i}, A_{i}=1, \widehat{e}   \right] =0.
		\]
		Similar arguments are used to show that
		\[
		\mathbb{E} \left[\sup_{\pi \in \Pi_n} \left| D_{1}^{(1)}(\pi) + D_{4}(\pi) \right|   \right] = O \left( \sqrt{ \mathrm{VC}(\Pi_n) \frac{ b(n) }{ n^{1+\zeta_{e} }  } } \right).
		\]

		{\it Fifth term.} Recall that $\mathbb{E}\left[ A_{i} - e(X_{i}) \Big| X_{i}, \widehat{\rho}_{1}^{-}, \widehat{\rho}_{0}^{-} \right] = 0$.
		Applying a similar argument to bound the fourth term, we can similarly bound the fifth term as follows:
		\[
		\mathbb{E} \left[ \sup_{\pi \in \Pi_n} | D_5 (\pi)| \right] = O \left( \sqrt{ \mathrm{VC}(\Pi_n) \frac{ b(n) }{ n^{1+\zeta_{\rho} }  } } \right).
		\]

		{\it Wrapping Up \cref{lemma: nuisance parameter estimation error}. } Combining five terms above with \cref{assumption: policy size}  and \cref{assumption: nuisance parameter estimation error} gives 
		\[
		\begin{aligned}
			& \mathbb{E} \left[ \sup_{\pi \in \Pi_n} \left|\frac{1}{n}\sum_{k=1}^{K}\sum_{i\in{\mathcal{I}_{k}}} \pi(X_{i}) \left( \phi_{1}^{-}\left(Z_{i};\widehat{\eta}_{W}\right) - \phi_{1}^{-}\left(Z_{i};\eta_{W}\right)  \right) \right| \right] \\
			= & O\left( \sqrt{      \frac{ \mathrm{VC}(\Pi_n) 
					b\left( n \right) }{n^{1 + \min \left\{ \zeta_{\rho},\zeta_q, \zeta_e  \right\} }} } \right)  + O\left( \frac{b(n)}{n^{\zeta_{q}}} \right) \\
			= & o  (n^{-1/2}  ).
		\end{aligned}
		\]
	\end{proof}

	\subsection{Proof of Theorem \ref{theorem: asymptotic welfare loss}}
	
	\begin{proof}[Proof of \cref{theorem: asymptotic welfare loss}]
		
		We prove \cref{equation:regret_MMW}; the proof of \cref{equation:regret_MMI} follows by a similar argument and is omitted here for brevity.

		Suppose $\pi_{W,n}  \in  \argmax_{\pi \in \Pi_n} W_{n}( \pi  )$. Otherwise, the proof can be modified by incorporating an $\varepsilon$-approximate optimizer argument. By the definitions of $\pi_{W,n}$ and $\widehat{\pi}_{W,n}$ and \cref{assumption: nuisance parameter estimation error}, it follows that for any $\pi \in \Pi_n$:
		\[
		\begin{aligned}
			W(\pi)- W(\widehat{\pi}_{W,n} ) 
			& =  W(\pi)-W_{n} ( \pi )    +  \underbrace{ W_{n} ( \pi )  -  W_{n} (\pi_{W,n}   ) }_{\leq 0}  +  \underbrace{ W_{n} (\pi_{W,n}  )  - \widehat{W }_{ n} (\pi_{W,n}    )}_{= o_P(n^{-1/2 })}\\
			& + \underbrace{\widehat{W }_{ n} (\pi_{W,n}    ) -    \widehat{W }_{ n} (\widehat{\pi}_{W,n}   ) }_{ \leq 0 }   +  \underbrace{\widehat{W }_{ n} (\widehat{\pi}_{W,n}   )  - W_{n} (\widehat{\pi}_{W,n}   ) }_{ o_P(n^{-1/2 })} +  W_{n} (\widehat{\pi}_{W,n} ) - W(\widehat{\pi}_{W,n} ) \\
			& \leq  W(\pi)-W_{n} ( \pi ) + W_{n} (\widehat{\pi}_{W,n} ) - W(\widehat{\pi}_{W,n} ) + r_n\\
			& \leq  2 \sup_{\pi \in \Pi_n} \left|W(\pi)-W_{n} ( \pi )\right| + r_n
		\end{aligned}
		\]
		where $r_n =  o_P(n^{-1/2})$ and $\sqrt{n}  \mathbb{E}|r_n| = o(1)$. For notational simplicity, we write
		\[
		\psi_{W, \pi}: z \mapsto  \psi_W( z, \pi; \eta_{W} ).
		\]
		Thus, for all $\pi \in \Pi_n$:
		\begin{equation}\label{equation: decomposition of V_theta}
			\begin{aligned}
				W(\pi)-W(\widehat{\pi}_{W,n}  ) 
				& \leq 2 \sup_{\pi \in \Pi_n } \left | W_{n} (\pi) - W(\pi) \right| + r_n\\
				& = 2 \sup _{\pi \in \Pi_n}\left|\left(\mathbb{E}_n-\mathbb{E}\right) \psi_{W, \pi}  \right| + r_n.
			\end{aligned}
		\end{equation}
		Without loss of generality, suppose that there exists $\pi_{W,n}^* \in \Pi_n$  such that  $W (\pi_{W,n}^*) = \max_{\pi \in \Pi_n} W(\pi)$. If no such \(\pi_{W,n}^*\) exists, the proof can be adapted using an \(\varepsilon\)-approximate optimizer, where \(\varepsilon \to 0\). Substituting $\pi_{W,n}^*$  into the preceding expression yields
		\[
		0 \leq  \mathrm{Reg}_W(\widehat{\pi}_{W,n}) =  W(\pi_{W,n}^*) - W( \widehat{\pi}_{W,n}  )   \leq 2 \sup_{\pi \in \Pi_n} \left|\left(\mathbb{E}_n-\mathbb{E}\right) \psi_{W, \pi}  \right|  + r_n.
		\]

		Define $\mathcal{F}_n = \{ \psi_{W, \pi} : \pi \in \Pi_n \}$. To complete the proof, it suffices to show the empirical process term decaying at rate $\sqrt{\mathrm{VC}(\Pi_n)/n }$, i.e.,
		\begin{equation}\label{equation: maximal inequality}
			\sup_{\psi \in \mathcal{F}_n } \left|\left(\mathbb{E}_n-\mathbb{E}\right) \psi  \right| = \sup_{\pi \in \Pi_n} \left|\left(\mathbb{E}_n-\mathbb{E}\right) \psi_{W, \pi}  \right| = O_{P}\left( \sqrt{\mathrm{VC}(\Pi_n)/ n} \right). 
		\end{equation}
		We show \cref{equation: maximal inequality} using  a similar argument in Theorem 2.14.1 in \cite{vaart2023empirical}, while allowing the function class $\mathcal{F}_n$ to vary with $n$.

		{\bf Step 1. Envelope function for $\mathcal{F}_n$. }
		As a first step for upper bounding $\sup_{\psi \in \mathcal{F}_n } \left|\left(\mathbb{E}_n-\mathbb{E}\right) \psi  \right|$, we construct an envelope function for the function class $\mathcal{F}_n$.  Under \cref{assumption: strict overlap}, we have the following bound:
		\[
		\begin{aligned}
			\mathbb{E} |Y|^2 & = \sum_{a\in \{0,1\} } \int  \mathbb{E}\left[ Y^2| X=x, A=a \right] e_a(x) \mathrm{d} F_{X|A}(x|a) \\
			& \geq  \kappa \int  \mathbb{E}\left[ Y^2| X=x, A=a \right]  \mathrm{d} F_{X|A}(x|a),
		\end{aligned}
		\]
		where $F_{X|A}(x|a)$ denote the conditional distribution of $X$ given $A=a$. Now, for any fixed $\Lambda \in [1, \infty)$ and $j, a \in \{0,1\}$, applying Jensen's inequality along with the inequality above yields: 
		\[
		\begin{aligned}
			\mathbb{E} \big| \rho_{j, \Lambda }^- (X, a)  \big|^2 &   \leq \mathbb{E} \left| \mathbb{E}\left[ Y^2 | X, A =a\right] \right|^2  =  \int  \mathbb{E}\left[ Y^2| X=x, A=a \right] \mathrm{d} F_{X|A=a}(x)   \\
			& \leq  \mathbb{E} |Y|^2 /\kappa.
		\end{aligned}
		\]
		As a result, the functions $e(x)$, $q^-(x,a)$, $\rho_{1}^-(x,a)$, and $\rho_{0}^-(x,a)$ are all square-integrable, and hence $\phi_t^-(z; \eta_{W})$ defined in \cref{equation: phi_t^- definition} for $t \in \{0,1\}$ is also square-integrable by Minkowski inequality. The function $F(\cdot) \equiv\left|\phi_0^-(\cdot; \eta_{W}) \right| + \left|\phi_1^-(\cdot; \eta_{W}) \right|$ is an envelop function for $\mathcal{F}_n$ satisfying $|\psi| \leq F$ for all $\psi\in \mathcal{F}_n$ and $\|F\|_{\mathbb{P},2} < \infty$.
		
		{\bf Step 2. Uniform entropy bound. } Second, we derive a uniform upper bound on the covering number of $\mathcal{F}_n$ for all $n$. Lemma A.1 in supplement of \cite{kitagawa2018a} implies that $\mathrm{VC}(\mathcal{F}_n) \leq \mathrm{VC}(\Pi_n)$.  By Theorem 2.6.7 in \cite{vaart2023empirical}, there is a universal $K >0$ such that for all $n$, the following inequality holds
		\[
		\sup_{Q} \log N\left(\epsilon\|F\|_{Q,2}, \mathcal{F}_n, L^2(Q) \right) \leq  K \mathrm{VC}(\mathcal{F}_n) \log \epsilon^{-1} , \quad \forall 0< \epsilon < 1,
		\]
		where the supremum is taken over all discrete probability measures $Q$ with $\| F \|_{Q,2} >0$. As a consequence, the uniform entropy can be upper bounded by:
		\[
		\sup_{Q} \int_0^1 \sqrt{ 1 + \log N\left(\epsilon\|F\|_{Q,2}, \mathcal{F}_n, L^2(Q) \right)    } \mathrm{d} \epsilon \leq 1 +  \sqrt{K} (1 + \sqrt{\pi}/2 ) \sqrt{ \mathrm{VC}(\Pi_n)}.
		\]
		It is important to note that the universal constant $K$ is independent of the function class $\mathcal{F}_n$; see Theorem 2.6.4 and 2.6.7 in \cite{vaart2023empirical} for details.
		
		{\bf Step 3. Bounding the supremum of the symmetrized process. } Given the sample $\{Z_i\}_{i=1}^n$, define the symmetrized empirical process $\left\{ \mathbb{G}_n^o f: f\in \mathcal{F}_n \right\}$ as
		\[
		\mathbb{G}_n^o:  f  \mapsto  \frac{1}{\sqrt{n}} \sum_{i=1}^n \varepsilon_i f(Z_i),
		\]
		where  the $\varepsilon_i$ are independent Rademacher random variables such that $\varepsilon_i = \pm 1$ with probability 1/2 each. Moreover, let $\phi(x) = e^{x^2}-1$ and the conditional Orlicz norm on $\mathcal{F}_n$ is defined as 
		\[
		\left \| \mathbb{G}_n^o f  \right\|_{\phi,n} = \inf \left\{ c> 0: \mathbb{E}\left[ \phi\left(|\mathbb{G}_n^o f|/c\right) | \{Z_i\}_{i=1}^n \right]\leq 1  \right\}.
		\]
		For more details on Orlicz norms, see in Chapter 2.2 in \cite{vaart2023empirical}.
		
		Conditionally on $\{Z_i\}_{i=1}^n$, the process $\mathbb{G}_n^o$ is sub-Gaussian for the $L^2(\mathbb{P}_n)$-seminorm $\|\cdot \|_n$ by Hoeffding's inequality. Formally, for any $f, g \in \mathcal{F}_n$,
		\[
		\begin{aligned}
			\left \| \mathbb{G}_n^o f - \mathbb{G}_n^o g \right\|_{\phi,n} \leq \| f-g \|_n \equiv \sqrt{  \frac{1}{n} \sum_{i=1}^n \left|f(Z_i) - g(Z_i)\right|^2 }.
		\end{aligned}
		\]
		The value $\eta_n = \sup_{f\in \mathcal{F}_n } \|  f\|_n $ is an upper bound for the radius of $\mathcal{F}_n \cup \{0\}$ with respect to this norm.
		The maximal inequality Theorem 2.2.4 in \cite{vaart2023empirical} gives
		\[
		\begin{aligned}
			\left\| \sup_{f \in \mathcal{F}_n}  \left| \mathbb{G}_n^o f \right|  \right\|_{\phi, n} & \leq K_{\phi} \int_0^{\eta_n} \sqrt{ 1 + \log N \left(\epsilon, \mathcal{F}_n, L^2(\mathbb{P}_n) \right) } \mathrm{d} \epsilon\\
			& \leq K_{\phi} \| F \|_n  \int_0^{1} \sqrt{ 1 + \log N \left(\epsilon \| F \|_n, \mathcal{F}_n, L^2(\mathbb{P}_n) \right) } \mathrm{d} \epsilon.
		\end{aligned}
		\]
		where $K_{\phi}$ is a universal constant only depending on the function $\phi$; see Theorem 2.2.4 in \cite{vaart2023empirical} for more details.
		By Problem 2.2.5 in \cite{vaart2023empirical}, we have 
		\[
		\mathbb{E} \left[ \sup_{f \in \mathcal{F}_n}  \left| \mathbb{G}_n^o f \right|^2 \Big| \{Z_i\}_{i=1}^n \right]  \leq  4 \log 2 \left\| \sup_{f \in \mathcal{F}_n}  \left| \mathbb{G}_n^o f \right|  \right\|_{\phi, n}^2. 
		\]
		Consequently, applying Jensen's inequality gives that there is a constant $K>0$ not depending on $n$ such that for all $n \in \mathbb{N}^+$,
		\[
		\begin{aligned}
			\mathbb{E} \left[ \sup_{f \in \mathcal{F}_n}  \left| \mathbb{G}_n^o f \right| \Big| \{Z_i\}_{i=1}^n \right] & \leq   K \| F \|_{n}  \int_0^{1} \sqrt{ 1 + \log N \left(\epsilon  \| F \|_{n} , \mathcal{F}_n, L^2(\mathbb{P}_n) \right) } \mathrm{d} \epsilon\\
			& \leq  K   \| F \|_{n} \sup_{Q} \int_0^1 \sqrt{ 1 + \log N \left(\epsilon  \| F \|_{Q,2} , \mathcal{F}_n, L^2(Q) \right)    } \mathrm{d} \epsilon\\
			& \leq  K  \| F \|_{n} \sqrt{\mathrm{VC}(\Pi_n) },
		\end{aligned}
		\]
		where $K>0$ is a constant independent of $n$. Taking expectation on both hand sides and applying  Lemma 2.3.1 in \cite{vaart2023empirical} gives 
		\[
		\begin{aligned}
			\mathbb{E}\left[  \sup_{\psi \in \mathcal{F}_n } \left|\left(\mathbb{E}_n-\mathbb{E}\right) \psi  \right|  \right] & \leq 2 n^{-1/2}  \mathbb{E} \left[ \sup_{f \in \mathcal{F}_n}  \left| \mathbb{G}_n^o f \right|  \right]  \\
			& \lesssim \sqrt{\mathrm{VC}(\Pi_n)/n }  \mathbb{E} \sqrt{ \frac{1}{n}  \sum_{i=1}^n \left| F(Z_i) \right|^2  } \\
			& \leq  \| F \|_{\mathbb{P},2} \sqrt{\mathrm{VC}(\Pi_n)/n },
		\end{aligned}
		\]
		where the last step follows from Jensen's inequality. The desired result follows.
	\end{proof}

	\section{ {Proofs for the Results in \cref{appendix: self-selection,appendix: baseline policy and mmr}}}\label{section:Proof for appendices}
	
	\subsection{Proof of Theorem \ref{thm:AMMW}}
	
	\begin{proof}[Proof of \cref{thm:AMMW}]
		By definition of $Y(\pi^{\dagger})$, the worst-case welfare $W(\pi^{\dagger})$ can be decomposed as follows:
		\begin{align*}
			W(\pi^{\dagger}) = & \inf_{ Q \in \mathcal{P} } \mathbb{E}_Q \left[ Y(\pi^{\dagger}) \right] \\
			= & \inf_{ Q \in \mathcal{P} } \mathbb{E}_Q \left[ \mathds{1}\{ \pi^{\dagger}(X) = 1 \} Y(1) + \mathds{1}\{ \pi^{\dagger}(X) = 0 \} Y(0) + \mathds{1}\{ \pi^{\dagger}(X) = \mathrm{S} \} Y
			\right] \\
			= & \inf_{ Q \in \mathcal{P} } \mathbb{E} \left[ \mathds{1}\{ \pi^{\dagger}(X) = 1 \} \mu_Q (X,1) + \mathds{1}\{ \pi^{\dagger}(X) = 0 \} \mu_Q (X,0)
			\right]  + \mathbb{E} \left[ \mathds{1}\{ \pi^{\dagger}(X) = \mathrm{S} \} \mathbb{E}[Y|X] \right],
		\end{align*}
		where the third equality holds because the counterfactual outcome under the self-selection arm $\mathrm{S}$ is observed, whose distribution is identified from the data and thus invariant to $Q$. Applying similar arguments as the proof of \cref{theorem: max-min welfare assignment rule}, we have that
		\begin{align*}
			W(\pi^{\dagger}) = &  \mathbb{E} \left[ \mathds{1}\{ \pi^{\dagger}(X) = 1 \} \inf_{ Q \in \mathcal{P} } \mu_Q (X,1) + \mathds{1}\{ \pi^{\dagger}(X) = 0 \} \inf_{ Q \in \mathcal{P} }\mu_Q (X,0) +  \mathds{1}\{ \pi^{\dagger}(X) = \mathrm{S} \} \mathbb{E}[Y|X]
			\right] \\
			= & \mathbb{E} \left[ \mathds{1}\{ \pi^{\dagger}(X) = 1 \} \mu^{-} (X,1) + \mathds{1}\{ \pi^{\dagger}(X) = 0 \} \mu^{-} (X,0) + \mathds{1}\{ \pi^{\dagger}(X) = \mathrm{S} \} \mathbb{E}[Y|X]  \right], 
		\end{align*}
		which completes the proof.
	\end{proof}

	\subsection{Proof of Theorem \ref{thm:AMMI}}
	\begin{proof}[Proof of \cref{thm:AMMI}]
		Notice that $\Delta( \pi^{\dagger})$ can be decomposed as
		\begin{align*}
			\Delta( \pi^{\dagger}) = & \inf_{ Q \in \mathcal{P} }  \mathbb{E}_Q  \Big[ \mathds{1}\{ \pi^{\dagger}(X) = 1 \} Y(1) + \left( \mathds{1}\{ \pi^{\dagger}(X) = 0 \} - 1 \right) Y(0) + \mathds{1}\{ \pi^{\dagger}(X) = \mathrm{S} \}  Y \Big] \\
			= & \inf_{ Q \in \mathcal{P} }  \mathbb{E} \Big[ \mathds{1}\{ \pi^{\dagger}(X) = 1 \} \mu_Q (X,1) - \left( 1 - \mathds{1}\{ \pi^{\dagger}(X) = 0 \} \right) \mu_Q (X,0) + \mathds{1}\{ \pi^{\dagger}(X) = \mathrm{S} \}  \mathbb{E}[Y|X] \Big].
		\end{align*}
		Applying similar arguments as the proof of \cref{theorem: max-min welfare assignment rule}, we obtain
		\begin{align*}
			\Delta(\pi^{\dagger}) 
			& = \mathbb{E} \left[ \mathds{1}\{ \pi^{\dagger}(X) = 1 \} \inf_{ Q \in \mathcal{P} }\mu_Q (X,1) - \left( 1 - \mathds{1}\{ \pi^{\dagger}(X) = 0 \} \right) \sup_{ Q \in \mathcal{P} }\mu_Q (X,0) + \mathds{1}\{ \pi^{\dagger}(X) = \mathrm{S} \}  \mathbb{E}[Y|X] \right] \\
			& = \mathbb{E} \left[ \mathds{1}\{ \pi^{\dagger}(X) = 1 \} \mu^{-} (X,1) +  \mathds{1}\{ \pi^{\dagger}(X) = 0 \}  \mu^{+} (X,0) + \mathds{1}\{ \pi^{\dagger}(X) = \mathrm{S} \}  \mathbb{E}[Y|X]  \right] - \mathbb{E}[\mu^{+} (X,0)].
		\end{align*}
		This completes the proof.
	\end{proof}
	
	\subsection{Proof of \cref{thm:MMI-baseline}}
	
	\begin{proof}[Proof of \cref{thm:MMI-baseline}]
		Consider the following derivation:
		\begin{align*}
			& \quad\  \mathbb{E}_Q\left[ Y(\pi(X)) - Y\left(\pi_0(X) \right) \right] \\
			& =  \mathbb{E}_{Q} \left[ (Y(1)-Y(0))  \left( \pi (X) - \pi_{0} (X)  \right) \right] = \mathbb{E} \left[ \tau_{Q}(X)  \left( \pi (X) - \pi_{0} (X)  \right) \right] \\
			& = \mathbb{E} \left[ \tau_{Q}(X)  \max\{\pi (X) - \pi_{0} (X),0\} \right] - \mathbb{E} \left[ \tau_{Q}(X)  \max\{\pi_{0} (X) - \pi (X),0\} \right]. 
		\end{align*}
		Since the regions $\{x: \pi(x) > \pi_{0}(x)\}$ and $\{x: \pi(x) < \pi_{0}(x)\}$ are disjoint, after applying similar arguments as the proof of \cref{theorem: max-min welfare assignment rule}, we have
		\begin{align*}
			& \inf_{ Q \in \mathcal{P} } \left\{\mathbb{E} \left[ \tau_{Q}(X)  \max\{\pi (X) - \pi_{0} (X),0\} \right] - \mathbb{E} \left[ \tau_{Q}(X)  \max\{\pi_{0} (X) - \pi (X),0\} \right]\right\} \\
			= &  \inf_{ Q \in \mathcal{P} }\mathbb{E} \left[ \tau_{Q}(X)  \max\{\pi (X) - \pi_{0} (X),0\} \right] - \sup_{ Q \in \mathcal{P} }\mathbb{E} \left[ \tau_{Q}(X)  \max\{\pi_{0} (X) - \pi (X),0\} \right] \\
			= & \ \mathbb{E} \left [ \tau^{-}(X) \max \left\{ \pi(X)  -  \pi_{0} (X), 0 \right\} \right] - \mathbb{E} \left [ \tau^{+}(X) \max \left\{ \pi_{0}(X) - \pi(X), 0 \right\} \right], 
		\end{align*}
		which completes the proof.
	\end{proof}

	\subsection{Proof of Theorem \ref{theorem: minimax regret policy}}
	
	\begin{proof}[Proof of \cref{theorem: minimax regret policy}]
		We start the proof with the following derivation:
		\[
		\begin{aligned}
			\inf_{\pi \in \Pi} \sup_{Q \in \mathcal{P} }   \mathrm{Reg}_Q( \pi )  &  = \inf_{\pi \in \Pi} \sup_{Q \in \mathcal{P} }   \sup_{\pi^\prime \in \Pi} \mathbb{E}_Q [Y(\pi^\prime(X))] -  \mathbb{E}_Q [Y(\pi(X))]  \\
			&  = \inf_{\pi \in \Pi}  \sup_{\pi^\prime \in \Pi} \sup_{Q \in \mathcal{P} }   \mathbb{E}_Q \left [ \left(Y(1) - Y(0) \right) \left(  \pi^\prime(X)  -  \pi (X) \right)   \right]  \\
			& =  \inf_{\pi \in \Pi}  \sup_{\pi^\prime \in \Pi} \sup_{Q \in \mathcal{P} }   \mathbb{E} \left [ \tau_{Q}(X) \left(  \pi^\prime(X)  -  \pi (X) \right)   \right] . \\
		\end{aligned}
		\]
		Notice that
		\begin{align*}
			& \mathbb{E} \left [ \tau_{Q}(X) \left(  \pi'(X)  -  \pi (X) \right) \right] \\
			= & \mathbb{E} \left [ \tau_{Q}(X) \max \left\{ \pi'(X) - \pi(X), 0 \right\} \right] - \mathbb{E} \left [ \tau_{Q}(X) \max \left\{ \pi(X) - \pi'(X), 0 \right\} \right], 
		\end{align*}
		and the fact that $ \{x: \pi(x)  >  \pi^\prime (x) \}$ and $ \{x: \pi(x)  <  \pi^\prime (x) \}$ are disjoint. Applying similar arguments as the proof of \cref{theorem: max-min welfare assignment rule}, we have
		\begin{equation*}
			\begin{aligned}
				& \sup_{ Q \in \mathcal{P} } \mathbb{E} \left [ \tau_{Q}(X) \left(  \pi'(X)  -  \pi (X) \right) \right] \\
				= &\ \mathbb{E} \left [ \tau^{+}(X) \max \left\{ \pi'(X) - \pi(X), 0 \right\} \right] - \mathbb{E} \left [ \tau^{-}(X) \max \left\{ \pi(X) - \pi'(X), 0 \right\} \right], 
			\end{aligned}
		\end{equation*}
		which completes the proof of \cref{equation: minimax regret expression - minimax formula}.
		
		Then we turn to derive the first-best policy. When $\Pi = \Pi_o$ consists of all measurable policies, it suffices to consider the following optimizations for any $x \in \mathcal{X}$:
		\begin{equation}\label{equation: piece-wise optimization}
			\sup_{ \pi^{\prime}(x) \in [0,1]  }      \tau^+(x)       \max \left\{ \pi'(x) - \pi(x), 0 \right\}   -     \tau^-(x)       \max \left\{ \pi(x) - \pi'(x), 0 \right\}.
		\end{equation}
		It is straightforward to show that
		\[
		\tilde{\pi}'(x) =  \mathds{1}\{ \left(  1 - \pi(x) \right) \tau^{+} (x)   + \pi(x) \tau^{-} (x) > 0 \}
		\]
		is an optimal solution to \cref{equation: piece-wise optimization}. Substituting $\tilde{\pi}^\prime(\cdot)$ into \cref{equation: piece-wise optimization}, we find that the first-best policy minimizing
		\[
		\tau^+(x)       \max \left\{ \tilde{\pi}'(x) - \pi(x), 0 \right\}   -     \tau^-(x)       \max \left\{ \pi(x) - \tilde{\pi}'(x), 0 \right\}
		\]
		is given by
		\[
		\pi^\star_{R} (x) 
		= \begin{cases}
			1, & \text{if}\ \   \tau^-(x) \geq 0, \\
			\tau^+(x)/(\tau^+(x) - \tau^-(x)), & \text{if}\ \  \tau^-(x) < 0 < \tau^+(x) \\
			0, & \text{if}\ \  \tau^+(x) \leq 0, \\
		\end{cases}.  
		\]
		
	\end{proof}
	


	\section{Proofs  for the Identification under Multi-Valued Treatment}\label{section:Multi_Identification}

	Before presenting the formal proofs of \cref{proposition: sharp bounds for conditional mean with discrete treatment,theorem: max-min welfare under discrete treatment}, we first establish a more general result that extends Proposition 1B of \cite{dorn2023sharp} to the multi-valued treatment setting. As shown below, \cref{proposition: sharp bounds for conditional mean with discrete treatment} follows as a direct corollary of \cref{proposition: existence of distribution with discrete treatment}.
	
	\begin{lemma}\label{proposition: existence of distribution with discrete treatment}
		For any $a\in\mathcal{A}$ and any random variable $E\in(0,1)$ satisfying
		\[
		\mathbb{E} \left[ \mathds{1}\{A=a\} / E | X=x  \right] =1,
		\]
		and 
		\begin{equation} \label{assumption: MSM bound under discrete treatment}
			\frac{\left( e_{a}(X) + \left[ 1-e_{a}(X) \right] / \Lambda \right) \mathds{1}\{A=a\}}{e_{a}(X)} \leq \frac{\mathds{1}\{A=a\}}{E} \leq \frac{ \left( e_{a}(X) + \left[ 1-e_{a}(X) \right]  \Lambda \right)\mathds{1}\{A=a\}}{e_{a}(X)},
		\end{equation}
		we can construct random variables $\left(X,Y(a_{1}), Y(a_{2}),\dots, Y(a_{d}), A, U\right)$ defined on the same probability space as $(X,Y,A,E)$ and an associated putative propensity score $e_{a} (x,u) = \mathbb{P} [ A=a | X=x, U=u ]$ such that
		\begin{enumerate}
			\item[(1)] $ Y = \sum_{j=1}^{d}\mathds{1}\{A=a_{j}\}Y(a_{j})$.
			\item[(2)] $\left(Y(a_{1}), Y(a_{2}),\dots, Y(a_{d}) \right) \indep A \mid (X,U)$ and 
			\[
			e_{a}(X)/\left( e_{a}(X) + \left[ 1-e_{a}(X) \right] \Lambda \right) \leq  e_{a} (X,U) \leq  e_{a}(X)/\left( e_{a}(X) + \left[ 1-e_{a}(X) \right] / \Lambda \right).
			\]
			\item[(3)] $\mathds{1}\{A=a\}/e_{a} (X,U) = \mathds{1}\{A=a\}/E$.
		\end{enumerate}
	\end{lemma}


	\begin{proof}[Proof of \cref{proposition: existence of distribution with discrete treatment}]
		Define the following conditional distribution functions:
		\[
		F(y|x,a) = \mathbb{P}(Y\leq{y}|X=x,A=a),
		\]
		\[
		G(y|x,a,e) = \mathbb{P}(Y\leq{y}|X=x,A=a,E=e),
		\]
		\[
		H(e|x,a) = \mathbb{P} (E\leq e|X=x,A=a),
		\]
		\[
		K_{a}(u|x) = \int_{-\infty}^{u} \frac{e_{a}(x)}{1-e_{a}(x)} \frac{1-e}{e} \text{d} H(e|x,a).
		\]
		Given $\mathbb{E} \left[ \mathds{1}\{A=a\} / E \mid X=x \right] =1$, $K_{a}(u|x)$ is a valid CDF for any $x\in\mathcal{X}$ and $a\in\mathcal{A}$.
		For simplicity, we provide the construction for $a=a_{1}$; identical arguments apply to other treatments. Let $V, V_{a_{1}},V_{a_{2}},\dots,V_{a_{d}}$ be i.i.d. $\text{Uniform}[0,1]$ random variables, independent of $(X,Y,A,E)$. We  construct $Y(a_{1}), Y(a_{2}),\dots, Y(a_{d})$ and $U$ as:
		\[
		Y(a_{1}) = \mathds{1} \{ A=a_{1} \} Y + \mathds{1} \{ A \neq a_{1} \} G^{-1} (V_{a_{1}}|X,a_{1},U),
		\]
		\[
		Y(a_{j}) =  \mathds{1} \{ A \neq a_{j} \} F^{-1} (V_{a_{j}}|X,a_{j}) + \mathds{1} \{ A = a_{j} \} Y \quad\text{for}\quad j\in\{ 2,3,\dots,d \},
		\]
		\[
		U = \mathds{1} \{ A=a_{1} \} E + \mathds{1} \{ A \neq a_{1} \} K^{-1}_{a_{1}} (V|X).
		\]
		We verify that the constructed random variables satisfy the three required properties. Property (1) follows directly from the construction. We now proceed to examine the remaining properties. To verify property (2), we consider the following derivation. Conditional on $X, U, A=a_{1}$, we have
		\begin{align*}
			& \mathbb{P} ( Y(a_{1}) \leq y_{1}, Y(a_{2}) \leq y_{2}, \dots, Y(a_{d}) \leq y_{d} | X, U, A=a_{1}  ) \\
			=  & \mathbb{P} \left ( Y \leq y_{1}, F^{-1} (V_{a_{2}}|X,a_{2}) \leq y_{2}, \dots, F^{-1} (V_{a_{d}}|X,a_{d}) \leq y_{d} | X, U, A=a_{1} \right ) \\
			= & \mathbb{P} \left ( Y \leq y_{1} | X, U, A=a_{1} \right ) \prod\limits_{j=2}^{d} \mathbb{P} \left (  F^{-1} (V_{a_{j}}|X,a_{j}) \leq y_{j} | X, U, A=a_{1} \right ) \\
			= & G(y_{1}|X,a_{1},U) \prod\limits_{j=2}^{d} F (y_{j}|X,a_{j}).
		\end{align*}
		Conditional on $X, U, A=a_{j}$ for $j\neq 1$, we have
		\begin{align*}
			& \mathbb{P} ( Y(a_{1}) \leq y_{1}, Y(a_{2}) \leq y_{2}, \dots, Y(a_{d}) \leq y_{d} | X, U, A=a_{j}  ) \\
			=  & \mathbb{P} \left ( G^{-1} (V_{a_{1}}|X,a_{1},U) \leq y_{1}, F^{-1} (V_{a_{2}}|X,a_{2}) \leq y_{2}, \dots, Y \leq y_{j}, \dots , F^{-1} (V_{a_{d}}|X,a_{d}) \leq y_{d} | X, U, A=a_{j} \right ) \\
			= & \mathbb{P} \left ( G^{-1} (V_{a_{1}}|X,a_{1},U) \leq y_{1} | X, U, A=a_{j} \right ) \prod\limits_{\ell=2,\ell\neq{j}}^{d} \mathbb{P} \left (  F^{-1} (V_{a_{\ell}}|X,a_{\ell}) \leq y_{\ell} | X, U, A=a_{\ell} \right ) \\
			&\times \mathbb{P} \left ( Y \leq y_{j} | X, U, A=a_{j} \right ) \\
			= & G(y_{1}|X,a_{1},U) \prod\limits_{\ell=2,\ell\neq{j}}^{d} F (y_{\ell}|X,a_{\ell}) \times \mathbb{P} \left ( Y \leq y_{j} | X, K^{-1} (V|X), A=a_{j} \right ) \\
			= & G(y_{1}|X,a_{1},U) \prod\limits_{j=2}^{d} F (y_{j}|X,a_{j}),
		\end{align*}
		Combining the two expressions above gives       $\left(Y(a_{1}), Y(a_{2}),\dots, Y(a_{d}) \right) \indep A \mid (X,U)$. By Bayes' rule,
		\begin{align*}
			e_{a_{1}}(x,u) = & e_{a_{1}}(x) \frac{\text{d} \mathbb{P}(u|X=x,A=a_{1})}{\text{d} \mathbb{P}(u|X=x)} \\
			= & e_{a_{1}}(x) \frac{\text{d} \mathbb{P}(u|X=x,A=a_{1}) / \text{d}H(u|x,a_{1})}{\text{d} \mathbb{P}(u|X=x) / \text{d}H(u|x,a_{1})} \\
			= & \frac{e_{a_{1}}(x)}{e_{a_{1}}(x) + \frac{e_{a_{1}}(x)}{1-e_{a_{1}}(x)} \frac{1-u}{u} \left( 1-e_{a_{1}}(x) \right)} \\
			= & u.
		\end{align*}
		Since the support of $K_{a_{1}}(\cdot|x)$ is contained in that of $H(\cdot|x,a_{1})$, \cref{assumption: MSM bound under discrete treatment} implies
		\[
		e_{a_{1}}(X)/\left( e_{a_{1}}(X) + \left[ 1-e_{a_{1}}(X) \right] \Lambda \right) \leq   U \leq  e_{a_{1}}(X)/\left( e_{a_{1}}(X) + \left[ 1-e_{a_{1}}(X) \right] / \Lambda \right) \quad\text{a.s}.
		\]
		Finally, property (3) is also easy to verify. The event $\mathds{1} \{ A=a_{1} \}$ implies $U=E$, so 
		\[
		\mathds{1}\{A=a_{1}\}/e_{a_{1}} (X,U) = \mathds{1}\{A=a_{1}\}/U = \mathds{1}\{A=a_{1}\}/E.
		\]
		
	\end{proof}

	\begin{proof}[Proof of \cref{proposition: sharp bounds for conditional mean with discrete treatment}]
		The proof follows directly from that of \cref{proposition: existence of distribution with discrete treatment}, using a similar argument as in  \cref{proposition: pi_mean}, and is therefore omitted for brevity.
	\end{proof}
	
	We now turn to proof \cref{theorem: max-min welfare under discrete treatment}.
	
	\begin{proof}[Proof of \cref{theorem: max-min welfare under discrete treatment}]
		According to \cref{proposition: sharp bounds for conditional mean with discrete treatment}, for any $Q\in \mathcal{P}_{\mathrm{M}}$ and $(x,a)\in\mathcal{X}\times\mathcal{A}$, we have 
		\[
		\mu_{Q}(x,a)\geq \inf_{ Q \in \mathcal{P}_{\mathrm{M}} } \mu_{Q}(x,a) =\mu^{-}(x,a).
		\]
		It follows that
		\[
		\mathbb{E} \left[ \sum_{a\in\mathcal{A}} {\mu}_{Q}(X,a) \pi_{a}(X)   \right] \geq \mathbb{E} \left[ \sum_{a\in\mathcal{A}} {\mu}^{-}(X,a) \pi_{a}(X)   \right]
		\]
		Since $ Q \in \mathcal{P}_{\mathrm{M} }$ is arbitrary, the inequality remains valid after taking the infimum over $Q\in\mathcal{P}_{\mathrm{M}}$ on both sides. Therefore, we conclude that
		\[
		\begin{aligned}
			W(\pi) &=  \inf_{ Q \in \mathcal{P}_{\mathrm{M} } } \mathbb{E}_Q \left[ \sum_{a\in \mathcal{A}} Y(a) \pi_{a}(X) \right] = \inf_{ Q \in \mathcal{P}_{\mathrm{M} } } \mathbb{E} \left[ \sum_{a\in \mathcal{A}} \mu_{Q}(X,a) \pi_{a}(X)   \right] \\
			& \geq \mathbb{E} \left[ \sum_{a\in \mathcal{A}}\mu^{-}(X,a) \pi_{a}(X)   \right]. \\
		\end{aligned}
		\]
	\end{proof}

\end{appendices}

\end{document}